\documentclass[a4paper,numberwithinsect, cleveref, autoref, thm-restate]{lipics-v2021}
\nolinenumbers
\hideLIPIcs
\usepackage{subcaption}
\usepackage[colorinlistoftodos]{todonotes}
\usepackage{tikz}
\usetikzlibrary{calc,positioning,fit,backgrounds,decorations.pathreplacing,arrows.meta,shapes.geometric,shapes.misc,intersections,shadows}
\usepackage{complexity}
\usepackage{mathtools}
\usepackage{pdflscape}

\newif\ifshortversion
\providecommand{\shortversionbuild}{false}
\expandafter\ifx\csname shortversion\shortversionbuild\endcsname\relax
  \shortversiontrue
\else
  \csname shortversion\shortversionbuild\endcsname
\fi
\long\def\shortonly#1{%
  \ifshortversion
    #1%
  \fi
}
\long\def\longonly#1{%
  \ifshortversion
  \else
    #1%
  \fi
}
\newcommand{\versionparagraph}[1]{%
  \ifshortversion
    \subparagraph*{#1}%
  \else
    \paragraph*{#1}%
  \fi
}

\newcommand{\yes}{\emph{yes}}

\newcommand{\dotcup}{\mathbin{\dot{\cup}}}

\DeclareMathOperator{\sgn}{sgn}
\DeclareMathOperator{\Pf}{Pf}

\DeclareMathOperator{\vc}{vc}

\DeclareMathOperator{\nd}{nd}
\DeclareMathOperator{\tc}{tc}
\DeclareMathOperator{\ml}{ml}
\DeclareMathOperator{\vi}{vi}

\DeclareMathOperator{\cvd}{cvd}

\newcommand{\prob}[1]{\textnormal{\textsc{#1}}}

\newcommand{\F}{\mathbb{F}}

\newcommand{\problemfield}[2]{%
    \par\noindent\hangindent=7em\hangafter=1
    \makebox[7em][l]{\normalsize\textbf{#1}}#2%
}
\newcommand{\boxproblem}[4]{
    \par\begin{center}
        \fbox{\parbox{.97\textwidth}{
            \normalsize\textsc{#1}\par\smallskip
            \problemfield{Input:}{\normalsize#2}
            \problemfield{Question:}{\normalsize#3}
            \problemfield{Parameter:}{\normalsize#4}
        }}
    \end{center}\par
}
\newcommand{\boxproblemnoparameter}[3]{
    \par\begin{center}
        \fbox{\parbox{.97\textwidth}{
            \normalsize\textsc{#1}\par\smallskip
            \problemfield{Input:}{\normalsize#2}
            \problemfield{Question:}{\normalsize#3}
        }}
    \end{center}\par
}

\author{Tomohiro Koana}{Graduate School of Information Science and Technology, The University of Tokyo, Japan}{tomohiro.koana@gmail.com}{https://orcid.org/0000-0002-8684-0611}{Supported in part by JST CREST Grant Number JPMJCR24Q2 and JST ERATO Grant Number JPMJER2301.}
\author{Soh Kumabe}{CyberAgent, Japan}{kumabe_soh@cyberagent.co.jp}{https://orcid.org/0000-0002-1021-8922}{}
\author{Yota Otachi}{Nagoya University, Japan}{otachi@nagoya-u.jp}{https://orcid.org/0000-0002-0087-853X}{Partially supported by JSPS KAKENHI Grant Numbers JP22H00513, JP24H00697, JP25K03076, JP25K03077.}
\authorrunning{T. Koana, S. Kumabe, and Y. Otachi}

\title{Complexity of induced subgraph isomorphism and maximum common induced subgraph parameterized by cluster vertex deletion number}
\titlerunning{Complexity of ISI and MCIS parameterized by CVD}

\Copyright{Tomohiro Koana, Soh Kumabe, and Yota Otachi}



\EventEditors{John Q. Open and Joan R. Access}
\EventNoEds{2}
\EventLongTitle{42nd Conference on Very Important Topics (CVIT 2016)}
\EventShortTitle{CVIT 2016}
\EventAcronym{CVIT}
\EventYear{2016}
\EventDate{December 24--27, 2016}
\EventLocation{Little Whinging, United Kingdom}
\EventLogo{}
\SeriesVolume{42}
\ArticleNo{23}
\begin{document}
\maketitle

\begin{abstract}
We study the parameterized complexity of \textsc{Induced Subgraph Isomorphism (ISI)} and \textsc{Maximum Common Induced Subgraph (MCIS)} with respect to the cluster vertex deletion number $k$. For ISI, we give a randomized $O^*(k^{O(k)})$-time algorithm, showing that ISI is fixed-parameter tractable under this parameter and resolving an open question of Hanaka et al.~[WALCOM 2026]. Our algorithm is optimal under the Exponential Time Hypothesis (ETH), and is based on a reduction to \textsc{Exact Multicolored Matching} solvable via algebraic techniques. For MCIS, we present a randomized $O^*(2^{O(k^2)})$-time algorithm via a reduction to a weighted variant of \textsc{Exact Multicolored Matching}, and we prove a matching ETH-based lower bound by showing that a $k$-by-$k$ binary matrix feasibility problem with list-constrained rows and columns admits no $O^*(2^{o(k^2)})$-time algorithm, which may be of independent interest. These results reveal that, in this setting, MCIS is strictly harder than ISI. Finally, for the three-graph variant 3-MCIS, we show that it becomes NP-hard already when each input graph has cluster vertex deletion number~2.
\end{abstract}

\section{Introduction}

Graph pattern matching captures a basic algorithmic task: detect a prescribed pattern in a host graph. A closely related task is to compare two graphs by the amount of structure they can share. In this paper we focus on the induced setting, where a match must respect both adjacency and non-adjacency.
% This yields two canonical formulations: an induced containment question, and a maximum common induced substructure question whose optimum value is often used (possibly after normalization) as a graph similarity score~\cite{BunkeS98,SutersAZSSL05,YuWLLC25}.
These induced pattern-matching formulations are widely regarded as fundamental problems in graph matching~\cite{Bunke00,ConteFSV04} and have been investigated since the 1970s~\cite{BarrowB76,Levi73,Ullmann76}. They are motivated, for instance, by chemoinformatics, where matching induced substructures models chemical similarity and reactions~\cite{McGregor82,RaymondWillett02}. Beyond this, they arise in applications such as graph databases~\cite{DBLP:conf/sigmod/YanYH05,DBLP:conf/europar/RoweG21}, malware detection~\cite{DBLP:journals/compsec/ParkRS13}, and bioinformatics including RNA structural homology searching~\cite{DBLP:journals/bmcbi/HuangLJ06,DBLP:journals/tcbb/SongLHMXC06}.

Formally, we study the following problems in this paper.
\boxproblemnoparameter{\textsc{Induced Subgraph Isomorphism (ISI)}}{Two graphs $G_1$ (host) and $G_2$ (pattern).}{Does $G_1$ contain an induced subgraph isomorphic to $G_2$?}

\boxproblemnoparameter{\textsc{Maximum Common Induced Subgraph (MCIS)}}{Two graphs $G_1$ and $G_2$ and an integer $t \in \mathbb{N}$.}{Do there exist induced subgraphs $H_1 \subseteq G_1$ and $H_2 \subseteq G_2$ on at least $t$ vertices such that $H_1$ and $H_2$ are isomorphic?}

\textsc{ISI} is the special case of \textsc{MCIS} with threshold $t = |V(G_2)|$.
Both \textsc{ISI} and \textsc{MCIS} are computationally difficult: they subsume many \NP-hard graph problems such as \textsc{Independent Set}, \textsc{Clique}, and \textsc{Induced Matching}.
Since \textsc{Independent Set} and \textsc{Clique} are \W[1]-hard when parameterized by solution size, \textsc{ISI} is \W[1]-hard when parameterized by the pattern size $|V(G_2)|$, and \textsc{MCIS} is \W[1]-hard when parameterized by $\min \{ |V(G_1)|, |V(G_2)| \}$.
Given this hardness for ``natural'' parameters, researchers have turned to fixed-parameter algorithms parameterized by structural parameters of the input graphs.
We follow the same direction in the present paper.
Specifically, we parameterize by the sum of the two input graphs' structural parameters (for \textsc{ISI}, this is equivalent to parameterizing by the host graph's structural parameter $p$ when $p$ is induced-subgraph-monotone, i.e., $p(G') \le p(G)$ for all induced subgraphs $G'$ of $G$).
This parameterization by both input graphs is necessary.
For \textsc{ISI}, the pattern graph alone does not provide a meaningful parameter: the problem generalizes \textsc{Independent Set} and \textsc{Clique}, so it is para-\NP-hard for most standard structural parameters of the pattern graph.
Similarly, for \textsc{MCIS}, parameterizing only one of the two input graphs is not sufficient.

\versionparagraph{Prior work on structural parameterization.}
\textsc{ISI} and \textsc{MCIS} are already \NP-hard on disjoint unions of paths, via a simple reduction from \textsc{3-Partition}~\cite{Damaschke90,DBLP:books/fm/GareyJ79}.
Consequently, they are para-\NP-hard for parameters such as bandwidth, distance to path forest, and feedback edge set number, as observed by Hanaka et al.~\cite{HanakaOOV25}.
Moreover, \textsc{ISI} is \NP-hard on cographs~\cite{Damaschke90} and hence para-\NP-hard for modular-width.
Even on trees of treedepth $3$, Bodlaender et al.~\cite{DBLP:journals/algorithmica/BodlaenderHKKOO20} proved that \textsc{ISI} (and hence \textsc{MCIS}) is \NP-hard.\footnote{Although this hardness proof is for (non-induced) \textsc{Subgraph Isomorphism}, almost the same construction implies \NP-hardness for \textsc{ISI}; see~\cite{GimaHKKO22,HanakaOOV25}.}

On the positive side, Abu-Khzam~\cite{DBLP:journals/ipl/Abu-Khzam14} showed that \textsc{MCIS} is FPT for
$k := \vc(G_1)+\vc(G_2)$, with running time $O^*(k^{O(k)})$, where $\vc(\cdot)$ denotes the vertex cover number.\footnote{We use the $O^*$ notation to suppress the polynomial factors.}
Incidentally, this is tight under Exponential Time Hypothesis (ETH): Abu-Khzam et al.~\cite{DBLP:journals/tcs/Abu-KhzamBS17} proved that \textsc{ISI} admits no
$O^*(k^{o(k)})$-time algorithm unless ETH fails.
Gima et al.~\cite{GimaHKKO22} proved that \textsc{MCIS} (and hence \textsc{ISI}) is FPT for
$k := \vi(G_1)+\vi(G_2)$, where $\vi(\cdot)$ denotes the vertex integrity.
Hanaka et al.~\cite{HanakaOOV25} established that \textsc{MCIS} is FPT for
$k := \ml(G_1)+\ml(G_2)$, $k := \tc(G_1)+\tc(G_2)$, and $k := \nd(G_1)+\nd(G_2)$, where $\ml(\cdot)$, $\tc(\cdot)$, and $\nd(\cdot)$
denote the max leaf number, twin cover number, and neighborhood diversity, respectively.
Each of these parameters is unbounded even on disjoint unions of paths.

In this paper, we focus on the cluster vertex deletion number.
The parameterized complexity of \textsc{ISI} and \textsc{MCIS} with respect to this parameter was still open (see \Cref{fig:hierarchy}). To the best of our knowledge, among the standard parameters, this is the only remaining unresolved case.
Hanaka et al.~\cite{HanakaOOV25} gave an \XP{} algorithm and an FPT-approximation algorithm for \textsc{MCIS} parameterized by $\cvd$, and left open whether \textsc{MCIS} admits an exact FPT algorithm for this parameter.

\begin{figure}[t]
    \centering
    \begin{tikzpicture}[
        xscale=1.4,
        line cap=round,
        line join=round,
        edge/.style={draw=black!80, line width=0.9pt},
        rednode/.style={
            draw=red!70!black,
            fill=red!7,
            rounded corners=3pt,
            thick,
            inner xsep=10pt,
            inner ysep=3pt,
            text=red!70!black,
            minimum height=18pt,
            minimum width=50pt,
            drop shadow={shadow xshift=0.7pt, shadow yshift=-0.7pt, opacity=0.20}
        },
        greennode/.style={
            draw=green!50!black,
            fill=green!10,
            rounded corners=3pt,
            thick,
            inner xsep=10pt,
            inner ysep=3pt,
            text=green!50!black,
            minimum height=18pt,
            minimum width=50pt,
            drop shadow={shadow xshift=0.7pt, shadow yshift=-0.7pt, opacity=0.20}
        },
        greennode*/.style={
        greennode,
        ultra thick,
        draw=green!60!black,
        fill=green!20
        },
    ]
    % ---- nodes (approximate coordinates; tweak to taste) ----
    \node[rednode] (cw)  at ( -2.0, 2.4) {cw};
    \node[rednode] (tw)  at ( 0.0, 1.9) {tw};
    \node[rednode] (pw)  at ( 0.0, 0.9) {pw};

    \node[rednode] (mw)  at (-4.0, 1.4) {mw \cite{Damaschke90}};
    \node[rednode] (sd)  at (-2.0, 1.4) {sd};
    \node[greennode*]   (cvd) at (-2.0, 0.4) {cvd$^{\star}$};

    \node[rednode] (td)  at ( 0.0, -0.1) {td \cite{GimaHKKO22}};
    
    \node[greennode] (ml)  at ( 2.0, -0.1) {ml \cite{HanakaOOV25}};

    \node[greennode] (nd) at (-4.0,-0.6) {nd \cite{HanakaOOV25}};
    \node[greennode] (tc) at (-2.0,-0.6) {tc \cite{HanakaOOV25}};
    \node[greennode]  (vi) at ( 0.0,-1.1) {vi \cite{GimaHKKO22}};
    \node[greennode]  (vc) at ( -2.0,-1.6) {vc \cite{DBLP:journals/ipl/Abu-Khzam14}};

    % ---- edges (drawn in background so they sit behind nodes) ----
    \begin{pgfonlayer}{background}
    \draw[edge] (cw) -- (tw);
    \draw[edge] (cw) -- (mw);
    \draw[edge] (cw) -- (sd);

    \draw[edge] (mw) -- (nd);
    \draw[edge] (mw) -- (tc);

    \draw[edge] (sd) -- (nd);
    \draw[edge] (sd) -- (cvd);
    \draw[edge] (cvd) -- (tc);
    \draw[edge] (sd) -- (td);

    \draw[edge] (tw) -- (pw);
    \draw[edge] (pw) -- (td);
    
    \draw[edge] (pw) -- (ml);

    \draw[edge] (td) -- (vi);
    \draw[edge] (vi) -- (vc);

    \draw[edge] (nd) -- (vc);
    \draw[edge] (tc) -- (vc);

    \end{pgfonlayer}

    \end{tikzpicture}
    \caption{
    Hierarchy of structural parameters for \textsc{ISI}/\textsc{MCIS}. Green nodes indicate \FPT{} results, red nodes indicate para-\NP-hardness. A parameter placed higher is smaller (more restrictive). The starred node $\cvd$ denotes our new result. Abbreviations: cw = clique-width, ml = max leaf number, pw = path-width, tw = tree-width, mw = modular-width, sd = shrub-depth, cvd = cluster vertex deletion number, td = treedepth, nd = neighborhood diversity, tc = twin cover number, vi = vertex integrity, vc = vertex cover number.
    } \label{fig:hierarchy}
\end{figure}
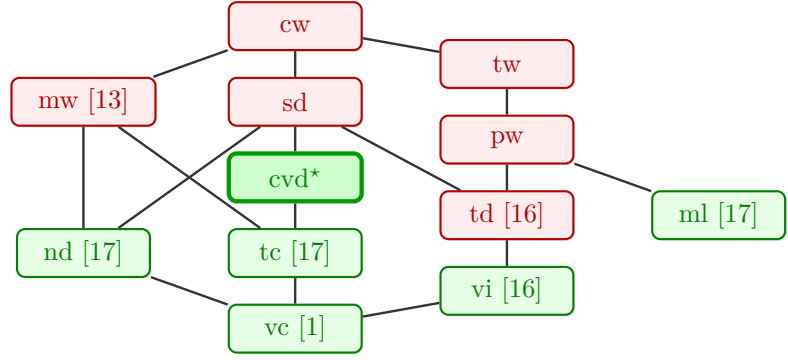

\subsection{Our contributions}

We resolve the open problem of Hanaka et al.~\cite{HanakaOOV25} by giving randomized FPT algorithms for \textsc{ISI} and \textsc{MCIS} parameterized by
$k := \cvd(G_1)+\cvd(G_2)$.\footnote{All randomized algorithms in this paper are Monte Carlo with one-sided error: they always output no on no-instances, but may output no on some yes-instances (false negatives).}
This completes the picture for the standard structural parameterizations of \textsc{ISI} and \textsc{MCIS} summarized in \Cref{fig:hierarchy}.

For \textsc{ISI}, we obtain the following.
\begin{restatable}{theorem}{isifpt}\label{thm:isi-fpt}
    \prob{ISI} can be solved in randomized $O^*(k^{O(k)})$ time, where $k = \cvd(G_1) + \cvd(G_2)$.
\end{restatable}
Since $\vc(G)\ge \cvd(G)$ for every graph $G$, \Cref{thm:isi-fpt} strengthens Abu-Khzam's earlier FPT algorithm parameterized by vertex cover number~\cite{DBLP:journals/ipl/Abu-Khzam14} (albeit with randomization) while maintaining the same asymptotic dependence $k^{O(k)}$ on the parameter.
Moreover, this dependence is essentially optimal already for the vertex cover parameterization: Abu-Khzam et al.~\cite{DBLP:journals/tcs/Abu-KhzamBS17} showed that \textsc{ISI} admits no $O^*(k^{o(k)})$-time algorithm for $k=\vc(G_1)+\vc(G_2)$ unless ETH fails.
Therefore, we cannot expect to improve the running time bound of \Cref{thm:isi-fpt}.

We also obtain a randomized FPT algorithm for \textsc{MCIS}.
\begin{restatable}{theorem}{mcisfpt}\label{thm:mcis-fpt}
    \prob{MCIS} can be solved in randomized $O^*(2^{O(k^2)})$ time, where $k = \cvd(G_1) + \cvd(G_2)$.
\end{restatable}

Our algorithms for \Cref{thm:isi-fpt,thm:mcis-fpt} follow the same strategy.
We first guess some structural information about a solution, such as a partial correspondence between vertices in the cluster vertex deletion sets.
For both \textsc{ISI} and \textsc{MCIS}, this guessing step is the bottleneck.
It invokes $O^*(k^{O(k)})$ branches for \textsc{ISI}, and  $O^*(2^{O(k^2)})$ branches for \textsc{MCIS}.

Once the guess is fixed, the remaining choices can be encoded as a weighted matching problem in an edge-colored multigraph, where all edge weights are bounded by a polynomial in the input size.
This leads to the following auxiliary problem, which we believe is of independent interest.

\boxproblem{\textsc{Weighted Exact Multicolored Matching}}%
{An undirected multigraph $G=(V,E)$ with $n := |V|$ vertices and $m := |E|$ edges, an edge-coloring $c \colon E \to \{0\}\cup \{1,\dots,k\}$, nonnegative integer weights $w \colon E \to \mathbb{N}$, and a threshold $t \in \mathbb{N}$.}%
{Is there a matching $M \subseteq E$ such that $|\{e \in M : c(e)=i\}|=1$ for each $i \in \{1,\dots,k\}$ and $\sum_{e \in M} w(e) \ge t$?}%
{$k$}

For this problem, we give an algebraic randomized pseudo-polynomial algorithm:

\begin{restatable}{proposition}{wemm} \label{prop:mpm}
    \prob{Weighted Exact Multicolored Matching} can be solved in randomized $O^*(2^k W)$ time and $(n + m)^{O(1)}$ space when each weight is at most $W$.
\end{restatable}

At a high level, we design a polynomial that algebraically enumerates matchings meeting the color and weight constraints, using Pfaffians together with interpolation and inclusion-exclusion.
We then test the existence of such a matching by polynomial identity testing, which is randomized.
This raises the natural question of whether \prob{Weighted Exact Multicolored Matching} admits a deterministic FPT algorithm.
However, even the unweighted special case is closely related to \textsc{Exact Matching}, for which no direct derandomization route is known.
Although \textsc{Exact Matching} admits a randomized polynomial-time algorithm via the isolation lemma \cite{DBLP:journals/combinatorica/MulmuleyVV87}, derandomizing it is a major open problem; in particular, even the existence of a deterministic FPT algorithm parameterized by $k$ remains unknown, e.g., \cite{DBLP:conf/stacs/Maalouly23,MurakamiY25}.

\begin{restatable}{proposition}{wemmtoem} \label{prop:wemmtoem}
  If (unweighted) \prob{Exact Multicolored Matching} admits a deterministic FPT algorithm parameterized by the number of nonzero colors, then \prob{Exact Matching} admits a deterministic FPT algorithm parameterized by the target number of red edges.
\end{restatable}

In our \textsc{ISI} and \textsc{MCIS} algorithms, all guessing steps are deterministic, and randomness enters only through this matching subroutine.
Thus, \Cref{prop:wemmtoem} suggests that derandomizing these algorithms via our approach is difficult.

\medskip

Aside from derandomization, it is natural to ask whether the $O^*(2^{O(k^2)})$ running time in \Cref{thm:mcis-fpt} can be improved to $O^*(k^{O(k)})$.
Somewhat surprisingly, we show that the quadratic dependence in the exponent is unavoidable: in fact, \Cref{thm:mcislb} rules out any $O^*(2^{o(k^2)})$-time algorithm under ETH.

\begin{restatable}{theorem}{mcislb}\label{thm:mcislb}
    Assuming ETH, there is no $O^*(2^{o(k^2)})$-time algorithm for \prob{MCIS} parameterized by $k = \cvd(G_1) + \cvd(G_2)$.
\end{restatable}

To the best of our knowledge, this is the first standard structural parameter for which \textsc{ISI} and \textsc{MCIS} provably have different optimal FPT dependences on the parameter.

The proof of \Cref{thm:mcislb} proceeds via a reduction that uses the following intermediate problem, which may be of independent interest.

\boxproblem{\textsc{List-Constrained Boolean Matrix (LCBM)}}{An integer $k\in \mathbb{N}$, two sets of length-$k$ Boolean vectors $\mathcal{R}_i \subseteq \{0,1\}^k$ and $\mathcal{C}_i \subseteq \{0,1\}^k$ for each $i \in [k]$.}{Is there a Boolean matrix $A \in \{0,1\}^{k\times k}$ such that, for every $i \in [k]$, the $i$-th row $A[i,\cdot]$ lies in $\mathcal{R}_i$ and the $i$-th column $A[\cdot,i]$ lies in $\mathcal{C}_i$?}{$k$}

We allow the lists $\mathcal{R}_i$ and $\mathcal{C}_i$ to be exponentially large (up to $2^k$ each).
Note that \textsc{LCBM} admits a straightforward $O^*(2^{O(k^2)})$-time algorithm: enumerate all $2^{k^2}$ Boolean matrices $A \in \{0,1\}^{k\times k}$ and test whether every row and column of $A$ is compatible with the corresponding lists.
We show that this trivial algorithm is essentially optimal under ETH:

\begin{restatable}{theorem}{lcbmlb}
\label{thm:lcbmlb}
    Assuming ETH, there is no $O^*(2^{o(k^2)})$-time algorithm for \prob{LCBM}.
\end{restatable}

To obtain the ETH lower bound for \textsc{MCIS}, we essentially encode the row constraints into one graph and the column constraints into the other, so that a large common induced subgraph exists if and only if there is a matrix satisfying all row and column lists.
This kind of encoding is inherently impossible for \textsc{ISI}.

\medskip

Finally, we study a natural $\ell$-graph generalization of \textsc{MCIS}, which we call \textsc{$\ell$-MCIS}.

\boxproblemnoparameter{\textsc{$\ell$-Maximum Common Induced Subgraph ($\ell$-MCIS)}}{$\ell$ graphs $G_1,\dots, G_{\ell}$ and an integer $t \in \mathbb{N}$.}{Do there exist induced subgraphs $H_1 \subseteq G_1,\dots, H_{\ell} \subseteq G_{\ell}$ on at least $t$ vertices such that $H_1,\dots, H_{\ell}$ are isomorphic?}
The case $\ell=2$ is equivalent to \textsc{MCIS}.
We show that, with respect to the cluster vertex deletion number, the two-graph case lies at the boundary of tractability: while (2-)\textsc{MCIS} admits an FPT algorithm parameterized by $\cvd$, the problem becomes intractable already for three input graphs, even when each graph is extremely close to a cluster graph.
\begin{restatable}{theorem}{tmcislb}
\label{thm:tmcislb}
    \textsc{3-MCIS} is NP-hard when $\cvd(G_i) = 2$ for all $i \in \{ 1, 2, 3 \}$.
\end{restatable}

This sharp transition is specific to $\cvd$.
For vertex cover number, \textsc{$\ell$-MCIS} remains fixed-parameter tractable: letting
$k := \max_{i \in [\ell]} \vc(G_i)$, we obtain the following.
\begin{restatable}{theorem}{tmcisvc}
    \textsc{$\ell$-MCIS} can be solved in $O^*(2^{O(k^2)}\ell)$ time, where $k = \max_{i \in [\ell]} \vc(G_i)$.
\end{restatable}
Although we do not study it in this paper, we expect that many existing structurally parameterized FPT algorithms extend to \textsc{$\ell$-MCIS}. We leave this to future work.

\subsection{Further Related Work}

Abu-Khzam et al.~\cite{DBLP:journals/tcs/Abu-KhzamBS17} proved that, when parameterized by $\vc(G_1)+\vc(G_2)$, \textsc{MCIS} does not admit a polynomial kernel unless $\NP\subseteq \coNP/\mathsf{poly}$.
They also proved \textsc{ISI} is \W[1]-hard even on $C_4$-free bipartite graphs with degeneracy $2$.
Marx and Schlotter~\cite{DBLP:journals/algorithmica/MarxS13} investigated \textsc{ISI} for the case where both $G_1$ and $G_2$ are interval graphs. They gave an FPT algorithm parameterized by $|V(G_1)|-|V(G_2)|$, as well as proving \W[1]-hardness parameterized by $|V(G_2)|$.
Heggernes et al.~\cite{DBLP:journals/tcs/HeggernesHMV15} proved that, when $G_1$ is an interval graph and $G_2$ is a proper interval graph, \textsc{ISI} admits an FPT algorithm parameterized by the number of connected components of $G_2$. They also proved that \textsc{ISI} remains \NP-complete when both $G_1$ and $G_2$ are proper interval graphs.
% Damaschke~\cite{Damaschke90} proved that \textsc{ISI} is \NP-complete on cographs.
Chen et al.~\cite{DBLP:conf/icalp/ChenTW08} proved that, for any graph class $\mathcal{G}$ containing arbitrarily large graphs, \textsc{ISI} for the case $G_2\in \mathcal{G}$ is \W[1]-hard parameterized by $|V(G_2)|$.

%\todo{non-induced version, other graph classes}

%\begin{itemize}
%    \item $O^*(k^{O(k)})$-time algorithm \cite{DBLP:journals/ipl/Abu-Khzam14}
%    \item $O^*(k^{o(k)})$ time ETH lower bound for $k = \vc(G_1) + \vc(G_2)$ \cite{DBLP:journals/tcs/Abu-KhzamBS17}
%    \item \textsc{ISI} is FPT when parametreized by vertex integrity\cite{DBLP:journals/algorithmica/BodlaenderHKKOO20}
%    \item \textsc{ISI} is in P on connected proper interval graphs and connected bipartite permutation graphs \cite{DBLP:journals/tcs/HeggernesHMV15}
%    \item  \textsc{ISI} where the pattern is a graph class $\mathcal{C}$ is solvable in polynomial time if and only
%    if $\mathcal{C}$ contains no arbitrarily large graphs. \cite{DBLP:conf/icalp/ChenTW08}
%\end{itemize}

\subsection{Preliminaries}
%\todo{revise; combine with current preliminaries}

Given a positive integer $a$, we write $[a] = \{1, \dots, a\}$.
For a matrix $A$, we denote the $(i,j)$-element of $A$ by $A[i,j]$, the $i$-th row by $A[i,\cdot]$, and the $j$-th column by $A[\cdot,j]$.

We use standard terminology from parameterized complexity, see, e.g.,
the book of Cygan et al.~\cite{DBLP:books/sp/CyganFKLMPPS15}.
A graph is a \emph{cluster graph} if it is a disjoint union of cliques; equivalently, it has no induced path on three vertices.
A \emph{cluster vertex deletion set} of a graph $G$ is a vertex subset $X\subseteq V(G)$ such that $G-X$ is a cluster graph.
The \emph{cluster vertex deletion number} of a graph $G$, denoted by $\cvd(G)$, is the minimum size of a cluster vertex deletion set of $G$.
A standard branching algorithm computes a cluster vertex deletion set of size at most $k$ in $O^*(3^{k})$ time.
A faster algorithm achieves $O^*(1.7549^k)$ time~\cite{DBLP:journals/mst/TianXY25}.

For two graphs $G_1$ and $G_2$, an \emph{induced embedding} from $G_2$ into $G_1$ is an injective map $f \colon V(G_2) \to V(G_1)$ such that for all $u,v \in V(G_2)$,
\[
\{u,v\} \in E(G_2) \iff \{f(u),f(v)\} \in E(G_1).
\]
Equivalently, $f$ is an isomorphism between $G_2$ and the induced subgraph $G_1[f(V(G_2))]$.
For $Z_2 \subseteq V(G_2)$, a \emph{partial induced embedding} is an induced embedding from $G_2[Z_2]$ to $G_1$.
Using this terminology, \textsc{ISI} asks for an induced embedding of $G_2$ into $G_1$.

%Let $\mathcal{G}_{\mathrm{cluster}}$ denote the class of cluster graphs (disjoint unions of cliques).
%For $k\ge 0$, let $k\mathrm{C}\subseteq \mathcal{G}_{\mathrm{cluster}}$ be the subclass consisting of graphs with $k$ cliques.
%We write $\mathcal{G}_{\mathrm{cluster}} + k\mathrm{C}$ for the class of graphs $H$  that is obtained from the disjoint union of two cluster graphs $G_1$ and $G_2$ (with $G_1$ having exactly $k$ cliques) by adding an arbitrary bipartite edge set between $V(G_1)$ and $V(G_2)$.

%In the \textsc{Induced Subgraph Isomorphism (ISI)} problem, we are given graphs $G_1$ (host) and $G_2$ (pattern), we are tasked with finding an induced embedding from $G_2$ to~$G_1$.

%We consider \textsc{ISI} parameterized by the \emph{cluster vertex deletion number}: $k = \cvd(G_1) + \cvd(G_2)$.
%For $i = 1,2$, let $X_i$ be the cluster vertex deletion set of $G_i$.
%Write $\mathcal{C}_i = \{ C_i^1, C_i^2, \dots \}$ for the family of cliques (connected components) of $G_i-X_i$, for $i\in\{1,2\}$.
%We prove the following:

%For a set of rows $I$ and columns $J$, we denote by $A[I, J]$ the submatrix containing rows $I$ and columns $J$.
%If $I$ contains all rows ($J$ contains all columns), then we use the shorthand $A[\cdot, J]$ ($A[I, \cdot]$, respectively).

\subsection{Technical Overviews}

In this section, we give technical overviews for our main results.

\versionparagraph{Overview for \Cref{thm:isi-fpt}}

We begin with the $O^*(k^{O(k)})$-time algorithm for \textsc{ISI}.
First consider the simpler case where $G_2$ is a cluster graph ($\cvd(G_2)=0$); denote its cliques by $\mathcal{C}_2$.
We show that this case can be solved in $O^*(2^{O(k)})$ time.
Let $X_1$ be a cluster vertex deletion set of $G_1$.
In time $2^{|X_1|}$ we guess which vertices of $X_1$ are used by the embedding (and discard the rest of $X_1$).
This yields the following situation:
\begin{itemize}
    \item $G_1[X_1]$ is a cluster graph (denote its clique components by $\mathcal{S}_1$); otherwise, $G_2$ cannot be embedded into~$G_1$.
    \item $G_1-X_1$ is a cluster graph (denote its cliques by $\mathcal{C}_1$); $X_1$ is a cluster vertex deletion set. 
\end{itemize}
Each component $C_2\in \mathcal{C}_2$ must be embedded into the union of at most one clique $S_1\in \mathcal{S}_1$ and at most one clique $C_1\in \mathcal{C}_1$; otherwise, the image of $C_2$ would contain vertices from two disjoint cliques, contradicting that $C_2$ is a clique.
Moreover, any clique in $\mathcal{S}_1$ or $\mathcal{C}_1$ can intersect the image of at most one clique in $\mathcal{C}_2$; otherwise, distinct cliques in $\mathcal{C}_2$ would be adjacent.

Thus, the remaining task is to assign each component $C_2\in \mathcal{C}_2$ to $S_1 \in \mathcal{S}_1$, $C_1 \in \mathcal{C}_1$, or a pair of cliques $(S_1, C_1)\in \mathcal{S}_1\times \mathcal{C}_1$, in a way that each clique in $\mathcal{S}_1 \cup \mathcal{C}_1$ is used at most once.
We encode admissible choices in an edge-colored bipartite graph with vertex set $\mathcal{S}_1 \cup \mathcal{C}_1 \cup \mathcal{C}_2$.
For each $C_1 \in \mathcal{C}_1$ and $C_2 \in \mathcal{C}_2$, we add an edge $(C_1,C_2)$ of color $S_1$ whenever $C_2$ can be embedded into $S_1\cup C_1$.
If $C_2$ embeds entirely into $C_1$ without interfering with other potential embedded cliques, we add $(C_1,C_2)$ with the special color $0$, which may be used arbitrarily many times.
If $C_2$ embeds entirely into $S_1$, we add an edge corresponding to $S_1$ for $C_2$ with color $S_1$.
With the structural observations above, an edge set that covers all vertices of $\mathcal{C}_2$ and uses each nonzero color exactly once corresponds to a valid induced embedding.
In particular, it ensures the selected clique pairs are anticomplete.
This is an instance of (unweighted) \textsc{Exact Multicolored Matching} with at most $|\mathcal{S}_1|\leq |X_1|\leq k$ colors, which can be solved in randomized $O^*(2^k)$ time via algebraic techniques (\Cref{prop:mpm}).

For the general case, we must handle the cluster vertex deletion set of $G_2$.
As we show in \Cref{lem:modulator-to-modulator}, if $G_2$ admits an induced embedding $\phi$ into $G_1$, there exists a minimal cluster vertex deletion set $X_2$ of $G_2$ that is mapped entirely into $X_1$, i.e., $\phi(X_2) \subseteq X_1$.
Any graph has at most $3^k$ minimal cluster vertex deletion sets of size at most $k$, so we enumerate all candidates for $X_2$ in $O^*(3^k)$ time.
For each choice, we guess the partial induced embedding $X_2\to X_1$ among $k^{O(k)}$ possibilities.
We then need to check whether this partial embedding can be extended.
For example, to decide whether a clique $C_2$ can be embedded into a clique $C_1$ while respecting the fixed mapping on $X_2$, we compare adjacency patterns to the fixed set: vertices with the same pattern are interchangeable, so feasibility reduces to counting.
This ``signature'' argument is formalized in \Cref{lem:clique-to-clique} and reused in our algorithm for \textsc{MCIS}.

\versionparagraph{Overview for \Cref{thm:mcis-fpt}}

Next, we give an overview of our $O^*(2^{O(k^2)})$-time algorithm for \textsc{MCIS}.
As discussed earlier, the strategy parallels \Cref{thm:isi-fpt}: we guess a certain structural information and then reduce the remaining choice among cliques to \textsc{Weighted Exact Multicolored Matching}.
We highlight the additional guessing and sketch the construction of the matching instance.
Let $H$ be a hypothetical maximum common induced subgraph which is isomorphic to $H_1\subseteq G_1$ and $H_2\subseteq G_2$.
For each $i\in \{1,2\}$, let $X_i$ be a cluster vertex deletion set of $G_i$.
We first guess, from $2^{2k}$ options, which vertices of $X_i$ are used in $H_i$ and delete the rest; hence we may assume $X_i\subseteq V(H_i)$ and let $X'_i\subseteq V(H)$ be the corresponding vertex set in $H$.

For the reduction to \textsc{Weighted Exact Multicolored Matching}, we need the structure of $H[X'_1\cup X'_2]$, the graph induced by the union of preimages of $X_1$ and $X_2$, to be fixed.
We already know $H[X'_i]$ as it is isomorphic to $H_i[X_i]$.
We further guess, from $k^{O(k)}$ options, the vertex set $X'_1\cap X'_2$ together with the correspondence of vertices between them.
Unlike \textsc{ISI}, we must also guess the edges \emph{between} $X'_1\setminus X'_2$ and $X'_2\setminus X'_1$; this yields $2^{|X'_1\setminus X'_2|\cdot |X'_2\setminus X'_1|}\leq 2^{k^2}$ possibilities and forms the bottleneck, which is unavoidable under ETH (\Cref{thm:mcislb}).

These guesses determine the interface of $H$ with the modulators and induce a structured decomposition; see \Cref{fig:mcis}.
Let $I:=X'_1\cap X'_2$ and $J_i:=X'_i\setminus I$.
Then $K:=V(H)\setminus (I\cup J_1\cup J_2)$ induces a cluster graph; let $\mathcal{D}$ be its clique components, and let $\mathcal{T}_i$ be the clique components of $H[J_i]$.
Because $H[K\cup J_i]$ is a cluster graph, each $D\in\mathcal{D}$ can be joined to at most one clique in $\mathcal{T}_i$, which yields the four-way partition $\mathcal{D}=\mathcal{D}_1\dot\cup\mathcal{D}_2\dot\cup\mathcal{D}_3\dot\cup\mathcal{D}_0$ depending on whether $D$ attaches to $J_1$, to $J_2$, to both, or to neither.
Symmetrically, $\mathcal{T}_i$ splits into $\mathcal{T}_{i,0}$, $\mathcal{T}_{i,i}$, and $\mathcal{T}_{i,3}$ according to which $\mathcal{D}$-clique it pairs with.
This decomposition drives the subsequent reduction to \textsc{Weighted Exact Multicolored Matching}.
\shortonly{%
\begin{sidewaysfigure}
    \begin{tikzpicture}[
    scale=0.98,
    font=\small,
    clique/.style={
        draw, rounded corners=5pt, line width=0.9pt,
        minimum width=2.8cm, minimum height=0.8cm, inner sep=7pt
    },
    vtx/.style={
        draw, ellipse, line width=0.9pt,
        minimum width=10mm, minimum height=6mm, inner sep=0pt
    },
    container/.style={
        draw, rounded corners=10pt, line width=0.9pt, inner sep=6pt
    },
    conn/.style={line width=1.8pt},
    dconn/.style={line width=1.8pt, dashed}
    ]

    \node[clique] (T10) at (-8,2) {};
    \node[anchor=south] (T10lab) at ($(T10.north)+(0,0.18)$) {$\mathcal T_{1,0}$};
    \node[vtx] (T10a) at ($(T10.south)+(-0.60,0.4)$) {};
    \node[vtx] (T10b) at ($(T10.south)+( 0.60,0.4)$) {};

    \node[clique] (T11) at (-5,2) {};
    \node[anchor=south] (T11lab) at ($(T11.north)+(0,0.18)$) {$\mathcal T_{1,1}$};
    \node[vtx] (T11a) at ($(T11.south)+(-0.60,0.4)$) {};
    \node[vtx] (T11b) at ($(T11.south)+( 0.60,0.4)$) {};

    \node[clique] (T13) at (-2,2) {};
    \node[anchor=south] (T13lab) at ($(T13.north)+(0,0.18)$) {$\mathcal T_{1,3}$};
    \node[vtx] (T13a) at ($(T13.south)+(-0.60,0.4)$) {};
    \node[vtx] (T13b) at ($(T13.south)+( 0.60,0.4)$) {};

    \node[clique] (T23) at (2,2) {};
    \node[anchor=south] (T23lab) at ($(T23.north)+(0,0.18)$) {$\mathcal T_{2,3}$};
    \node[vtx] (T23a) at ($(T23.south)+(-0.60,0.4)$) {};
    \node[vtx] (T23b) at ($(T23.south)+( 0.60,0.4)$) {};

    \node[clique] (T22) at (5,2) {};
    \node[anchor=south] (T22lab) at ($(T22.north)+(0,0.18)$) {$\mathcal T_{2,2}$};
    \node[vtx] (T22a) at ($(T22.south)+(-0.60,0.4)$) {};
    \node[vtx] (T22b) at ($(T22.south)+( 0.60,0.4)$) {};

    \node[clique] (T20) at (8,2) {};
    \node[anchor=south] (T20lab) at ($(T20.north)+(0,0.18)$) {$\mathcal T_{2,0}$};
    \node[vtx] (T20a) at ($(T20.south)+(-0.60,0.4)$) {};
    \node[vtx] (T20b) at ($(T20.south)+( 0.60,0.4)$) {};

    \node[clique] (D1) at (-5,0) {};
    \node[anchor=north] (D1lab) at ($(D1.south)-(0,0.18)$) {$\mathcal D_1$};
    \node[vtx] (D1a) at ($(D1.south)+(-0.60,0.40)$) {};
    \node[vtx] (D1b) at ($(D1.south)+( 0.60,0.40)$) {};

    \node[clique] (D3) at (0,0) {};
    \node[anchor=west] (D3lab) at ($(D3.east)+(0.18,-0.09)$) {$\mathcal D_3$};
    \node[vtx] (D3a) at ($(D3.south)+(-0.60,0.40)$) {};
    \node[vtx] (D3b) at ($(D3.south)+( 0.60,0.40)$) {};

    \node[clique] (D2) at (5,0) {};
    \node[anchor=north] (D2lab) at ($(D2.south)-(-0,0.18)$) {$\mathcal D_2$};
    \node[vtx] (D2a) at ($(D2.south)+(-0.60,0.40)$) {};
    \node[vtx] (D2b) at ($(D2.south)+( 0.60,0.40)$) {};

    \node[clique, minimum width=4.6cm] (D0) at (0,-1.4) {};
    \node[anchor=west] (D0lab) at ($(D0.east)+(0.18,-0.09)$) {$\mathcal D_0$};
    \node[vtx] (D0a) at ($(D0.south)+(-1.65,0.40)$) {};
    \node[vtx] (D0b) at ($(D0.south)+(-0.55,0.40)$) {};
    \node[vtx] (D0c) at ($(D0.south)+( 0.55,0.40)$) {};
    \node[vtx] (D0d) at ($(D0.south)+( 1.65,0.40)$) {};

    \node[container, fit=(T10)(T11)(T13)(T10lab)(T11lab)(T13lab)] (Jone) {};
    \node[anchor=south] at ($(Jone.north)+(-3,0.20)$) {$J_1 (=\phi_1^{-1}(X_1)\setminus I)$};

    \node[container, fit=(T23)(T22)(T20)(T23lab)(T22lab)(T20lab)] (Jtwo) {};
    \node[anchor=south] at ($(Jtwo.north)+(3,0.20)$) {$J_2(=\phi_2^{-1}(X_2)\setminus I)$};

    \node[container, fit=(D1)(D3)(D2)(D0)(D1lab)(D3lab)(D2lab)(D0lab)] (Kbox) {};
    \node[anchor=north west] at ($(Kbox.south west)+(-0.05,-0.15)$) {$K$};

    \node[container, minimum width=10cm, minimum height=1.2cm] (Ibox) at (0, 5.5) {};
    \node[anchor=south] at ($(Ibox.north)+(0,0.20)$) {$I (=\phi_1^{-1}(X_1)\cap \phi_2^{-1}(X_2))$};

    \draw[conn] (D1a) -- (T11a);
    \draw[conn] (D1b) -- (T11b);
    \draw[conn] (D2a) -- (T22a);
    \draw[conn] (D2b) -- (T22b);
    \draw[conn] (D3a) -- (T13a);
    \draw[conn] (D3b) -- (T13b);
    \draw[conn] (D3a) -- (T23a);
    \draw[conn] (D3b) -- (T23b);

    \draw[dconn] ($(Jone.north east)-(3,0)$) .. controls (-1,4.30) and (1,4.30) .. ($(Jtwo.north west)+(3,0)$);
    \draw[dconn] (Jtwo.north) -- (Ibox) -- (Jone.north);
    \draw[dconn] (Ibox.west) .. controls (-13,4.5) and (-13,0) .. (Kbox.west);

    \end{tikzpicture}
    \caption{The structure of $H$: each ellipse is a clique in the indicated family, thick segments denote complete joins between the two cliques, and the dashed lines indicate that edges in between are arbitrary.}
    \label{fig:mcis}
\end{sidewaysfigure}
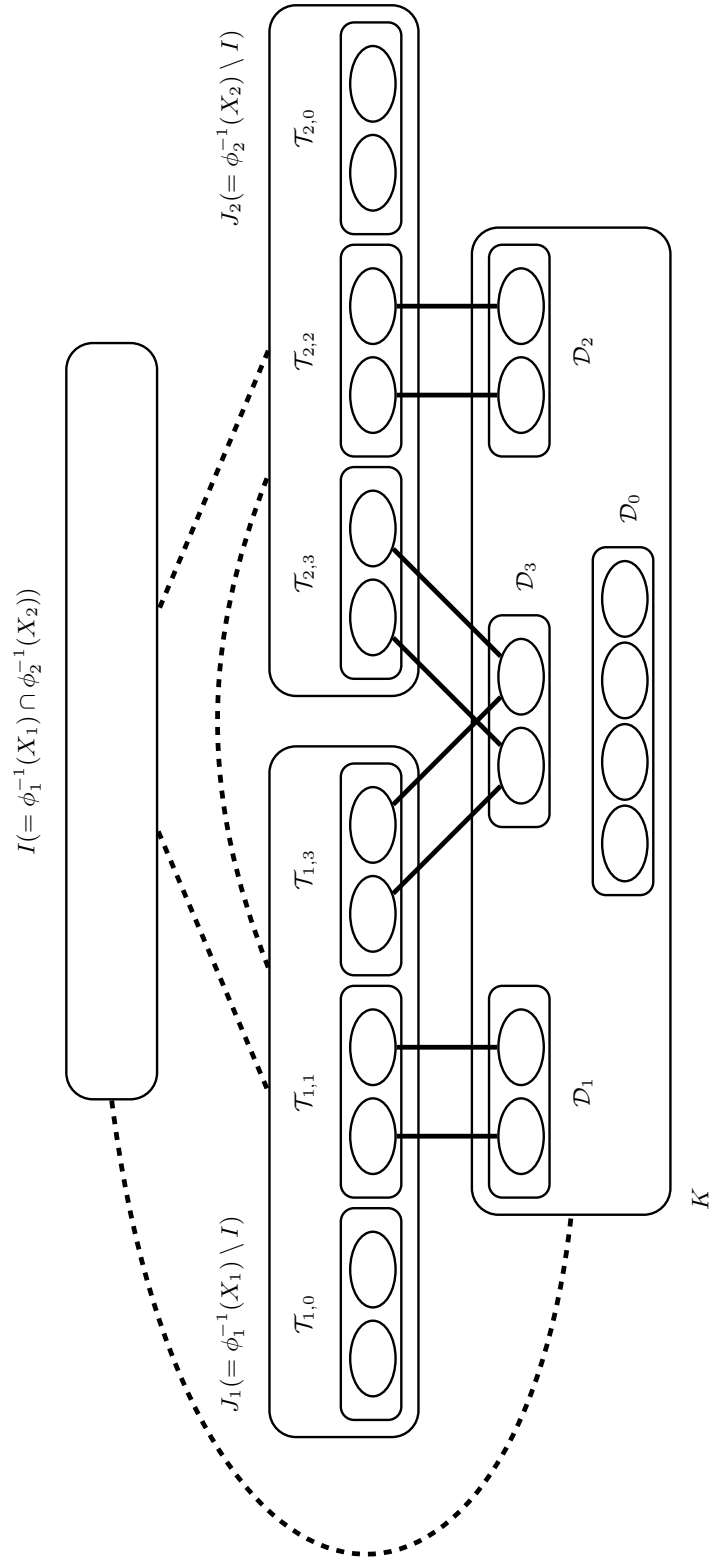
}
%First, we guess the partition of $J_i$ into $\mathcal{S}_{i,0}$, $\mathcal{S}_{i,i}$, and $\mathcal{S}_{i,3}$, where each clique of each part is embedded into a clique of $\mathcal{T}_{i,0}$, union of a clique in $\mathcal{T}_{i,i}$ and a clique in $\mathcal{D}_i$, and union of a clique in $\mathcal{T}_{i,3}$ and a clique in $\mathcal{D}_3$, respectively.
%Then, 
We build an edge-colored bipartite graph whose left side consists of clique components of $G_1-X_1$ together with the cliques in $\mathcal{S}_{1,0}$, and whose right side consists of clique components of $G_2-X_2$ together with $\mathcal{S}_{2,0}$, where $\mathcal{S}_{i,0}$ is the set of guessed cliques in which cliques in $\mathcal{T}_{i,0}$ are embedded.
Edges encode feasible placements of cliques of $H$: for each pair of cliques $(C_1,C_2)$ from the two sides we add colored edges representing the possibility that a clique $D\in\mathcal{D}_3,\mathcal{D}_1,\mathcal{D}_2$, or $\mathcal{D}_0$ maps into $C_1$ and $C_2$ together with its attached $\mathcal{T}$-cliques; edges incident to $\mathcal{S}_{1,0}$ or $\mathcal{S}_{2,0}$ represent cliques in $\mathcal{T}_{1,0}$ or $\mathcal{T}_{2,0}$ that embed entirely on the other side.
Nonzero colors must be used exactly once, while a special color $0$ can be reused.
Edge weights correspond to the maximum feasible size of the underlying clique $D$. The relevant quantities are computed via the \emph{signature argument} by entry-wise comparison of vectors representing inter-clique edges.

\versionparagraph{Overview for \Cref{prop:mpm}}
\Cref{prop:mpm} states that \textsc{Weighted Exact Multicolored Matching} can be solved in randomized $O^*(2^k W)$ time and polynomial space when all weights are at most $W$.
We encode perfect matchings as monomials in the Pfaffian of a skew-symmetric matrix whose entries carry variables for colors and weights.
Using interpolation and inclusion-exclusion over the $k$ colors, we isolate the polynomial corresponding to matchings that use each nonzero color exactly once.
We then interpolate on the weight variable and apply the Schwartz--Zippel lemma (polynomial identity testing) to detect whether a coefficient of degree at least $t$ is nonzero, yielding the stated randomized bound.

\versionparagraph{Overview for \Cref{thm:mcislb}}
Here we outline the hardness proof for \Cref{thm:mcislb}, which states that assuming ETH there is no $O^*(2^{o(k^2)})$-time algorithm for \textsc{MCIS} parameterized by $k=\cvd(G_1)+\cvd(G_2)$.
The proof has two steps: we first reduce $3$-SAT on $k^2$ variables to \textsc{LCBM} (\Cref{thm:lcbmlb}) with parameter $k$, and then reduce \textsc{LCBM} (parameter $k$) to \textsc{MCIS} (parameter $k$).

Let $\varphi$ be a $3$-CNF on $n$ variables. We may assume each variable appears only a constant number of times by the sparsification lemma of Impagliazzo et al.~\cite{DBLP:journals/jcss/ImpagliazzoPZ01}.
We encode $\varphi$ as an instance of \textsc{LCBM}.
We first partition the variables into $k:=O(\sqrt{n})$ parts $X_1,\dots, X_k$ so that each part contains $O(\sqrt{n})$ variables.
We do the same for the clauses to obtain a partition $S_1,\dots, S_k$.
Let $Y_j$ be the set of variables that appear in clauses in $S_j$.
We can ensure $|X_i\cap Y_j|\leq 1$ for all $i,j$ by a suitable partition obtained via greedy coloring of a bounded-degree graph.

In our construction, each cell $A[i,j]$ corresponds to the assignment of the unique variable in $X_i\cap Y_j$; if there is no such variable, we may set $A[i,j]$ arbitrarily.
The row constraint $\mathcal{R}_i$ enforces consistency among all occurrences of variables in $X_i$ by requiring equal values in the corresponding cells.
The column constraint $\mathcal{C}_j$ enforces satisfaction of the clauses in $S_j$ by restricting the column to vectors that satisfy all clauses over $Y_j$.
Thus, a matrix $A$ satisfying all row and column constraints corresponds to a satisfying assignment of $\varphi$.

We now reduce \textsc{LCBM} to \textsc{MCIS}.
Recall that, in our algorithm for \textsc{MCIS}, the bottleneck is guessing the edges between $P:=X'_1\setminus X'_2$ and $Q:=X'_2\setminus X'_1$ in the hypothetical solution~$H$.
In the reduction, vertices in $P$ and $Q$ correspond to the rows and columns of the matrix~$A$, respectively, and an edge $(i,j)\in P\times Q$ represents setting $A[i,j]=1$.

We now sketch the graph construction; see \Cref{fig:mcis-lb:g0}.
\shortonly{%
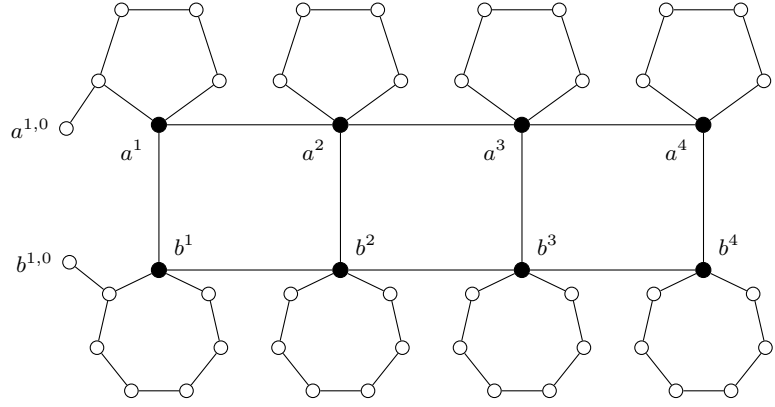
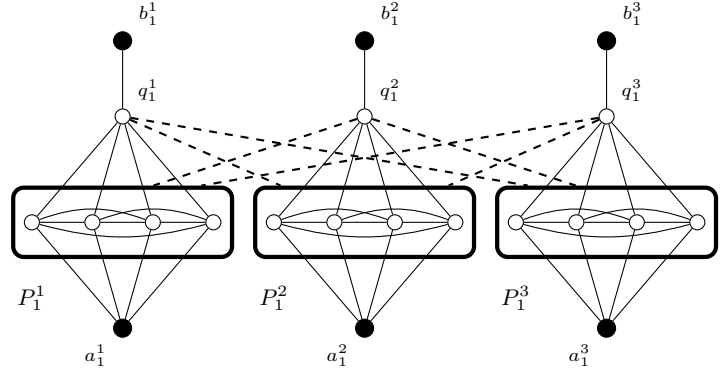
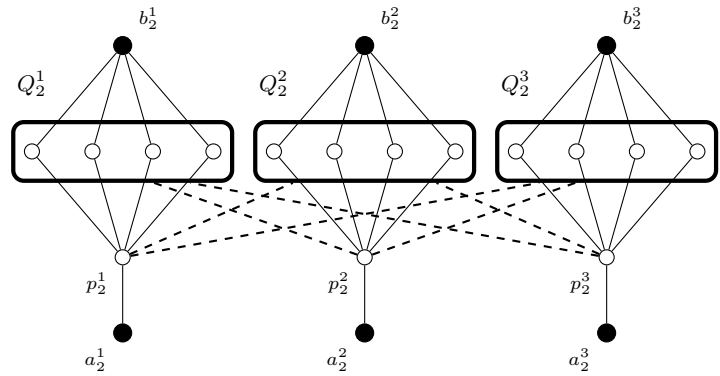
\begin{figure}
    \centering
    \begin{subfigure}[t]{.9\linewidth}
    \begin{center}
        \begin{tikzpicture}[
    x=2cm,y=2cm,
    main/.style={circle,draw,fill=black,inner sep=2pt},
    aux/.style={circle,draw,inner sep=1.8pt},
    every node/.style={font=\footnotesize},
    scale=1.2
  ]

  \def\k{4}
  \def\ay{0.8}
  \def\by{0}
  \def\rA{0.35}
  \def\rB{0.35}

  \foreach \i in {1,...,\k} {
    \node[main,label=below left:{$a^{\i}$}] (a\i) at (\i,\ay) {};
    \node[main,label=above right:{$b^{\i}$}] (b\i) at (\i,\by) {};
  }

  \foreach \i in {1,...,\numexpr\k-1\relax} {
    \pgfmathtruncatemacro{\j}{\i+1}
    \draw (a\i) -- (a\j);
    \draw (b\i) -- (b\j);
  }

  \foreach \i in {1,...,\k} {
    \draw (a\i) -- (b\i);
  }

  \foreach \i in {1,...,\k} {
    \pgfmathsetmacro{\cx}{\i}
    \pgfmathsetmacro{\cy}{\ay + \rA}
    \foreach \j in {1,...,4} {
      \pgfmathsetmacro{\angle}{270 + 72*\j}
      \coordinate (tmp) at ({\cx + \rA*cos(\angle)},{\cy + \rA*sin(\angle)});
      \node[aux] (a-\i-\j) at (tmp) {};
    }
    \draw (a\i) -- (a-\i-1) -- (a-\i-2) -- (a-\i-3) -- (a-\i-4) -- (a\i);
  }

  \foreach \i in {1,...,\k} {
    \pgfmathsetmacro{\cx}{\i}
    \pgfmathsetmacro{\cy}{\by - \rB}
    \foreach \j in {1,...,6} {
      \pgfmathsetmacro{\angle}{90 + 360/7*\j}
      \coordinate (tmpb) at ({\cx + \rB*cos(\angle)},{\cy + \rB*sin(\angle)});
      \node[aux] (b-\i-\j) at (tmpb) {};
    }
    \draw (b\i) -- (b-\i-1) -- (b-\i-2) -- (b-\i-3)
      -- (b-\i-4) -- (b-\i-5) -- (b-\i-6) -- (b\i);
  }

  \pgfmathsetmacro{\cxA}{1}
  \pgfmathsetmacro{\cyA}{\ay + \rA}
  \pgfmathsetmacro{\angA}{216}
  \node[aux,label=left:{$a^{1,0}$}] (a10)
    at ({\cxA + 1.8*\rA*cos(\angA)},{\cyA + 1.8*\rA*sin(\angA)}) {};
  \draw (a10) -- (a-1-4);

  \pgfmathsetmacro{\cxB}{1}
  \pgfmathsetmacro{\cyB}{\by - \rB}
  \pgfmathsetmacro{\angB}{90 + 360/7*1}
  \node[aux,label=left:{$b^{1,0}$}] (b10)
    at ({\cxB + 1.8*\rB*cos(\angB)},{\cyB + 1.8*\rB*sin(\angB)}) {};
  \draw (b10) -- (b-1-1);

\end{tikzpicture}
    \end{center}
    \subcaption{Base gadget $G_0$ (illustrated for $k=4$): two length-$k$ paths with attached 5- and 7-cycles, matched by edges $a^i b^i$, plus the extra vertices $a^{1,0}$ and $b^{1,0}$.}
    \label{fig:lb:G0}
    \end{subfigure}

    \vspace{5ex}

    \begin{subfigure}{.9\linewidth}
    \begin{center}
    \begin{tikzpicture}[
        x=2cm,y=2cm,
        main/.style={circle,draw,fill=black,inner sep=2.4pt},
        vtx/.style={circle,draw,inner sep=2pt},
        rect/.style={draw,rounded corners,inner sep=4pt},
        every node/.style={font=\scriptsize},
        >={Latex[length=3mm,width=2mm]}
      ]
      \foreach \row/\i in {0/1, 1.6/2, 3.2/3} {
        \begin{scope}[shift={(\row,0)}]
        \node[main,label=below left:{$a_1^{\i}$}] (a-\i) at (0.6,-0.7) {};
        \node[vtx] (p-\i-1) at (0,0) {};
        \node[vtx] (p-\i-2) at (.4,0) {};
        \node[vtx] (p-\i-3) at (.8,0) {};
        \node[vtx] (p-\i-4) at (1.2,0) {};
        \node[ultra thick,rect,inner ysep=10pt,fit=(p-\i-1)(p-\i-2)(p-\i-3)(p-\i-4)] (R-\i) {};
        \node[font=\footnotesize] at (0, -.5) {$P_1^{\i}$};
        \foreach \j in {1,...,4} {\draw (a-\i) -- (p-\i-\j);}
        \draw (p-\i-1) -- (p-\i-2);
        \draw (p-\i-2) -- (p-\i-3);
        \draw (p-\i-3) -- (p-\i-4);
        \draw (p-\i-1) edge[bend left=20] (p-\i-3);
        \draw (p-\i-2) edge[bend left=20] (p-\i-4);
        \draw (p-\i-1) edge[bend right=15] (p-\i-4);
        \node[vtx,label=above right:{$q_1^{\i}$}] (q-\i) at (0.6,0.7) {};
        \node[main,label=above right:{$b_1^{\i}$}] (b-\i) at (0.6,1.2) {};
        \draw (b-\i) -- (q-\i);
        \foreach \j in {1,...,4} {\draw (q-\i) -- (p-\i-\j);}
        \end{scope}
      }
      \draw[dashed, thick] (q-1) -- (R-2);
      \draw[dashed, thick] (q-1) -- (R-3.155);
      \draw[dashed, thick] (q-3) -- (R-1.25);
      \draw[dashed, thick] (q-3) -- (R-2);
      \draw[dashed, thick] (q-2) -- (R-1.50);
      \draw[dashed, thick] (q-2) -- (R-3.130);
    \end{tikzpicture}
    \end{center}
    \subcaption{Construction of $G_1$: each $P_1^i$ is a clique of row vertices attached to $a_1^i$, and each $q_1^j$ attaches to $b_1^j$; dashed edges indicate optional adjacencies encoding vectors in $\mathcal{R}_i$.}
    \label{fig:lb:G1}
    \end{subfigure}
  
    \vspace{5ex}

    \begin{subfigure}{.9\linewidth}
        \begin{center}
    \begin{tikzpicture}[
        x=2cm,y=2cm,
        main/.style={circle,draw,fill=black,inner sep=2.4pt},
        vtx/.style={circle,draw,inner sep=2pt},
        rect/.style={draw,rounded corners,inner sep=4pt},
        every node/.style={font=\scriptsize},
        >={Latex[length=3mm,width=2mm]}
      ]
      \foreach \row/\i in {0/1, 1.6/2, 3.2/3} {
        \begin{scope}[shift={(\row,0)}]
        \node[main,label=below left:{$a_2^{\i}$}] (a2-\i) at (0.6,-0.5) {};
        \node[vtx, label=below left:{$p_2^{\i}$}] (p2-\i) at (0.6,0) {};
        \node[vtx] (q2-\i-1) at (0,0.7) {};
        \node[vtx] (q2-\i-2) at (.4,0.7) {};
        \node[vtx] (q2-\i-3) at (.8,0.7) {};
        \node[vtx] (q2-\i-4) at (1.2,0.7) {};
        \node[ultra thick,rect,inner ysep=8pt,fit=(q2-\i-1)(q2-\i-2)(q2-\i-3)(q2-\i-4)] (R2-\i) {};
        \node[font=\footnotesize] at (0, 1.15) {$Q_2^{\i}$};
        \draw (a2-\i) -- (p2-\i);
        \node[main,label=above right:{$b_2^{\i}$}] (b2-\i) at (0.6,1.4) {};
        \foreach \j in {1,...,4} {\draw (b2-\i) -- (q2-\i-\j);}
        \foreach \j in {1,...,4} {\draw (p2-\i) -- (q2-\i-\j);}
        \end{scope}
      }
      \draw[dashed, thick] (p2-1) -- (R2-2);
      \draw[dashed, thick] (p2-1) -- (R2-3.205);
      \draw[dashed, thick] (p2-3) -- (R2-1.335);
      \draw[dashed, thick] (p2-3) -- (R2-2);
      \draw[dashed, thick] (p2-2) -- (R2-1.315);
      \draw[dashed, thick] (p2-2) -- (R2-3.225);
    \end{tikzpicture}
    \end{center}
        \subcaption{Construction of $G_2$: each $Q_2^i$ is an independent set attached to $p_2^i$ and $b_2^i$; dashed edges indicate optional adjacencies from $p_2^i$ to $Q_2^j$ encoding vectors in $\mathcal{C}_j$.}
        \label{fig:lb:G2}
    \end{subfigure}
    \caption{Constructions. Dashed lines indicate that edge may or may not exist.}
    \label{fig:mcis-lb:g0}
\end{figure}
}
The graph $G_1$ has $k$ vertices $q_1^1,\dots, q_1^k$ in the cluster vertex deletion set.
It also has $k$ cliques $P_1^1,\dots, P_1^k$, where each vertex $p_1^{i,v}\in P_1^i$ corresponds to a vector $v\in \mathcal{R}_i$.
We add an edge between $p_1^{i,v}$ and $q_1^j$ if $v[j]=1$.
The graph $G_2$ is defined analogously with the roles of $p$ and $q$ interchanged, and with independent sets in place of cliques.
Specifically, $G_2$ has $k$ vertices $p_2^1,\dots, p_2^k$ in the cluster vertex deletion set and $k$ independent sets $Q_2^1,\dots, Q_2^k$ corresponding to $\mathcal{C}_j$, and we connect $p_2^{i}$ and $q_2^{j,v}$ by an edge if $v[i]=1$.

With additional care, using the $O(k)$-vertex base gadget in \Cref{fig:lb:G0}, we can ensure that $q_1^j$ (resp.\ $p_2^i$) must be used in the common subgraph and corresponds to one of the vertices $q_2^j\in Q_2^j$ (resp.\ $p_1^i\in P_1^i$), which selects a vector $c_j\in \mathcal{C}_j$ (resp.\ $r_i\in \mathcal{R}_i$).
Moreover, the edges of $G_1$ between $\{p_1^i\}_{i\in [k]}$ and $\{q_1^j\}_{j\in [k]}$ must be isomorphic to the edges of $G_2$ between $\{p_2^i\}_{i\in [k]}$ and $\{q_2^j\}_{j\in [k]}$.
In particular,
$r_i[j]=1\Longleftrightarrow \{p_1^{i},q_1^j\}\in E(G_1)\Longleftrightarrow \{p_2^{i},q_2^j\}\in E(G_2)\Longleftrightarrow c_j[i]=1$,
which is exactly the \textsc{LCBM} constraint.

\versionparagraph{Overview for Theorem~\ref{thm:tmcislb}}

Finally, we outline our NP-hardness proof for \textsc{3-MCIS}.
While for two graphs, \textsc{MCIS} reduces to a bipartite matching problem over clique components, for three graphs, the analogous construction yields a three-way matching on triples of cliques, which captures the core difficulty.

Consider \textsc{3-MCIS} on graphs $G_i$ ($i\in [3]$) with cluster vertex deletion set $X_i=\{x_i^1, x_i^2\}$ of size $2$, and let $\mathcal{C}_i$ be the set of cliques of $G_i-X_i$.
Let $H$ be a hypothetical solution.
In this overview we restrict to solutions where $H$ has cluster vertex deletion set $X=\{x^1,x^2\}$ with $x^j$ mapped to $x_i^j$ under the embedding into each $G_i$; the full construction enforces this.
In this setting, the remaining task becomes a restricted maximum-weight three-dimensional matching problem: the vertex set is $\mathcal{C}_1 \dotcup \mathcal{C}_2 \dotcup \mathcal{C}_3$, and a triple $(C_1,C_2,C_3)$ has weight equal to the largest clique that can be embedded into all three $C_i$.
Although weights are not arbitrary, the structure is still expressive enough to encode \textsc{3-SAT}.

To characterize the weights, note that vertices in a clique component of $G_i-X_i$ fall into $4$ classes according to adjacency to $X_i$.
We represent each clique $C_i$ by a vector $v_{C_i}=(v_{C_i}^1,\dots, v_{C_i}^{4})$, where $v_{C_i}^j$ counts vertices in the $j$-th class.
Then the weight of $(C_1,C_2,C_3)$ is
\(
\sum_{j=1}^{4}\min(v_{C_1}^j,v_{C_2}^j,v_{C_3}^j),
\)
the size of the entrywise minimum.
Thus \textsc{3-MCIS} can be viewed as a matching problem on three sets of points in $4$ dimensions.

We then reduce a variant of \textsc{3-SAT} to this geometric matching problem; see \Cref{fig:variable-gadget}.
Informally, the first two coordinates enforce consistency of variable assignments by arranging variable gadgets ``anti-diagonally'', and the last two coordinates enforce clause satisfaction with a similar anti-diagonal alignment, so constraints from different variables do not interfere.
\shortonly{%
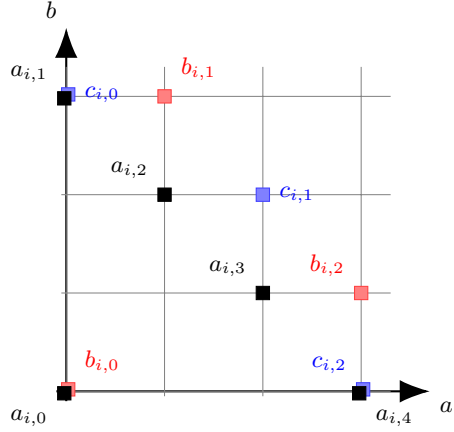
\begin{figure}
    \centering
    \begin{tikzpicture}[scale=0.13, every node/.style={font=\small}, >={Latex[length=4mm,width=3mm]}]
    \draw[thick,->] (-1,0) -- (37,0) node[below right] {$a$};
    \draw[thick,->] (0,-1) -- (0,37) node[above left] {$b$};
    \draw[step=10,help lines] (-0.5,-0.5) grid (33,33);

    \tikzset{
        V1/.style={draw=black,fill=black,inner sep=2pt, minimum size=5pt},
        V2/.style={draw=red!75,fill=red!50,rectangle,inner sep=2pt, minimum size=5pt},
        V3/.style={draw=blue!75,fill=blue!50,rectangle,inner sep=2pt, minimum size=5pt}
    }

    \node[V2,label={[red]above right:$b_{i,0}$}] at ($(0,0)+(0.2,0.2)$) {};
    \node[V2,label={[red]above right:$b_{i,1}$}] at (10,30) {};
    \node[V2,label={[red]above left:$b_{i,2}$}] at (30,10) {};

    \node[V3,label={[blue]right:$c_{i,0}$}] at ($(0,30)+(0.2,0.2)$) {};
    \node[V3,label={[blue]right:$c_{i,1}$}] at (20,20) {};
    \node[V3,label={[blue]above left:$c_{i,2}$}] at ($(30,0)+(0.2,0.2)$) {};

    \node[V1,label={below left:$a_{i,0}$}] at ($(0,0)-(0.2,0.2)$) {};
    \node[V1,label={above left:$a_{i,1}$}] at ($(0,30)-(0.2,0.2)$) {};
    \node[V1,label={above left:$a_{i,2}$}] at (10,20) {};
    \node[V1,label={above left:$a_{i,3}$}] at (20,10) {};
    \node[V1,label={below right:$a_{i,4}$}] at ($(30,0)-(0.2,0.2)$) {};

    \end{tikzpicture}
    \caption{Variable gadget positions in the $(x,y)$-plane. Only the first two coordinates are shown; the last two coordinates are $N^2$ (or $(p_{i,t}N,(n-p_{i,t})N)$), and the shift $v_i$ applies in all four dimensions.}
    \label{fig:variable-gadget}
\end{figure}
}

\shortonly{%
\versionparagraph{Scope of this version.}
In this short version, we only present the algorithmic result for \textsc{ISI}. The results for \textsc{MCIS}, the ETH lower bounds, and the hardness results for \textsc{3-MCIS} are deferred to the full version.
}

\ifshortversion
\section{Algorithm for ISI}
\else
\section{FPT algorithms for ISI and MCIS}
\fi

\ifshortversion
We now develop the \textsc{ISI} algorithm, using the matching subroutine from \Cref{prop:mpm}.
\else
We begin by presenting the algebraic algorithm for \textsc{Weighted Exact Multicolored Matching}, which we use as a subroutine in our \textsc{ISI} and \textsc{MCIS} algorithms (\Cref{ssec:matching}).
We then develop the \textsc{ISI} algorithm (\Cref{ssec:reduction}) and the \textsc{MCIS} algorithm (\Cref{ssec:mcis}).
\fi

\longonly{
\subsection{Weighted Exact Multicolored Matching} \label{ssec:matching}

In \textsc{Weighted Exact Multicolored Matching}, the input is an undirected multi-graph $G=(V,E)$, an edge-coloring $c:E\to\{0,1,\dots,k\}$, nonnegative integer weights $w:E\to\mathbb{N}$, and an integer $t\in\mathbb{N}$.
The task is to find a matching $M\subseteq E$ such that $|\{e\in M:\,c(e)=i\}|=1$ for each $i\in\{1,\dots,k\}$ and $w(M)=\sum_{e\in M} w(e)\ge t$.
When $w(e)=1$ for all edges $e\in E(G)$ and $t=n/2$, we call the problem \textsc{Exact Multicolored Matching}.
This problem captures \textsc{Rainbow Matching}: given an edge-colored graph and an integer $k$, \textsc{Rainbow Matching} asks whether there is a matching of size $k$ whose edges all receive distinct colors.
It is the special case of our formulation with $w \equiv 1$, $t = k$, and colors $\{1, \dots, k\}$; if $n > 2k$, we can pad the instance by adding $n - 2k$ new vertices and connecting each of them to all original vertices with color-$0$ edges so the padding does not interact with the original instance.

Bannach et al.~\cite{DBLP:conf/mfcs/BannachBMMLRS20} studied a related variant.
Given a multi-graph $G=(V,E)$, an edge-coloring $c:E\to\{0,1,\dots,k\}$, a weight function $w:E\to\mathbb{N}$, and an integer $t\in\mathbb{N}$, the task is to find a perfect matching $M$ with $w(M)\ge t$ such that every color from $\{1, \dots, k\}$ is used at least once.
Colors may repeat across the chosen edges, which is crucial for their application.
They gave a randomized $O^*(2^k)$-time algorithm for this problem.

In this section, we will prove the following:
\wemm*

The classic \textsc{Exact Matching} problem asks, for a graph $G=(V,E)$, a distinguished set $R\subseteq E$ of ``red'' edges, and an integer $k\in \mathbb{N}$, whether there exists a perfect matching $M$ with $|M\cap R|=k$.
It admits a randomized polynomial-time algorithm via the isolation lemma of Mulmuley et al.~\cite{DBLP:journals/combinatorica/MulmuleyVV87}, but the existence of a deterministic (even parameterized) polynomial-time algorithm remains a long-standing open question (see, e.g., \cite{DBLP:conf/stacs/Maalouly23,MurakamiY25}).
The following shows that derandomizing \Cref{prop:mpm} would resolve a difficult research problem.
%\todo{The last part of this paragraph should appear after \Cref{prop:wemmtoem}.}

\wemmtoem*

Before proving \Cref{prop:mpm}, we recall the algebraic tools we rely on.
Throughout we work over a finite field $\mathbb F$ of characteristic 2 with $n^{\Theta(1)}$ elements.
Arithmetic operations over $\mathbb{F}$ take $\log^{O(1)} n$ time. 

Given a set $S$, let $\Pi(S)$ denote the set of perfect matchings on $S$ (i.e., the collection of partitions of $S$ into sets of size exactly two). 
Given a skew-symmetric matrix $A$ (i.e., a matrix that equals $A = -A^\top$ the negative of its transpose) with rows and columns indexed by a set $S$, 
its Pfaffian is defined to be 
    \begin{equation}
    \label{eq:Pfaffdef}
    \Pf A = \sum_{M\in \Pi(S)} \sgn(M)\prod_{ \{u,v\}\in M} A[u,v]
    \end{equation}
for a function $\sgn \colon\Pi(S)\to \{-1,1\}$ whose definition is not relevant here because we work over a characteristic-2 field (see e.g., \cite[Section 7.3.2]{murota1999matrices} for the definition).

An arithmetic circuit is a way of constructing a polynomial by starting with its input variables, and iteratively building up more complicated expressions by using the standard  operations of addition, multiplication, and division.
The size of an arithmetic circuit is the total number of arithmetic operations used to construct its final output polynomial in this way. 
Our algorithms use the well known fact that Pfaffians admit polynomial-size arithmetic circuits that only use addition and multiplication operations \cite{Rote2001}. 

\begin{proposition}[Pfaffian Construction]
    \label{prop:Pfaffian}
    In $n^{O(1)}$ time we can construct a division-free arithmetic circuit of $n^{O(1)}$ size over $\mathbb F$ whose inputs are indeterminate entries of an $n\times n$ skew-symmetric matrix $A$, and whose output is the polynomial $\Pf A$.
\end{proposition}

While \textsc{Exact Matching} is typically solved via the isolation lemma
\cite{DBLP:journals/combinatorica/MulmuleyVV87},
we present an alternative (and more natural in our opinion) approach via the polynomial identity testing, also known as the Schwartz-Zippel lemma.
It is proven for example in \cite[Theorem 7.2]{Motwani1995}.

\begin{proposition}[Polynomial Identity Testing]
    \label{prop:PIT}
    Let $P$ be a nonzero polynomial over a finite field $\mathbb F$ of degree at most $d$.
    If each variable of $P$ is assigned an independent, uniform random value from $\mathbb F$,
    then the corresponding evaluation of $P$ is nonzero with probability at least $1-d/|\mathbb F|$.
\end{proposition}

Following recent works (e.g., \cite{DBLP:conf/stacs/AkmalK25,eiben2025determinantal}), we build the required polynomial using two algebraic transformations: interpolation and inclusion-exclusion. 

\begin{lemma}[Interpolation]
  \label{lemma:interpolation}
  Let $P(z)$ be a polynomial of degree $n - 1$ over a field $\F$.
  Suppose that $P(z_i) = p_i$ for distinct $z_1, \cdots, z_n \in \F$.
  By the Lagrange interpolation, 
  \[
    P(z)
    = \sum_{i \in [n]} p_i \prod_{j \in [n] \setminus \{ i \}} \frac{z - z_j}{z_i - z_j}.
  \]
  Thus, given $n$ evaluations $p_1, \dots, p_n$ of $P(z)$, the coefficient of $z^t$ in $P(z)$ for every $t \in [n]$ can be computed in polynomial time.
\end{lemma}

\begin{lemma}[Inclusion-exclusion \cite{Wahlstrom13STACS}]
  \label{lemma:inclusion-exclusion}
  Let $P(Y)$ be a polynomial over a set of variables $Y = \{ y_1, \cdots, y_n \}$ and a field of characteristic $2$.
  For $T \subseteq [n]$, $Q_T$ be a polynomial identical to $P$ except that the coefficients of monomials not divisible by $\prod_{i \in T} y_i$ are zero.
  Then, $Q_T = \sum_{I \subseteq T} P_{-I}$, where $P_{-I}(y_1, \cdots, y_n) = P(y_1', \cdots, y_n')$ for $y_i' = y_i$ if $i \notin I$ and $y_i' = 0$ otherwise. 
\end{lemma}

Now we are ready to give proofs of \Cref{prop:mpm}:

\wemm*

\begin{proof}
  If $n < 2k$, the instance is trivially a no-instance, so assume $n \ge 2k$.
  Add $n-2k$ fresh vertices $D$, make $D$ a color-$0$, weight-$0$ clique, and connect every $d \in D$ to every $v \in V$ with a color-$0$, weight-$0$ edge, obtaining $G'$.
  Any matching that uses exactly one edge of each color in $[k]$ can be extended to a perfect matching of $G'$ by matching every vertex in $D$ to a distinct unused vertex of $V$ via a color-$0$ edge.
  Any perfect matching of $G'$ with $k$ non-zero colors projects back to a solution in $G$ with weights preserved.
  Thus it suffices to find a perfect matching in $G'$ that uses exactly $k$ non-zero colors and maximizes the total weight.

  Introduce edge variables $x_e$ for $e \in E'$ (collectively $X$), a color-$0$ variable $y_0$, color variables $y_1,\dots,y_k$ (collectively $Y$), and a weight variable $z$.
  For an edge set $S \subseteq E'$, define
  \[ P_S(X, y_0, Y, z) = \left( \prod_{e \in S} x_e \right) \cdot  \left( \prod_{e \in S} y_{c(e)} \right) \cdot z^{w(S)}, \]
  i.e., $P_S$ encodes the chosen edges, their colors, and the total weight as exponents.
  For a family of edge sets $\mathcal{S}$, define
  \[ P_{\mathcal{S}} = \sum_{S \in \mathcal{S}} P_S(X, y_0, Y, z). \]

  Let $A$ be the skew-symmetric matrix indexed by $V'$ with
  \[
      A[u,v] = \sum_{e \colon e \text{ joins } u \text{ and } v} x_e \cdot y_{c(e)} \cdot z^{w(e)} .
  \]
  By the definition of the Pfaffian,
  \( \Pf A = P_\mathcal{M}(X, y_0, Y, z), \)
  where $\mathcal{M}$ is the family of all perfect matchings in $G'$.

  Let $\mathcal{M}'$ be the perfect matchings in $G'$ that contain exactly $n-2k$ color-$0$ edges.
  Every perfect matching has $|V'|/2 = n-k$ edges, and its $y_0$-exponent encodes the number of color-$0$ edges.
  Interpolating in $y_0$ via \Cref{lemma:interpolation}, we extract the coefficient of $y_0^{n-2k}$ in $\Pf A$ with $n-k+1$ evaluations, yielding $P_{\mathcal{M'}}(X, 1, Y, z)$.
  Let $\mathcal{M}^{\star}$ be the perfect matchings in $G'$ with exactly $n-2k$ color-$0$ edges and one edge of each color $i \in [k]$.
  The polynomial $P_{\mathcal{M}^{\star}}(X, 1, 1, z)$ is the coefficient of $\prod_{i \in [k]} y_i$ in $P_{\mathcal{M'}}(X, 1, Y, z)$.
  Inclusion-exclusion (\Cref{lemma:inclusion-exclusion}) retrieves this coefficient from $2^k$ evaluations of $P_{\mathcal{M'}}(X, 1, Y, z)$.

  Let $p(X, z) = P_{\mathcal{M}^{\star}}(X, 1, 1, z)$.
  Because every edge has weight at most $W$ and each perfect matching has $n-k$ edges, the degree of $p$ in $z$ is at most $(n-k) W$.
  For each $t'$, the coefficient of $z^{t'}$ in $p$ is the sum of $\prod_{e \in M} x_e$ over all $M \in \mathcal{M}^{\star}$ of total weight exactly~$t'$.
  Therefore, the maximum weight is the largest $t'$ such that the coefficient of $z^{t'}$ in $p$ is nonzero.
  By \Cref{lemma:interpolation}, we recover all coefficients of $z^{t'}$ in $p$ from $(n - k)W + 1$ evaluations of $p(X, z)$.
  Assigning each $x_e$ uniformly at random from $\F$, 
  we evaluate each coefficient. 
  By \Cref{prop:PIT}, this is correct with probability at least $1 - 1/\operatorname{poly}(n)$.
\end{proof}

Next, we prove \Cref{prop:wemmtoem}.
We use perfect hash families to derandomize color-coding~\cite{DBLP:journals/jacm/AlonYZ95}.
A family $\mathcal{H}$ of functions $h \colon [n] \to [k]$ is an \emph{$(n,k)$-perfect hash family} if for every $S \subseteq [n]$ with $|S| = k$ there is an $h \in \mathcal{H}$ that is injective on $S$.
Naor, Schulman, and Srinivasan~\cite{DBLP:conf/focs/NaorSS95} gave a deterministic construction of size $e^{k + O(\log^2 k)} \log n$ computable in time $e^{k  + O(\log^2 k)} n \log n$.

\wemmtoem*

\begin{proof}
  Consider an \textsc{Exact Matching} instance $(G, R, k)$ with red edge set $R \subseteq E(G)$.
  For a function $h \colon R \to [k]$, build an \prob{Exact Multicolored Matching} instance on $G$ by coloring $e \in R$ with $h(e)$ and $e \in E(G) \setminus R$ with color $0$.
  If $M$ is a perfect matching with exactly $k$ red edges, then for any $h$ that is injective on $M \cap R$, the colored instance has an exact multicolored matching (one edge of each color $1,\dots,k$, remaining edges of color $0$).

  Let $\mathcal{H}$ be an $(|R|, k)$-perfect hash family from \cite{DBLP:conf/focs/NaorSS95} of size $e^{k + O(\log^2 k)} \log |R|$.
  For each $h \in \mathcal{H}$, construct the colored instance and run the deterministic $O^*(f(k))$-time algorithm for \prob{Exact Multicolored Matching}.
  If the original instance is a yes-instance, some $h \in \mathcal{H}$ is injective on $M \cap R$, so the algorithm accepts.
  Conversely, an accepting run produces a multicolored matching using exactly one edge of each color $1,\dots,k$, which are all red, hence a perfect matching $M$ with $|M \cap R| = k$.
  The overall running time is $e^{k + O(\log^2 k)} \log |R| \cdot f(k) \cdot |E|^{O(1)}$, giving a deterministic FPT algorithm for \textsc{Exact Matching}.
\end{proof}
}

\ifshortversion
\subsection{Main Reduction} \label{ssec:reduction}
\else
\subsection{Induced Subgraph Isomorphism} \label{ssec:reduction}
\fi

\ifshortversion
In this section, we give a randomized $O^*(k^{O(k)})$-time algorithm for \prob{ISI}.
\else
In this section, we will give an FPT algorithm for \prob{ISI}:
\isifpt*
\fi
We proceed as follows.
First, we define an extension variant and show that it is polynomial-time solvable when the remaining graphs are cliques.
Next, we show how a minimal cluster vertex deletion set of the pattern can be placed within a cluster vertex deletion set of the host.
Finally, we reduce the resulting extension instance to \textsc{Exact Multicolored Matching} (Lemma~\ref{lem:target}), yielding the claimed FPT algorithm.

\versionparagraph{Extension problem}

We define an auxiliary problem \textsc{Induced Subgraph Isomorphism Extension (ISIE)} as follows:
Given graphs $G_1$ (host), $G_2$ (pattern), and a partial induced embedding $\alpha \colon Z_2 \to Z_1$ with $Z_1 \subseteq V(G_1)$ and $Z_2 \subseteq V(G_2)$, decide whether there exists an induced embedding $\phi \colon V(G_2)\to V(G_1)$ that extends $\alpha$ (i.e., $\phi|_{Z_2}=\alpha$).

We give a polynomial-time algorithm for the case where $G_i - Z_i$ are cliques. This will be used in the reduction to \textsc{Exact Multicolored Matching}. 
The algorithm relies on a signature argument: we bucket vertices by their adjacency pattern to the fixed set $Z_i$ and compare the resulting counts.
\longonly{We will reuse this signature argument in the algorithm for \textsc{MCIS} as well.}

% We define the \textsc{Maximum Common Induced Subgraph Extension (MCISE)} problem as follows.
% Given graphs $G_1$, $G_2$, and $H'$, with partial induced embeddings $\alpha_i \colon Z_{H} \to Z_i$, where $Z_{H} \subseteq V(H')$ and $Z_{i} \subseteq V(G_i)$, for $i = 1, 2$, and an integer $t \in \mathbb{N}$, the goal is to find a $t$-vertex common induced subgraph $H$ extending $\alpha_i$, i.e., $V(H) \supseteq V(H')$ and there are induced embeddings $\phi_{i} \colon V(H) \to V(G_i)$ extending $\alpha_i$ for $i = 1, 2$.

\begin{lemma} \label{lem:clique-to-clique}
    \textsc{ISIE} is polynomial-time solvable when $G_i - Z_i$ for $i=1,2$ are both cliques.
\end{lemma}
\begin{proof}
Let $r := |Z_1| = |Z_2|$ and order $Z_2 = (z_2^1,\dots,z_2^r)$; write $z_1^j := \alpha(z_2^j) \in Z_1$.
First check that $\alpha$ is an isomorphism between $G_2[Z_2]$ and $G_1[Z_1]$.
If this fails, no extension exists.

For $i\in\{1,2\}$ and $u\in V(G_i)\setminus Z_i$, define its $Z_i$-signature by
\[
\sigma_i(u) \in \{0,1\}^r,\qquad (\sigma_i(u))_j := \mathbf{1}[\{u, z_i^j\}\in E(G_i)].
\]
Bucket vertices by signature $t \in \{0,1\}^r$. We claim that an extension exists if and only if for every $t$,
\[
|\{v \in V(G_2)\setminus Z_2 : \sigma_2(v)=t\}|
\ \le\
|\{u \in V(G_1)\setminus Z_1 : \sigma_1(u)=t\}|.
\]
If the inequalities hold, greedily assign each $v$ of type $t$ to a distinct $u$ of the same type $t$, obtaining an injective map $\phi$ on $V(G_2)\setminus Z_2$.
By construction, $\phi$ preserves all adjacencies/non-adjacencies to $Z_2$, and since $G_2 - Z_2$ and $G_1 - Z_1$ are cliques, it preserves adjacencies within $V(G_2)\setminus Z_2$.
This yields an induced embedding extending $\alpha$.

Conversely, any induced embedding extending $\alpha$ must map each $v \in V(G_2)\setminus Z_2$ to some $u \in V(G_1)\setminus Z_1$ with the same signature (to preserve adjacencies to all $z_2^j$), implying the counting conditions above. 

The counts can be computed in polynomial time by sorting signature vectors. Thus, the algorithm runs in polynomial time.
\end{proof}

\versionparagraph{Mapping cluster deletion set}
We show that within an induced embedding, $G_{2}$ has a minimal cluster vertex deletion set whose image is contained in a cluster vertex deletion set $X_{1}$ of $G_{1}$:

\begin{lemma} \label{lem:modulator-to-modulator}
Let $G_1, G_2$ be graphs, and let $\phi \colon V(G_{2}) \to V(G_{1})$ be an induced embedding from $G_{2}$ to $G_{1}$.
If $X_{1}$ is a cluster vertex deletion set of $G_{1}$,
then there is a minimal cluster vertex deletion set $X_{2}$ of $G_{2}$ such that $\phi(X_{2}) \subseteq X_{1}$.
\end{lemma}

\begin{proof}
Since $G_{2}$ and $G_{1}[\phi(V(G_{2}))]$ are isomorphic
and $X_{1}$ is a cluster vertex deletion set of $G_{1}$,
the set $\phi(V(G_{2})) \cap X_{1}$ is a cluster vertex deletion set of $G_{1}[\phi(V(G_{2}))]$.
This implies that $\phi^{-1}(\phi(V(G_{2})) \cap X_{1})$ is a cluster vertex deletion set of $G_{2}$.
Let $X_{2}$ be a subset of $\phi^{-1}(\phi(V(G_{2})) \cap X_{1})$ that is a minimal cluster vertex deletion set of $G_{2}$.
This set $X_{2}$ shows the lemma as $\phi(X_{2}) \subseteq \phi(V(G_{2})) \cap X_{1} \subseteq X_{1}$.
\end{proof}

Our algorithm proceeds as follows.
Fix a cluster vertex deletion set $X_1$ of $G_1$ with $|X_1|\le k$.
Enumerate all minimal cluster vertex deletion sets $X_2$ of $G_2$ of size at most $|X_1|$ using the standard three-way branching on a $P_3$, which takes $O^*(3^{|X_1|})$ time.
By \Cref{lem:modulator-to-modulator}, if an induced embedding $\phi$ exists, then some enumerated $X_2$ satisfies $\phi(X_2)\subseteq X_1$.

For each such $X_2$, we try all $k^{O(k)}$ options of a set $Z_1\subseteq X_1$ of size $|X_2|$ and a bijection $\alpha\colon X_2\to Z_1$ that is an induced embedding between $G_2[X_2]$ and $G_1[Z_1]$ (rejecting a guess otherwise).
We further guess, from $2^{O(k)}$ candidates, a set $A = \phi(V(G_2)\setminus X_2)\cap X_1$, which is the part of the image that lands in $X_1\setminus Z_1$.
Discard the guess if $G_1[A]$ is not a cluster graph.
This is safe because $G_1[A]$ is a cluster graph: $A$ lies in $\phi(V(G_2)\setminus X_2)$ and $G_2 - X_2$ is a cluster graph.
We may assume, for the sake of analysis, that $A$ is inclusion-wise minimal, i.e., there is no induced embedding with $\phi(V(G_2) \setminus X_2) \cap X_1 \subset A$.

Fix such a choice of $X_2$, $Z_1$, $\alpha$, and $A$. It remains to find an induced embedding extending~$\alpha$.

\begin{lemma} \label{lem:target}
With the guesses above (the partial induced embedding $\alpha \colon X_2 \to Z_1$ and the set $A \subseteq X_1 \setminus Z_1$), we can in randomized $O^*(2^k)$ time either output an induced embedding $\phi \colon V(G_2) \to V(G_1)$ extending $\alpha$ with $A = \phi(V(G_2)\setminus X_2)\cap X_1$, or correctly conclude that none exists.
\end{lemma}

We first show how \Cref{lem:target} implies \Cref{thm:isi-fpt}, and then prove the lemma.

\begin{proof}[Proof of \Cref{thm:isi-fpt} assuming \Cref{lem:target}]
  %Fix a cluster vertex deletion set $X_1$ of $G_1$ with $|X_1|\le k$.
  %Enumerate all minimal cluster vertex deletion sets $X_2$ of $G_2$ of size at most $|X_1|$ using the standard three-way branching on a $P_3$, which takes $O^*(3^{|X_1|})$ time.
  %By \Cref{lem:modulator-to-modulator}, if an induced embedding $\phi$ exists, then some enumerated $X_2$ satisfies $\phi(X_2)\subseteq X_1$.

  %For each such $X_2$ we try all choices of a set $Z_1\subseteq X_1$ of size $|X_2|$ and a bijection $\alpha\colon X_2\to Z_1$ that is an induced embedding between $G_2[X_2]$ and $G_1[Z_1]$ (rejecting a guess otherwise).
  %Let $A := X_1 \setminus Z_1$ and $B := V(G_1)\setminus X_1$; discard the guess if $G_1[A]$ is not a cluster graph.
  Applying \Cref{lem:target} for each guess $A$ (and $B = V(G_1)\setminus X_1$) runs in randomized $O^*(2^{|A|}) \le O^*(2^{|X_1|})$ time.
  When an induced embedding $\phi$ exists, the enumeration produces the required $X_2$, and one branch chooses $Z_1=\phi(X_2)$ and the corresponding $\alpha$. Then \Cref{lem:target} reconstructs $\phi$.
  Conversely, if \Cref{lem:target} outputs an induced embedding for some guess, it is by definition an induced embedding of $G_2$ into $G_1$ extending $\alpha$, hence a valid solution to the original instance.
  The total running time is $O^*(3^{|X_1|}\cdot |X_1|^{O(|X_1|)}\cdot 2^{|X_1|}) = O^*(k^{O(k)})$, as claimed.
\end{proof}

We now prove \Cref{lem:target} by reducing the remaining instance to \prob{Exact Multicolored Matching}, which yields a randomized $O^*(2^{k})$-time algorithm.

\begin{proof}[Proof of \Cref{lem:target}]

Let $B = V(G_1)\setminus X_1$.
Since $G_1[A]$, $G_1[B]$, and $G_2 - X_2$ are cluster graphs, let $\mathcal{S}_1$, $\mathcal{C}_1$, and $\mathcal{C}_2$ denote their clique components, respectively.
Let $\ell =|\mathcal{S}_1|$. Note that $\ell \le |A| \le k$.
We now give a randomized $O^*(2^{\ell})$-time algorithm.

We first apply a reduction rule on $G_1$.
For $v \in B$, let $N_A(v) = N_{G_1}(v) \cap A$.

\begin{claim}\label{claim:rr}
For $v \in B$,
if either
\begin{enumerate}
    \item there exists a clique $S_1 \in \mathcal{S}_1$ and vertices $w,w'\in S_1$ with $w\notin N_A(v)$ and $w'\in N_A(v)$, or
    \item there exist $w\in N_A(v)\cap S_1$ and $w'\in N_A(v)\cap S_1'$ for two distinct cliques $S_1,S_1'\in\mathcal{S}_1$,
\end{enumerate}
then $v \notin \phi(V(G_2))$ for every induced embedding $\phi$ with $A \subseteq \phi(V(G_2))$.
\end{claim}

\begin{proof}[Proof of \Cref{claim:rr}]
Suppose, for contradiction, that $v \in \phi(V(G_2))$.
Its preimage lies in $V(G_2)\setminus X_2$.
Since any $w \in A$ also lies in $\phi(V(G_2))$, its preimage also lies in $V(G_2)\setminus X_2$.

In case (1), pick $w,w'\in S_1$ with $w\notin N_A(v)$ and $w'\in N_A(v)$.
Since $S_1$ is a clique, $ww'\in E(G_1)$; moreover $w'v\in E(G_1)$ and $wv\notin E(G_1)$.
Thus $(w,w',v)$ induces a $P_3$ inside $G_1[\phi(V(G_2)\setminus X_2)]$, so the corresponding preimages induce a $P_3$ in $G_2 - X_2$, contradicting that $G_2 - X_2$ is a cluster graph.

In case (2), take $w\in S_1$ and $w'\in S_1'$ with $S_1\neq S_1'$ and $w,w'\in N_A(v)$.
Then $wv, w'v\in E(G_1)$ while $ww'\notin E(G_1)$ (different cliques), so $(w,v,w')$ induces a $P_3$ in $G_1[\phi(V(G_2)\setminus X_2)]$, yielding the same contradiction.
\end{proof}

We exhaustively delete vertices $v\in B$ that satisfy one of the conditions of \Cref{claim:rr}.
Then, for every $v\in B$, either $N_A(v)=\emptyset$ or there exists a clique $S_1\in\mathcal{S}_1$ with $N_A(v)=S_1$.
For $S_1 \in \mathcal{S}_1$ and $C_1 \in \mathcal{C}_1$, define $N_{C_1}(S_1) = \{ v \in C_1 \mid N_A(v) = S_1 \}$, and $C_1^{\circ} = C_1 \setminus \bigcup_{S_1 \in \mathcal{S}_1} N_{C_1}(S_1)$.

\begin{claim} \label{claim:clique-map}
    Fix a clique $C_2 \in \mathcal{C}_2$. 
    Then, $\phi(C_2)$ is contained in one of the following sets: (i) $C_1^{\circ}$ for some $C_1 \in \mathcal{C}_1$, (ii) $S_1$ for some $S_1 \in \mathcal{S}_1$, or (iii) $S_1 \cup N_{C_1}(S_1)$ for some $C_1 \in \mathcal{C}_1$ and $S_1 \in \mathcal{S}_1$.
\end{claim}
\begin{proof}[Proof of \Cref{claim:clique-map}]
    Suppose that $\phi(C_2) \cap B$ is contained in $C_1 \in \mathcal{C}_1$ (or $\phi(C_2) \cap B = \emptyset$).

    If $\phi(C_2)$ is disjoint from $A$, then $\phi(C_2)\subseteq C_1$. 
    Since different cliques of $\mathcal{C}_2$ are pairwise anticomplete in $G_2\setminus Z_2$, their images must be pairwise anticomplete in $G_1$, so $\phi(C_2)$ cannot contain a vertex adjacent to $A$ (otherwise it would create edges to the image covering $A$), implying that $\phi(C_2) \subseteq C_1^{\circ}$.

    If $\phi(C_2)$ meets $A$, it must be contained in some $S_1 \in \mathcal{S}_1$ on the $A$-side.
    Since every vertex of $\phi(C_2)\cap B$ must then be adjacent to all of $S_1$, $\phi(C_2)\subseteq S_1 \cup N_{C_1}(S_1)$ for some pair $(S_1,C_1)$.
\end{proof}

We reduce to an instance of \prob{Exact Multicolored Matching} (all edge weights are $1$).
The target matching size is $t := |\mathcal{C}_2|$.
Build an edge-colored bipartite graph $H$ with left side $\mathcal{S}_1 \cup \mathcal{C}_1$, right side $\mathcal{C}_2$, and color set $\{0\}\cup \mathcal{S}_1$ (color $0$ plus one color per $S_1\in\mathcal{S}_1$).
Add edges as follows:
\begin{itemize}
    \item For each $S_1\in\mathcal{S}_1$ and $C_2\in\mathcal{C}_2$, add an edge $(S_1,C_2)$ of color $S_1$ if and only if $C_2$ can be embedded into $S_1$ extending $\alpha$.
    \item For each $C_1\in\mathcal{C}_1$ and $C_2\in\mathcal{C}_2$, add an edge $(C_1,C_2)$ of color $0$ if and only if $C_2$ can be embedded into $C_1^{\circ}$ extending $\alpha$.
    \item For each $C_1\in\mathcal{C}_1$, $S_1\in\mathcal{S}_1$, and $C_2\in\mathcal{C}_2$, add an edge $(C_1,C_2)$ of color $S_1$ if and only if $C_2$ can be embedded into $S_1\cup N_{C_1}(S_1)$ extending $\alpha$.
\end{itemize}

By \Cref{lem:clique-to-clique}, the instance can be constructed in polynomial time.
We prove correctness below, i.e., there is a desired induced embedding if and only if there is a multicolored matching with $|\mathcal{C}_2|$ edges.

($\Rightarrow$)
Let $\phi$ be an induced embedding. We build a matching $M$ in $H$ as follows. For each $C_2 \in \mathcal{C}_2$, by \Cref{claim:clique-map} exactly one of the following holds:
\begin{enumerate}[(i)]
    \item $\phi(C_2)\subseteq C_1^{\circ}$ for some $C_1\in\mathcal{C}_1$: add the edge $(C_1,C_2)$ of color $0$ to $M$.
    \item $\phi(C_2)\subseteq S_1$ for some $S_1\in\mathcal{S}_1$: add the edge $(S_1,C_2)$ of color $S_1$ to $M$.
    \item $\phi(C_2)\subseteq S_1\cup N_{C_1}(S_1)$ for some $S_1\in\mathcal{S}_1$ and $C_1\in\mathcal{C}_1$: add the edge $(C_1,C_2)$ of color $S_1$ to $M$.
\end{enumerate}
Cliques in $\mathcal{C}_2$ are pairwise anticomplete in $G_2\setminus Z_2$, hence their images cannot meet the same target clique in $G_1$.
Consequently, no two edges of $M$ share a left endpoint, so $M$ is a matching in $H$.
For every $S_1\in\mathcal{S}_1$, there is a unique $C_2 \in \mathcal{C}_2$ such that $S_1 \subseteq \phi(C_2)$, and thus $M$ contains exactly one edge of color $S_1$.

($\Leftarrow$)
Let $M$ be a perfect matching in $H$ that uses each color $S_1\in\mathcal{S}_1$ exactly once. For each edge $e\in M$ incident to a right vertex $C_2\in\mathcal{C}_2$, do:
\begin{itemize}
  \item if $e = (S_1,C_2)$ and its color is $S_1$, embed $C_2$ into $S_1$;
  \item if $e = (C_1,C_2)$ and its color is $0$, embed $C_2$ into $C_1^{\circ}$;
  \item if $e = (C_1,C_2)$ and its color is $S_1$, embed $C_2$ into $S_1\cup N_{C_1}(S_1)$.
\end{itemize}
By construction these embeddings extending $\alpha$ indeed exist.
Because $M$ is a matching and uses each color $S_1\in\mathcal{S}_1$ exactly once, no two chosen edges share a left endpoint and no two $C_2$'s are mapped using the same $S_1$ or the same $C_1$; consequently the images of distinct $C_2$'s are vertex-disjoint.

It remains to check that the union is an induced embedding, i.e., there are no unwanted edges across the images of different $C_2$'s:
(i) vertices mapped to $C_1^{\circ}$ have no neighbors in $A$, and
(ii) if an image uses $S_1 \cup N_{C_1}(S_1)$, then since color $S_1$ is used exactly once and the matched $C_1$ is unique, no other image meets $S_1$ or $C_1$, and vertices outside $N_{C_1}(S_1)$ are nonadjacent to $S_1$ by definition of $N_{C_1}(S_1)$. Hence no new edges appear across images, and we obtain an induced embedding $\phi$ extending $\alpha$.

Finally, by the minimality of $A$, we have $A = \phi(V(G_2) \setminus X_2) \cap X_1$.
\end{proof}

\longonly{
\subsection{Maximum Common Induced Subgraph} \label{ssec:mcis}

\begin{sidewaysfigure}
    \begin{tikzpicture}[
    scale=0.98,
    font=\small,
    clique/.style={
        draw, rounded corners=5pt, line width=0.9pt,
        minimum width=2.8cm, minimum height=0.8cm, inner sep=7pt
    },
    vtx/.style={
        draw, ellipse, line width=0.9pt,
        minimum width=10mm, minimum height=6mm, inner sep=0pt
    },
    container/.style={
        draw, rounded corners=10pt, line width=0.9pt, inner sep=6pt
    },
    conn/.style={line width=1.8pt},
    dconn/.style={line width=1.8pt, dashed}
    ]

    % -------------------------
    % J1 cliques (spread out more)
    % -------------------------
    \node[clique] (T10) at (-8,2) {};
    \node[anchor=south] (T10lab) at ($(T10.north)+(0,0.18)$) {$\mathcal T_{1,0}$};
    \node[vtx] (T10a) at ($(T10.south)+(-0.60,0.4)$) {};
    \node[vtx] (T10b) at ($(T10.south)+( 0.60,0.4)$) {};

    \node[clique] (T11) at (-5,2) {};
    \node[anchor=south] (T11lab) at ($(T11.north)+(0,0.18)$) {$\mathcal T_{1,1}$};
    \node[vtx] (T11a) at ($(T11.south)+(-0.60,0.4)$) {};
    \node[vtx] (T11b) at ($(T11.south)+( 0.60,0.4)$) {};

    \node[clique] (T13) at (-2,2) {};
    \node[anchor=south] (T13lab) at ($(T13.north)+(0,0.18)$) {$\mathcal T_{1,3}$};
    \node[vtx] (T13a) at ($(T13.south)+(-0.60,0.4)$) {};
    \node[vtx] (T13b) at ($(T13.south)+( 0.60,0.4)$) {};

    % -------------------------
    % J2 cliques (spread out more)
    % -------------------------
    \node[clique] (T23) at (2,2) {};
    \node[anchor=south] (T23lab) at ($(T23.north)+(0,0.18)$) {$\mathcal T_{2,3}$};
    \node[vtx] (T23a) at ($(T23.south)+(-0.60,0.4)$) {};
    \node[vtx] (T23b) at ($(T23.south)+( 0.60,0.4)$) {};

    \node[clique] (T22) at (5,2) {};
    \node[anchor=south] (T22lab) at ($(T22.north)+(0,0.18)$) {$\mathcal T_{2,2}$};
    \node[vtx] (T22a) at ($(T22.south)+(-0.60,0.4)$) {};
    \node[vtx] (T22b) at ($(T22.south)+( 0.60,0.4)$) {};

    \node[clique] (T20) at (8,2) {};
    \node[anchor=south] (T20lab) at ($(T20.north)+(0,0.18)$) {$\mathcal T_{2,0}$};
    \node[vtx] (T20a) at ($(T20.south)+(-0.60,0.4)$) {};
    \node[vtx] (T20b) at ($(T20.south)+( 0.60,0.4)$) {};

    % -------------------------
    % K cliques
    % -------------------------
    \node[clique] (D1) at (-5,0) {};
    \node[anchor=north] (D1lab) at ($(D1.south)-(0,0.18)$) {$\mathcal D_1$};
    \node[vtx] (D1a) at ($(D1.south)+(-0.60,0.40)$) {};
    \node[vtx] (D1b) at ($(D1.south)+( 0.60,0.40)$) {};

    \node[clique] (D3) at (0,0) {};
    \node[anchor=west] (D3lab) at ($(D3.east)+(0.18,-0.09)$) {$\mathcal D_3$};
    \node[vtx] (D3a) at ($(D3.south)+(-0.60,0.40)$) {};
    \node[vtx] (D3b) at ($(D3.south)+( 0.60,0.40)$) {};

    \node[clique] (D2) at (5,0) {};
    \node[anchor=north] (D2lab) at ($(D2.south)-(-0,0.18)$) {$\mathcal D_2$};
    \node[vtx] (D2a) at ($(D2.south)+(-0.60,0.40)$) {};
    \node[vtx] (D2b) at ($(D2.south)+( 0.60,0.40)$) {};

    \node[clique, minimum width=4.6cm] (D0) at (0,-1.4) {};
    \node[anchor=west] (D0lab) at ($(D0.east)+(0.18,-0.09)$) {$\mathcal D_0$};
    \node[vtx] (D0a) at ($(D0.south)+(-1.65,0.40)$) {};
    \node[vtx] (D0b) at ($(D0.south)+(-0.55,0.40)$) {};
    \node[vtx] (D0c) at ($(D0.south)+( 0.55,0.40)$) {};
    \node[vtx] (D0d) at ($(D0.south)+( 1.65,0.40)$) {};

    % -------------------------
    % Containers (fit includes the label nodes so labels are inside J1/J2/K)
    % -------------------------
    \node[container, fit=(T10)(T11)(T13)(T10lab)(T11lab)(T13lab)] (Jone) {};
    \node[anchor=south] at ($(Jone.north)+(-3,0.20)$) {$J_1 (=\phi_1^{-1}(X_1)\setminus I)$};

    \node[container, fit=(T23)(T22)(T20)(T23lab)(T22lab)(T20lab)] (Jtwo) {};
    \node[anchor=south] at ($(Jtwo.north)+(3,0.20)$) {$J_2(=\phi_2^{-1}(X_2)\setminus I)$};

    \node[container, fit=(D1)(D3)(D2)(D0)(D1lab)(D3lab)(D2lab)(D0lab)] (Kbox) {};
    \node[anchor=north west] at ($(Kbox.south west)+(-0.05,-0.15)$) {$K$};

    \node[container, minimum width=10cm, minimum height=1.2cm] (Ibox) at (0, 5.5) {};
    \node[anchor=south] at ($(Ibox.north)+(0,0.20)$) {$I (=\phi_1^{-1}(X_1)\cap \phi_2^{-1}(X_2))$};

    % -------------------------
    % Thick connections BETWEEN the ellipses (not rectangles)
    % -------------------------
    % D1 <-> T_{1,1}^1
    \draw[conn] (D1a) -- (T11a);
    \draw[conn] (D1b) -- (T11b);

    % D2 <-> T_{2,2}^1
    \draw[conn] (D2a) -- (T22a);
    \draw[conn] (D2b) -- (T22b);

    % D3 <-> T_{1,3}^1
    \draw[conn] (D3a) -- (T13a);
    \draw[conn] (D3b) -- (T13b);

    % D3 <-> T_{2,3}^1 (crossing like the sketch)
    \draw[conn] (D3a) -- (T23a);
    \draw[conn] (D3b) -- (T23b);

    \draw[dconn] ($(Jone.north east)-(3,0)$) .. controls (-1,4.30) and (1,4.30) .. ($(Jtwo.north west)+(3,0)$);

    \draw[dconn] (Jtwo.north) -- (Ibox) -- (Jone.north);
    \draw[dconn] (Ibox.west) .. controls (-13,4.5) and (-13,0) .. (Kbox.west);

    \end{tikzpicture}

    \caption{The structure of $H$: each ellipse is a clique in the indicated family, thick segments denote complete joins between the two cliques, and the dashed lines indicate that edges in between are arbitrary.} \label{fig:mcis}
\end{sidewaysfigure}

In this section, we give an FPT algorithm for \textsc{MCIS}.

\mcisfpt*

Let $X_i$ be the cluster deletion set of $G_i$ for each $i \in \{ 1, 2 \}$.
Consider a hypothetical solution $H$ with induced embeddings $\phi_i \colon V(H) \to V(G_i)$.
In contrast to the ISI algorithm, where the pattern graph is fixed and its modulator can be mapped into the host modulator (\Cref{lem:modulator-to-modulator}), the common subgraph $H$ is unknown and must embed into both $G_1$ and $G_2$.
We therefore guess which vertices of $H$ lie in $X_1$ and in $X_2$ and how these parts connect, and then maximize the remaining size by selecting compatible clique components.
As in the ISI case, this yields a reduction to multicolored matching, but here we need the weighted variant rather than the unweighted one.

\versionparagraph{Initial guesses}
We can guess the subsets $\phi_1(V(H)) \cap X_1$ and $\phi_2(V(H)) \cap X_2$ from $2^{O(k)}$ possible options.
By deleting, for each $i \in \{1,2\}$, the vertices in $X_i \setminus \phi_i(V(H))$, we may henceforth assume that $X_i \subseteq \phi_i(V(H))$.

Let $I := \phi_1^{-1}(X_1) \cap \phi_2^{-1}(X_2)$ denote the set of vertices of $H$ that are mapped into $X_i$ by both embeddings ($i=1,2$).
We can guess the images $\phi_i(I) \subseteq X_i$ from $2^{O(k)}$ options and the restrictions $\phi_i|_{I}$ from $k^{O(k)}$ options. Given $\phi_i|_{I}$, we can check in polynomial time that it is an induced embedding of $H[I]$ into $G_i[\phi_i(I)]$.

Define $X_i' := X_i \setminus \phi_i(I)$ for $i=1,2$, and set $J_i := \phi_i^{-1}(X_i')$ (so $|J_i|=|X_i'|\le k$ by injectivity).
As $I\cup J_i=\phi_i^{-1}(X_i)$, we already know the structure of $H[I\cup J_i]$ for each $i\in \{1,2\}$.
We guess the edges between $J_1$ and $J_2$ in $2^{|J_1|\cdot |J_2|}\leq 2^{O(k^2)}$ time; then, we know the complete structure of $H[I\cup J_1\cup J_2]$.

\versionparagraph{Structure of $H$}

%We define $X_i' := X_i \setminus \phi_i(I)$ for $i=1,2$, and set $J_i := \phi_i^{-1}(X_i')$ (so $|J_i|=|X_i'|\le k$ by injectivity). 
Since $\phi_i$ is an induced embedding and $G_i - X_i$ is a cluster graph, it follows that
\[
H\bigl[V(H)\setminus (I \cup J_i)\bigr] \cong G_i\bigl[\phi_i(V(H))\setminus X_i\bigr]
\]
is also a cluster graph. Equivalently, with $K := V(H)\setminus (I \cup J_1 \cup J_2)$, both $H[K \cup J_1]$ and $H[K \cup J_2]$ are cluster graphs. Note that $J_i = \phi_i^{-1}(X_i)\setminus I = \phi_i^{-1}(X_i)\setminus \phi_{3-i}^{-1}(X_{3-i})$, hence $J_1$ and $J_2$ are disjoint.

Let $\mathcal{D}$ be the family of cliques of $H[K]$ and, for $i\in\{1,2\}$, let $\mathcal{T}_i$ be the family of cliques of $H[J_i]$.
Since $H[K\cup J_i]$ is a cluster graph, every connected component of $H[K\cup J_i]$ is
(i) a clique $D \in \mathcal{D}$, 
(ii) a clique $T\in\mathcal{T}_i$, or
(iii) a union $D\cup T$ with $D\in\mathcal{D}$ and $T\in\mathcal{T}_i$.
Moreover, for each $D\in\mathcal{D}$, there is at most one $T\in\mathcal{T}_i$ such that $D\cup T$ is a connected component (otherwise $H[K \cup J_i]$ is not a cluster graph).

In particular, we can partition $\mathcal{D}$ into four parts \(\mathcal{D} =\mathcal{D}_0\,\dot\cup\,\mathcal{D}_1\,\dot\cup\,\mathcal{D}_2\,\dot\cup\,\mathcal{D}_3\) so that a clique $D\in \mathcal{D}$ is in
\begin{itemize}
    \item $\mathcal{D}_0$ if $D$ itself is a connected component of both $H[J_1\cup K]$ and $H[J_2\cup K]$,
    \item $\mathcal{D}_1$ if $D\cup T_1$ is a connected component of $H[J_1\cup K]$ for some (unique) $T_1\in \mathcal{T}_1$ and $D$ itself is a connected component of $H[J_2\cup K]$,
    \item $\mathcal{D}_2$ if $D$ itself is a connected component of $H[J_1\cup K]$ and $D\cup T_2$ is a connected component of $H[J_2\cup K]$ for some (unique) $T_2\in \mathcal{T}_2$, and
    \item $\mathcal{D}_3$ if $D\cup T_1$ is a connected component of $H[J_1\cup K]$ for some (unique) $T_1\in \mathcal{T}$ and $D\cup T_2$ is a connected component of $H[J_2\cup K]$ for some (unique) $T_2\in \mathcal{T}$.
\end{itemize}

For each $i \in \{ 1, 2 \}$, we also partition $\mathcal{T}_i$ into three parts \(\mathcal{T}_i=\mathcal{T}_{i,0}\,\dot\cup\,\mathcal{T}_{i,i}\,\dot\cup\,\mathcal{T}_{i,3}\) so that a clique $T_i\in \mathcal{T}_i$ is in
\begin{itemize}
    \item $\mathcal{T}_{i,0}$ if $T_i$ itself is a connected component of $H[J_i\cup K]$,
    \item $\mathcal{T}_{i,i}$ if $D\cup T_i$ is a connected component of $H[J_i\cup K]$ for some (unique) $D\in \mathcal{D}_i$, and
    \item $\mathcal{T}_{i,3}$ if $D\cup T_i$ is a connected component of $H[J_i\cup K]$ for some (unique) $D\in \mathcal{D}_3$.
\end{itemize}
We remark that, for $i\in \{1,2\}$, there is a one-to-one correspondence between cliques in $\mathcal{D}_i$ and $\mathcal{T}_{i,i}$, and $\mathcal{D}_3$ and $\mathcal{T}_{i,3}$, such that the union of corresponding cliques form a clique component of $H[J_i\cup K]$.
See \Cref{fig:mcis} for an illustration.
%\todo{remark about one-to-one correspondence of cliques}

For \(i\in\{1,2,3\}\), let \(\kappa_i := |\mathcal{D}_i|\).
Moreover, for $i\in \{1,2\}$, let $\lambda_{i}:=|\mathcal{T}_{i,0}|$.
By definition, \(|\mathcal{T}_{1,3}|=|\mathcal{T}_{2,3}|=|\mathcal{D}_3|=\kappa_3\) and \(|\mathcal{T}_{i,i}|=|\mathcal{D}_i|=\kappa_i\); particularly, $\kappa_1,\kappa_2,\kappa_3,\lambda_{1},\lambda_{2}\leq k$.
Write \(\mathcal{D}_i=\{D_i^1,\dots,D_i^{\kappa_i}\}\) for $i\in \{1,2,3\}$ and, for each \(i\in\{1,2\}\), 
\(\mathcal{T}_{i,0}=\{T_{i,0}^1,\dots,T_{i,0}^{\lambda_i}\}\),
\(\mathcal{T}_{i,i}=\{T_{i,i}^1,\dots,T_{i,i}^{\kappa_i}\}\), and
\(\mathcal{T}_{i,3}=\{T_{i,3}^1,\dots,T_{i,3}^{\kappa_3}\}\), 
so that the union of corresponding cliques of the same index induces a clique; particularly,
\begin{itemize}
    \item For each \(j\in\{1,\dots,\kappa_3\}\), the union \(D_3^j\cup T_{i,3}^j\) is a connected component of \(H[J_i\cup K]\) (for both \(i=1,2\)), and
    \item For each \(j\in\{1,\dots,\kappa_i\}\), the union \(D_i^j\cup T_{i,i}^j\) is a connected component of \(H[J_i\cup K]\) (for both \(i=1,2\)).
    %there exists a (unique) \(S_{1}^j\in\mathcal{A}_{1}\) such that \(A_{1}^j\cup B_{1,1}^j\) is a connected component of \(H[J_1\cup K]\), while \(A_{1}^j\) is a connected component of \(H[J_2\cup K]\).
    %\item For each \(B_{2,2}^j\in\mathcal{B}_{2,2}\) there exists a (unique) \(A_{2}^j\in\mathcal{A}_{2}\) such that \(A_{2}^j\cup B_{2,2}^j\) is a connected component of \(H[J_2\cup K]\), while \(A_{2}^j\) is a connected component of \(H[J_1\cup K]\).
    %\item For each \(A\in\mathcal{A}_0\), \(A\) is a connected component of \(H[J_1\cup J_2\cup K]\).
    %\item For each \(B_{i,0}\in\mathcal{B}_{i,0}\), \(B_{i,0}\) is a connected component of \(H[J_i\cup K]\).
\end{itemize}

\versionparagraph{Further guesses}
Let $\mathcal{C}_i$ denote the family of cliques of $G_i-X_i$, and let $\mathcal{S}_i$ denote the family of cliques of $G_i[X_i']$; discard the guess if $G_i[X_i']$ is not a cluster graph, since $X_i'=\phi_i(J_i)$ and thus $G_i[X_i']\cong H[J_i]$ must be a cluster graph for any valid guess.

We guess, in $2^{O(k)}$ time, the partition of $\mathcal{S}_i$ into three parts $\mathcal{S}_{i,0}=\{\phi_{i}(T_{i,0})\colon T_{i,0}\in \mathcal{T}_{i,0}\}$, $\mathcal{S}_{i,i}=\{\phi_{i}(T_{i,i})\colon T_{i,i}\in \mathcal{T}_{i,i}\}$, and $\mathcal{S}_{i,3}=\{\phi_{i}(T_{i,3})\colon T_{i,3}\in \mathcal{T}_{i,3}\}$.\footnote{
Here we slightly abuse notation so that $\phi_i(T_i)$ represents the clique of $\mathcal{S}_i$ induced by $\{\phi_i(v)\colon v\in T_i\}$.}
We further guess, in $k^{O(k)}$ time, the indexing of $\mathcal{S}_{i,0}$, $\mathcal{S}_{i,i}$, and $\mathcal{S}_{i,3}$; now, $S_{i,0}^j=\phi_i(T_{i,0}^j)$ for $j\in \lambda_i$, $S_{i,i}^j=\phi_i(T_{i,i}^j)$ for $j\in [\kappa_i]$, and $S_{i,3}=\phi_i(T_{i,3}^j)$ for $j\in [\kappa_3]$.
It automatically determines integers $\lambda_i$, $\kappa_i$, and $\kappa_3$; discard the guess if the value of $\kappa_3$ determined this way is inconsistent between $i\in \{1,2\}$ (that is, $|\mathcal{S}_{1,3}|\neq |\mathcal{S}_{2,3}|$).

%from can guess partial induced embeddings from $\bigcup \mathcal{T}_i$ into $X_i'$ from $k^{O(k)}$ options.
%Because $|J_i|\le k$, we can guess the structure of $H[J_i]$ from $k^{O(k)}$ optons by encoding its clique partition.
%Moreover, we can guess the graph $H[J_1 \cup J_2]$ from $2^{O(k^2)}$ options as $|J_1\cup J_2|\leq O(k)$.
%\todo{need to guess more; partition of $\mathcal{S}_i$.}

 \versionparagraph{Constructing a \textsc{Weighted Exact Multicolored Matching} instance}
We reduce to the \textsc{Weighted Exact Multicolored Matching} as follows.
Other than the special color $0$, we use \(\kappa_3+\kappa_1+\kappa_2+\lambda_1+\lambda_2\leq O(k)\) colors: \(\mathcal{D}_3\,\cup\,\mathcal{D}_1\,\cup\,\mathcal{D}_2\,\cup \, \mathcal{T}_{1,0}\,\cup \, \mathcal{T}_{2,0}\).
Build a bipartite graph $F$ whose left side is $\mathcal{C}_1\cup \mathcal{S}_{1,0}$ and right side is $\mathcal{C}_2\cup \mathcal{S}_{2,0}$.
For each $(C_1,C_2)\in \mathcal{C}_1\times \mathcal{C}_2$, we add edges of the following types between them, where we postpone the definition of the weights:
\begin{itemize}
    \item For each $j\in [\kappa_3]$, an edge of color $D_3^j$, which represents the possibility that $\phi_1(D^j_3\cup T^j_{2,3})\subseteq C_1$ and $\phi_2(D^j_3\cup T^j_{1,3})\subseteq C_2$.
    \item For each $j\in [\kappa_1]$, an edge of color $D_1^j$, which represents the possibility that $\phi_1(D_1^j)\subseteq C_1$ and $\phi_2(D_1^j\cup T_{1,1}^j)\subseteq C_2$.
    \item For each $j\in [\kappa_2$], an edge of color $D_2^j$, which represents the possibility that $\phi_1(D_2^j\cup T_{2,2}^j)\subseteq C_1$ and $\phi_2(D^j_2)\subseteq C_2$.
    \item An edge of color $0$, which represents the possibility that, for some $D_0\in \mathcal{D}_0$, $\phi_1(D_0)\subseteq C_1$ and $\phi_2(D_0)\subseteq C_2$.
\end{itemize}
Moreover, we add edges incident to $\mathcal{S}_{1,0}$ and $\mathcal{S}_{2,0}$ as follows:
\begin{itemize}
    \item For each $j\in [\lambda_1]$ and $C_2\in \mathcal{C}_2$, we add an edge of color $T_{1,0}^j$ between $S_{1,0}^j$ and $C_2$, which represents the possibility that $\phi_2(T_{1,0}^j)\subseteq C_2$ (where $S_{1,0}^j=\phi_1(T_{1,0}^j)$ by definition).
    \item For each $C_1\in \mathcal{C}_1$ and $j\in [\lambda_2]$, we add an edge of color $T_{2,0}^j$ between $C_1$ and $S_{2,0}^j$, which represents the possibility that $\phi_1(T_{2,0}^j)\subseteq C_1$ (where $S_{2,0}^j=\phi_2(T_{2,0}^j)$ by definition).
\end{itemize}

\versionparagraph{Edge weights}

From now on, we define the weight of an edge $(C_1,C_2)$ to represent the size of a clique $D\in \mathcal{D}$ (which we cannot afford to guess explicitly).
Although the details depend on which $\mathcal{D}_i$ contains $D$, in all four cases the weight is the maximum feasible size of $D$. 

To compute the weight, we first define the subsets $C_1^{\circ}\subseteq C_1$ and $C_2^{\circ}\subseteq C_2$;
intuitively, $C_i^{\circ}\subseteq C_i$ is the set of vertices in $C_i$ that can be in the image $\phi_i(D)$.
Then, the weight of the edge is the maximum of $|D|$ such that there is a way to extend $\phi_1|_{X_1}$ and $\phi_2|_{X_2}$ so that $\phi_1(D)\subseteq C_1^{\circ}$ and $\phi_2(D)\subseteq C_2^{\circ}$; we do not add the edge if no such extensions exist.
We compute this weight in polynomial time by the signature argument as in the proof of \Cref{lem:clique-to-clique}.

\subparagraph*{Weight of edges of color $D_3^j$.}
We first compute the weight of an edge of color $D_3^j$ between $C_1 \in \mathcal{C}_1$ and $C_2 \in \mathcal{C}_2$.
Such an edge represents the case where $\phi_{1}(D_3^j \cup T_{2,3}^j) \subseteq C_1$ and $\phi_2(D_3^j \cup T_{1,3}^j) \subseteq C_2$, and we have already guessed that $S_{i,3} = \phi_i(T_{i,3})$ for both $i \in \{1,2\}$.
To map $T_{2,3}^j$ into $C_1$, we must respect the guessed structure of $H[J_1\cup J_2]$, i.e., we require an injective map $\varphi_1 \colon T_{2,3}^j \to C_1$ such that $\phi_1|_{I\cup J_1} \cup \varphi_1$ is an induced embedding of $H[I \cup J_1 \cup T_{2,3}^j]$ into $G_1[\phi_1(I \cup J_1)\cup C_1]$.
More concretely, for every $t\in T_{2,3}^j\subseteq J_{2}$ and $u \in I \cup J_1$, we require that
\[
\{t,u\}\in E(H) \iff \{\varphi_1(t),\,\phi_1(u)\}\in E(G_1).
\]
Since the left-hand side is fixed by our guess of $H[I \cup J_1\cup J_2]$, we can find such a mapping $\varphi_1$ in polynomial time by the signature argument; there is at most one such mapping up to isomorphism because vertices in $C_1$ with the same signature to $X_1$ are equivalent.
If no such $\varphi_1$ exists, we do not add the edge.
%We do the same for $T_{1,3}^j$ and $C_2$ to obtain $\varphi_2$.

Let $C_1^\circ\subseteq C_1$ be the set of vertices $v\in C_1$ such that
\begin{itemize}
    \item $v\not \in \varphi_1(T^j_{2,3})$, and
    \item $N_{G_1}(v)\cap X'_1=S^j_{1,3}$. 
\end{itemize}
Intuitively, $C_1^{\circ}$ is the set of vertices that can be in $\phi_1(D_3^j)$; we exclude $\varphi_1(T^j_{2,3})$ because $T^j_{2,3}$ and $D_3^j$ are disjoint, and require $N_{G_1}(v)\cap X'_1=S^j_{1,3}$ because $\phi_2(D_3^j)\cup S^j_{1,3}$ is a clique component of $\phi_2(V[H])-X_2$.
In the same way, we define $C_2^{\circ}$.

Finally, let the weight of the edge $w_3(C_1,C_2,D_3^j)$ be the maximum of $|D_3^j|$ such that there is a pair of partial induced mappings $\psi_i\colon D_3^j\to C_i^{\circ}$ ($i\in \{1,2\}$) that $\phi_i|_{I\cup J_i}\cup \varphi_i \cup \psi_i$ is a partial induced embedding for both $i\in \{1,2\}$.
Recall that the vertices in $C_i^{\circ}$ have the same neighbor in $X'_i\cup T_{3-i,3}^j$; specifically, the neighbor is $S^j_{i,3}\cup T^j_{3-i,3}$.
Particularly, they can be distinguished only by the signature to $\phi_i(I)$; thus, the maximum of $|D_3^j|$ can be computed in polynomial time by the signature argument.
We refer to the choice of $D_3^j$ attaining the maximum size as $D_3^j(C_1,C_2)$ for the edge $(C_1,C_2)$.
%the largest common induced subgraph of $G_1[C_1^{\circ}]$ and $G_2[C_2^{\circ}]$ that extends $\phi_1|_{X_1}$ and $\phi_2|_{X_2}$, respectively; we do not add the edge if no such extension exists.
%Above discussion is summarized as follows.
%\begin{lemma}
%Let $\phi_i\colon V[H]\to V[G_i]$ ($i\in \{1,2\}$) be mappings extending $\phi_i|_{I\cup J_i}\cup \varphi_i\cup \psi_i$.
%Then, 
%\begin{itemize}
%    \item $\phi_i()$
%\end{itemize}
%\end{lemma}

\subparagraph*{Weight of edges of color $D_1^j$ and $D_2^j$.}
Next, we compute the weight of the edges colored $D_1^j$ and $D_2^j$. By symmetry, it suffices to consider $D_2^j$.
Such an edge represents the case where $\phi_1(D_2^j \cup T_{2,2}^j) \subseteq C_1$ and $\phi_2(D_2^j) \subseteq C_2$, and we have already guessed that $S_{2,2}^j = \phi_2(T_{2,2}^j)$.
As in the $D_3^j$ case, we enforce the guessed adjacencies between $T_{2,2}^j$ and $I\cup J_1$ and require an injective map $\varphi_1\colon T^j_{2,2}\to C_1$ that makes $\phi_1|_{I\cup J_1}\cup \varphi_1$ an induced embedding of $H[I\cup J_1\cup T_{2,2}^j]$ into $G_1[\phi_1(I\cup J_1)\cup C_1]$.
If no such $\varphi_1$ exists, we do not add the edge.
We define $C_1^{\circ}\subseteq C_1$ be the set of vertices $v\in C_1$ with  
\begin{itemize}
    \item $v\not \in \varphi_1(T^j_{2,2})$, and
    \item $N_{G_1}(v)\cap X'_1=\emptyset$.
\end{itemize}
Moreover, we define $C_2^{\circ}\subseteq C_2$ be the set of vertices $v\in C_2$ with
\begin{itemize}
    \item $N_{G_2}(v)\cap X'_2=S^j_{2,2}$.
\end{itemize}
Intuitively, $C_1^{\circ}$ and $C_2^{\circ}$ are the sets of vertices that can be in $\phi_1(D_2^j)$ and $\phi_2(D_2^j)$, respectively; for $C_1^{\circ}$, we exclude $\varphi_1(C_1)$ because $T^j_{2,2}$ and $D_2^j$ are disjoint, and require $N_{G_1}(v)\cap X'_1=\emptyset$ because $\phi_2(D_2^j)$ is a clique component of $\phi_2(V[H])-X_2$; for $C_2^{\circ}$, we require $N_{G_2}(v)\cap X'_2=S^j_{2,2}$ because $\phi_1(D_2^j)\cup S_{2,2}$ is a clique component of $\phi_1(V[H])-X_1$.

Finally, let the weight of the edge $w_2(C_1,C_2,D_2^j)$ be the maximum of $|D^j_2|$ such that there is a pair of partial induced mappings $\phi_i\colon D^j_2\to C_i^{\circ}$ ($i\in \{1,2\}$) that $\phi_1|_{I\cup J_1}\cup \varphi_1\cup \psi_1$ and $\phi_2|_{I\cup J_2}\cup \psi_2$ are partial induced embeddings. 
Recall that the vertices in $C_i^{\circ}$ can be distinguished only by the signature to $\phi_i(I)$; thus, the maximum of $|D^j_2|$ can be computed in polynomial time by the signature argument.
We refer to the choice of $D_2^j$ attaining the maximum size as $D_2^j(C_1,C_2)$ for the edge $(C_1,C_2)$.

%\begin{lemma}
%\end{lemma}

\subparagraph*{Weight of edges of color $0$.}
We now compute the weight of the edges of color $0$.
Such an edge represents the case where, for some $D_0\in \mathcal{D}_0$, $\phi_1(D_0)\subseteq C_1$ and $\phi_2(D_0)\subseteq C_2$.
We define $C_1^{\circ}\subseteq C_1$ be the set of vertices $v\in C_1$ with $N_{G_1}(v)\cap X'_1=\emptyset$; we require this because, assuming there is a clique $D_0\in \mathcal{D}_0$ with $\phi_1(D_0)\subseteq C_1$ and $\phi_2(D_0)\subseteq C_2$, $\phi_2(D_0)$ is a clique component of $\phi_2(V[H])-X_2$.
We define $C_2^{\circ}\subseteq C_2$ similarly.

Let the weight of the edge $w_0(C_1,C_2,0)$ be the maximum of $|D_0|$ such that there is a pair of partial induced mappings $\phi_i\colon D_0\to C_i^{\circ}$ ($i\in \{1,2\}$) that $\phi_i|_{I\cup J_i}\cup \psi_i$ is a partial induced embedding for each $i\in \{1,2\}$. 
Recall that the vertices in $C_i^{\circ}$ can be distinguished only by the signature to $\phi_i(I)$; thus, the maximum of $|D_0|$ can be computed in polynomial time by the signature argument.
We refer to the choice of $D_0$ attaining the maximum size as $D_0(C_1,C_2)$ for the edge $(C_1,C_2)$.

\subparagraph*{Weight of edges of color $T_{1,0}^j$ and $T_{2,0}^j$.}
Finally, we compute the weight of the edges of color $T_{1,0}^j$ and $T_{2,0}^j$.
By symmetry, it suffices to consider $T_{2,0}^j$.
Such an edge represents the case where $\phi_1(T_{2,0}^j)\subseteq C_1$ and $\phi_2(T_{2,0}^j)=S_{2,0}^j$.
As in the $D_3^j$ case, we enforce the guessed adjacencies between $T_{2,0}^j$ and $I\cup J_1$ and require an injective map $\varphi_1\colon T^j_{2,0}\to C_1$ such that $\phi_1|_{I\cup J_1}\cup \varphi_1$ is an induced embedding of $H[I\cup J_1\cup T_{2,0}^j]$ into $G_1[\phi_1(I\cup J_1)\cup C_1]$.
If no such $\varphi_1$ exists, we do not add the edge.
If such $\varphi_1$ exists, the weight of the edge is $0$.

%\todo{add edges that embeds $D_0$ into $T_{i,0}$}

\versionparagraph{Algorithm}

After the guesses, we build a colored bipartite graph $F$ as above, and solve \textsc{Weighted Exact Multicolored Matching} on it in randomized $2^{O(k)}$ time by \Cref{prop:mpm}.
Over all successful guesses, we return the maximum value of $|I\cup J_1\cup J_2|$ plus the weight of the chosen matching.
The time complexity is $O^*(2^{O(k^2)})$: the adjacency pattern between $J_1$ and $J_2$ has at most $2^{|J_1||J_2|}\le 2^{O(k^2)}$ possibilities, and the remaining guesses contribute only an additional $O^*(k^{O(k)})$ factor.

\versionparagraph{Correctness}

Now, we prove the correctness.
\begin{lemma}
The following are equivalent.
\begin{itemize}
    \item There is a graph $H$ on $t+|I\cup J_1\cup J_2|$ vertices and induced embeddings $\phi_i\colon V(H)\to V(G_i)$ for $i\in\{1,2\}$ that are consistent with the current guess.
    \item There is a \textsc{Weighted Exact Multicolored Matching} of weight at least $t$ in $F$. 
\end{itemize}
\end{lemma}
\begin{proof}
We prove both directions: first extract a multicolored matching from a solution $H$ and its embeddings, and then reconstruct $H$ and the embeddings from a matching.

($\Rightarrow$)
Let $H$ and the induced embeddings $\phi_i\colon V(H)\to V(G_i)$ $(i\in\{1,2\})$ be as in the lemma, extending the guessed $\phi_i|_{I\cup J_i}$.
Use the families $\mathcal{D}_{i}$ $(i\in \{0,1,2,3\})$, $\mathcal{T}_{i,0}$, $\mathcal{T}_{i,i}$, and $\mathcal{T}_{i,3}$ $(i\in \{1,2\})$ defined above.
Let $M$ be an edge set of $F$ obtained as follows:
\begin{itemize}
    \item For each $D^j_i\in \mathcal{D}_i$ $(i\in \{1,2,3\}, j\in [\kappa_i]$), add $(C_1,C_2,D^j_i)$ to $M$, where $C_1$ (resp. $C_2$) is the unique clique of $G_1-X_1$ (resp. $G_2-X_2$) with $\phi_1(D^j_i)\subseteq V[C_1]$ (resp. $\phi_2(D^j_i)\subseteq V[C_2]$).
    \item For each $D_0\in \mathcal{D}_0$, add $(C_1,C_2,0)$ to $M$, where $C_1$ (resp. $C_2$) is the unique clique of $G_1-X_1$ (resp. $G_2-X_2$) with $\phi_1(D_0)\subseteq V[C_1]$ (resp. $\phi_2(D_0)\subseteq V[C_2]$).
    \item For each $T^j_{2,0}$ $(j\in [\lambda_2])$, add $(C_1,S_{2,0}^j,T_{2,0}^j)$ to $M$, where $C_1$ is the unique clique of $G_1-X_1$ with $\phi_1(T_{2,0}^j)\subseteq V[C_1]$.
    \item For each $T^j_{1,0}$ $(j\in [\lambda_1])$, add $(S_{1,0}^j,C_2,T_{1,0}^j)$ to $M$, where $C_2$ is the unique clique of $G_2-X_2$ with $\phi_2(T_{1,0}^j)\subseteq V[C_2]$.
\end{itemize}
We first verify that $M$ is feasible, and then lower bound its weight by $|V[H]|-|I\cup J_1\cup J_2|$.
\begin{claim}
$M$ is a matching in $F$ that uses each nonzero color exactly once (colors in $\mathcal{D}_3\cup \mathcal{D}_1\cup \mathcal{D}_2\cup \mathcal{T}_{1,0}\cup \mathcal{T}_{2,0}$), and hence is a feasible solution to \textsc{Weighted Exact Multicolored Matching}.
\end{claim}
\begin{proof}
By definition, there are no edges between different cliques of the form $D_i^j$ ($i\in \{1,2,3\},j\in [\kappa_i]$), $D_0\in \mathcal{D}_0$, and $T^j_{2,0}$ ($j\in [\lambda_2]$).
Thus, the edge corresponding to these cliques does not share the left endpoint $C_1$. The same can be stated also for $C_2$.
Thus, $M$ is a matching.
Moreover, it is clear from the construction that $M$ contains exactly one edge of each color in $\mathcal{D}_3\cup \mathcal{D}_1\cup \mathcal{D}_2\cup \mathcal{T}_{1,0}\cup \mathcal{T}_{2,0}$; thus, $M$ is feasible.
\end{proof}
It remains to lower bound the weight of $M$.
Recall that the weight of each edge is the maximum size of a clique in $K$ that can be embedded into the cliques at its endpoints, considering only the embeddings of the relevant auxiliary cliques; more precisely,
\begin{itemize}
    \item the weight of $(C_1,C_2,D_3^j)$ is the maximum of $|D_3^j|$ such that $\phi_1|_{I\cup J_1}\cup \varphi_1\cup \psi_1\colon (I\cup J_1)\cup T_{2,3}^j\cup D_3^j\to X_1\cup C_1$ and $\phi_2|_{I\cup J_2}\cup \varphi_2\cup \psi_2\colon (I\cup J_2)\cup T_{1,3}^j\cup D_3^j\to X_2\cup C_2$ are partial induced mappings,
    \item the weight of $(C_1,C_2,D_2^j)$ is the maximum of $|D_2^j|$ such that $\phi_1|_{I\cup J_1}\cup \varphi_1\cup \psi_1\colon (I\cup J_1)\cup T_{2,2}^j\cup D_2^j\to X_1\cup C_1$ and $\phi_2|_{I\cup J_2}\cup \psi_2\colon (I\cup J_2)\cup D_2^j\to X_2\cup C_2$ are partial induced mappings,
    \item the weight of $(C_1,C_2,D_1^j)$ is the maximum of $|D_1^j|$ such that $\phi_1|_{I\cup J_1}\cup \varphi_1\colon (I\cup J_1)\cup D_1^j\to X_1\cup C_1$ and $\phi_2|_{I\cup J_2}\cup \varphi_2\cup \psi_2\colon (I\cup J_2)\cup T_{1,1}^j \cup D_1^j\to X_2\cup C_2$ are partial induced mappings,
    \item the weight of $(C_1,C_2,0)$ is the maximum of $|D_0|$ such that $\phi_1|_{I\cup J_1}\cup \varphi_1\colon (I\cup J_1)\cup D_0\to X_1\cup C_1$ and $\phi_2|_{I\cup J_2}\cup \psi_2\colon (I\cup J_2)\cup D_0\to X_2\cup C_2$ are partial induced mappings,
    \item the weight of $(C_1,S_{2,0}^j,T_{2,0}^j)$ is $0$, and
    \item the weight of $(S_{1,0}^j,C_2,T_{1,0}^j)$ is $0$.
\end{itemize}
Since $\phi_i$ is an induced embedding, when we restrict the domain of $\phi_i$, we obtain the partial induced embeddings $\phi_i|_{I\cup J_i}\cup \varphi_i\cup \psi_i$ or $\phi_i|_{I\cup J_i}\cup \psi_i$ meeting above conditions; particularly, the maximum size of $D_i^j$ and $D_0$ are bounded above by the weight of the corresponding edges.
Therefore, the total weight of $M$ is at least $|K|=|V[H]|-|I\cup J_1\cup J_2|$.

($\Leftarrow$)
Let $M$ be the computed matching.
We build $H$ from the cliques selected by $M$, and then extend the guessed embeddings accordingly.
By the multicoloredness constraint, each nonzero color appears exactly once in $M$; in particular, for every $i\in\{1,2,3\}$ and $j\in[\kappa_i]$, we define $D_i^j:=D_i^j(C_1,C_2)$ using the unique edge in $M$ of the form $(C_1,C_2,D_i^j)$.

Define $H$ by the graph on the vertex set $I\cup J_1\cup J_2\cup K$, where $K$ is defined by
\[
    K:= \left(\bigcup_{i\in \{1,2,3\}}\bigcup_{(C_1,C_2,D^j_i)\in M}D_i^j \right)\cup \bigcup_{(C_1,C_2,0)\in M}D_0(C_1,C_2).
\]
By construction, $|K|$ is equal to the weight of $M$; thus, $H$ contains at least $t+|I\cup J_1\cup J_2|$ vertices.
The edge set of $H$ is defined as follows:
\begin{itemize}
    \item The edges in $H[I\cup J_1\cup J_2]$ are fixed by the guess.
    \item For each $u\in J_i$ and $v\in K$ ($i\in \{1,2\}$), $(u,v)\in E[H]$ if and only if either 
    \begin{itemize}
        \item $u\in T_{i,3}^j$ and $v\in D_3^j$ for some $j\in [\kappa_3]$, or 
        \item $u\in T_{i,i}^j$ and $v\in D_i^j$ for some $j\in [\kappa_i]$.
    \end{itemize}
    \item For distinct $u,v\in K$, $(u,v)\in E[H]$ if and only if either
    \begin{itemize}
        \item both are in the same $D_i^j$ for $i\in \{1,2,3\}$ and $j\in [\kappa_i]$, or 
        \item both are in the same $D_0(C_1,C_2)$ for $(C_1,C_2,0)\in M$.
    \end{itemize}
    \item Edges between $u\in I$ and $v\in K$ are defined via the corresponding map $\psi_1$ (equivalently, $\psi_2$). More precisely,
    \[
        N_{H}(v)\cap I:=\psi_1^{-1}(N_{G_1}(\psi_1(v))\cap \phi_1|_{I\cup J_1}(I))=\psi_2^{-1}(N_{G_2}(\psi_2(v))\cap \phi_2|_{I\cup J_2}(I)),
    \]
    where $\psi_1$ or $\psi_2$ comes from the (unique) edge of the form 
    \begin{itemize}
        \item $(\cdot,\cdot,D_i^j)\in M$ if $v\in D_i^j$, and
        \item $(C_1,C_2,0)\in M$ if $v\in D_0(C_1,C_2)$.
    \end{itemize}
\end{itemize}

This completes the definition of $H$.
Using the unique edge of $M$ associated with each auxiliary clique in $K$, we extend the guessed partial embedding to all of $H$; define $\phi_1\colon H\to V[G_1]$ as follows:
%and $\phi_1\colon H\to V[G_1]$ be an extension of the guessed $\phi_1|_{I\cup J_1}$ such that
\begin{itemize}
    \item $\phi_1$ extends the guessed $\phi_1|_{I\cup J_1}$,
    \item $\phi_1|_{T^j_{2,3}}$ equals the $\varphi_1$ of the unique edge of the form $(\cdot,\cdot,D_3^j)\in M$,
    \item $\phi_1|_{D_3^j}$ equals the $\psi_1$ of the unique edge of the form $(\cdot,\cdot,D_3^j)\in M$,
    \item $\phi_1|_{T^j_{2,2}}$ equals the $\varphi_1$ of the unique edge of the form $(\cdot,\cdot,D_2^j)\in M$,
    \item $\phi_1|_{D_2^j}$ equals the $\psi_1$ of the unique edge of the form $(\cdot,\cdot,D_2^j)\in M$,
    \item $\phi_1|_{D_1^j}$ equals the $\psi_1$ of the unique edge of the form $(\cdot,\cdot,D_1^j)\in M$,
    \item $\phi_1|_{T^j_{2,0}}$ equals the $\varphi_1$ of the unique edge of the form $(\cdot,S^j_{2,0},T^j_{2,0})\in M$, and
    \item $\phi_1|_{D_0(C_1,C_2)}$ equals the $\psi_1$ of the edge $(C_1,C_2,0)\in M$ (for each color-$0$ edge in $M$).
\end{itemize}
%It suffices to prove that the following map $\phi_2\colon H\to V[G_2]$ is an induced mapping:
Define $\phi_2\colon H\to V[G_2]$ analogously.
To conclude, it suffices to show that both $\phi_1$ and $\phi_2$ are induced embeddings; by symmetry, it is enough to prove the following claim.
\begin{claim}
For each $u,v\in V[H]$,
\[
    (u,v)\in E[H]
    \Longleftrightarrow (\phi_1(u),\phi_1(v))\in E[G_1].
\]
\end{claim}
\begin{proof}
%\noindent{}\textbf{Case $u,v\in I\cup J_1\cup J_2$:}
%It suffices to show that the edge set of $H[I\cup J_1\cup J_2]\cong G_i[\phi_i(I\cup J_1\cup J_2)]$ defined by $\phi_i^{-1}$ is consistent to the guess of the structure of $H[I\cup J_1\cup J_2]$.
%From symmetry, we prove only for $\phi_1$.
%Since $\phi_1$ extends the guessed $\phi_1|_{I\cup J_1}$, it suffices ...

%\noindent{}\textbf{Case $u\in I\cup J_1\cup J_2$ and $v\in K$:}

%\noindent{}\textbf{Case $u,v\in K$:}

We proceed by casework.

\proofsubparagraph*{Case $u,v\in I\cup J_1$:} 
This case is trivial because $\phi_1$ extends $\phi_1|_{I\cup J_1}$, which is a partial induced embedding.

\proofsubparagraph*{Case $u\in I\cup J_1$ and $v\in J_2$:}
The clique component of $J_2$ containing $v$ is either in $\mathcal{T}_{2,3}$, $\mathcal{T}_{2,2}$, or $\mathcal{T}_{2,0}$.
The claim is from the definitions of $\varphi_1$ of the unique edge of the form $(\cdot,\cdot, D_3^j)$, $(\cdot,\cdot, D_2^j)$, and $(\cdot,S_{2,0}^j,T_{2,0}^j)$, respectively.

\proofsubparagraph*{Case $u\in I\cup J_1$ and $v\in K$:}
The clique component of $K$ containing $v$ is either in $\mathcal{D}_{3}$, $\mathcal{D}_{1}$, $\mathcal{D}_{2}$, or $\mathcal{D}_{0}$.
The claim is obtained from the fact that the edges of $H$ are defined using the corresponding $\psi_1$.
%of the unique edge of the form $(\cdot,\cdot,D^j_3)$, $(\cdot,\cdot,D^j_1)$, and $(\cdot,\cdot,D^j_2)$, and the edge $(C_1,C_2,D_0(C_1,C_2))\in M$, respectively.

\proofsubparagraph*{Case $u,v\in J_2$:}
Each of $u$ and $v$ is in a clique component of $J_2$, which is either in $\mathcal{T}_{2,3}$, $\mathcal{T}_{2,2}$, or $\mathcal{T}_{2,0}$; thus, they are mapped to some clique component of $G_1-X_1$ by $\varphi_1$ of the unique edge of the form $(\cdot,\cdot,D_3^j)$, $(\cdot,\cdot,D_2^j)$, or $(\cdot,S_{2,0}^j,T_{2,0}^j)$.
If $u$ and $v$ are in the same clique component of $J_2$, they will be mapped to the same clique component of $G_1-X_1$.
Conversely, if $u$ and $v$ are in different clique components, they will be mapped to different clique components of $G_1-X_1$ because there is at most one edge of $M$ of the form $(C_1,\cdot, \cdot)$ for each $C_1$.

\proofsubparagraph*{Case $u\in J_2$ and $v\in K$:}
$(u,v)\in E[H]$ holds if and only if $u\in T_{2,3}^j$ and $v\in D_{3}^j$ for some $j\in [\kappa_3]$ or $u\in T_{2,2}^j$ and $v\in D_{2}^j$ for some $j\in [\kappa_2]$.
In both cases, both $T_{2,i}^j$ and $D_{i}^j$ are mapped to the same clique component of $G_1-X_1$ by the $\varphi_1$ and $\psi_1$ of the unique edge of the form $(\cdot, \cdot, D_i^j)$.
Conversely, in other cases, they will be mapped to different clique components of $G_1-X_1$ because there is at most one edge of $M$ of the form $(C_1,\cdot, \cdot)$ for each $C_1$. 

\proofsubparagraph*{Case $u,v\in K$:}
Each of $u$ and $v$ is in a clique component of $K$, which is either in $\mathcal{D}_{3}$, $\mathcal{D}_{1}$, $\mathcal{D}_{2}$, or $\mathcal{D}_{0}$; thus, each is mapped to some clique component of $G_1-X_1$ by $\psi_1$ of the unique edge of the form $(\cdot,\cdot,D_i^j)$ if it is in $i\in \{1,2,3\}$ and of the edge of the form $(C_1,C_2,0)$ if it is in $D_0(C_1,C_2)$.
If $u$ and $v$ are in the same clique component of $K$, they will be mapped to the same clique component of $G_1-X_1$.
Conversely, if $u$ and $v$ are in different clique components, they will be mapped to different clique components of $G_1-X_1$ because there is at most one edge of $M$ of the form $(C_1,\cdot, \cdot)$ for each $C_1$.
\end{proof}

This concludes the proof.
\end{proof}

\versionparagraph{Remark}
In the case of \textsc{ISI}, we can, using \Cref{lem:modulator-to-modulator}, restrict our attention to solutions in which the modulator of $G_2$ is mapped entirely into the modulator of $G_1$.
For \textsc{MCIS} we can no longer make such an assumption:
vertices of the cluster vertex deletion set of $G_1$ may be mapped to vertices of $G_2$ outside the cluster vertex deletion set and vice versa. Consequently, our algorithm must explicitly guess how these deletion sets attach to each other (i.e., $H[J_1 \cup J_2]$).
This guess step appears to be unavoidable when parameterizing by $\cvd(G_1)+\cvd(G_2)$:
As we show in \Cref{sec:mcis-lb}, \textsc{MCIS} cannot be solved in $O^*(2^{o(k^2)})$ time assuming ETH.

\section{ETH lower bound for MCIS} \label{sec:mcis-lb}

We prove the following ETH lower bound, showing that the algorithm for \textsc{MCIS} of \Cref{thm:mcis-fpt} is essentially tight under ETH.

\mcislb*

We proceed in two steps. We first establish the ETH lower bound for \textsc{LCBM} (\Cref{ssec:lcbm}) and then reduce \textsc{LCBM} to \textsc{MCIS} to obtain the stated lower bound (\Cref{ssec:mcislb}).

\subsection{Hardness of \textsc{LCBM}} \label{ssec:lcbm}

We first prove the hardness of \textsc{LCBM}.
Recall that an instance of \textsc{LCBM} consists of an integer $k$ and families $\mathcal{R}_i,\mathcal{C}_i\subseteq \{0,1\}^k$ for $i\in [k]$, and asks whether there is a Boolean matrix $A\in\{0,1\}^{k\times k}$ whose $i$-th row lies in $\mathcal{R}_i$ and whose $i$-th column lies in $\mathcal{C}_i$ for every $i$.
A trivial algorithm is to, for each $k\times k$-matrix over $\{ 0, 1 \}$, check if it meets the feasibility conditions, which takes $2^{O(k^2)}$ time.
We show that this trivial algorithm is essentially optimal under ETH by a reduction from \textsc{3-SAT}.

%We consider the following matrix feasibility problem, which we call \textsc{List-Constrained Boolean Matrix (LCBM)}.
%\boxproblem{\textsc{List-Constrained Boolean Matrix (LCBM)}}{An integer $k\in \mathbb{N}$, two sets of length-$k$ Boolean vectors $\mathcal{R}_i \subseteq \{0,1\}^k$ and $\mathcal{C}_i \subseteq \{0,1\}^k$ for each $i \in [k]$.}{Is there a Boolean matrix $A \in \{0,1\}^{k\times k}$ such that, for every $i \in [k]$, the $i$-th row $A[i,\cdot]$ lies in $\mathcal{R}_i$ and the $i$-th column $A[\cdot,i]$ lies in $\mathcal{C}_i$?}{$k$}\todo{delete it; copied to intro}

%Given $k \in \mathbb{N}$, for each $i \in [k]$ we are given two sets of length-$k$ Boolean vectors, $\mathcal{R}_i \subseteq \{0,1\}^k$ and $\mathcal{C}_i \subseteq \{0,1\}^k$.
%The task is to find a Boolean matrix $A \in \{0,1\}^{k\times k}$ such that, for every $i \in [k]$, the $i$-th row $A[i,\cdot]$ lies in $\mathcal{R}_i$ and the $i$-th column $A[\cdot,i]$ lies in $\mathcal{C}_i$.
%Here, we allow $\mathcal{R}_i$ and $\mathcal{C}_i$ to have exponential size, which is up to $2^k$.

Let $\Phi$ be a 3-CNF with $n$ variables and $m$ clauses. 
By the sparsification lemma of Impagliazzo et al.\  \cite{DBLP:journals/jcss/ImpagliazzoPZ01}, we may additionally assume that each variable appears in at most $D$ clauses for a universal constant $D$.
In particular, we have $m = O(n)$.
%Without loss of generality, we assume that $n$ is a square.

%\todo{specify the size of $\mathcal{R}_i$ and $\mathcal{C}_i$. Additionally, I like using different index for rows $(i)$ and columns $(j)$.}
\versionparagraph{Construction}

We will construct an instance of \textsc{LCBM} of dimension $k \times \ell$ for $k=O(\sqrt{n})$ and $\ell=O(\sqrt{n})$ (to get a square matrix, add dummy rows or columns where all vectors are allowed).
Our intermediate goal is to find a ``good'' pair of partitions of the set of variables into $k$ blocks $X_1,\dots, X_k$ and the set of clauses into $\ell$ groups $S_1,\dots, S_{\ell}$.
For $j\in [\ell]$, let $Y_j$ be the set of variables that appear in some clause in $S_j$.
The desired property is that, for each pair $(i,j)\in [k]\times [\ell]$, $|X_i\cap Y_j|\leq 1$.

For instance, we assume that we have such a pair of partitions and construct vector sets $\mathcal{R}_i$ and $\mathcal{C}_j$.
Let $i\in [k]$ and define
\[
    \iota(x) := \{ j\in[\ell] : x \text{ appears in some clause of } S_j \}
\]
for each $x\in X_i$.
Since $|X_i\cap Y_j|\leq 1$ for each $(i,j)\in [k]\times [\ell]$, $\iota(x)$ for different $x\in X_i$ are disjoint.
We define $\mathcal{R}_i\subseteq\{0,1\}^{\ell}$ to be the set of all vectors $r$ such that for every variable $x\in X_i$, the coordinates indexed by $\iota(x)$ are all equal.

Let $j\in [\ell]$.
We define $\mathcal{C}_j\subseteq\{0,1\}^{k}$ to be the
set of all vectors $c$ such that, (partially) assigning $c[i]$ to the unique variable $x\in X_i\cap Y_j$ (if any) for each $i$ satisfies all clauses in $S_j$.

The reduction runs in $2^{O(k)}=2^{O(\sqrt{n})}$ time.
The following lemma proves the correctness of this reduction.
\begin{lemma}
    Assume $|X_i\cap Y_j|\leq 1$ holds for all $(i,j)\in [k]\times [\ell]$. Then, $\Phi$ is satisfiable if and only if $(k, \{ \mathcal{R}_i \}_{i \in [k]}, \{ \mathcal{C}_i \}_{i \in [k]})$ is a yes-instance of \textsc{LCBM}.
\end{lemma}
\begin{proof}
    $(\Rightarrow)$
    If $\Phi$ is satisfiable by an assignment $\sigma$, define $A\in\{0,1\}^{k\times \ell}$ as follows:
    For  each $i\in[k]$ and each $j\in [\ell]$, if $X_i\cap Y_j=\{x\}$ then set $A[i,j]:=\sigma(x)$, and set all other entries arbitrarily.
    For each $i \in [k]$, the $i$-th row is constant on each $\iota(x)$, so $A[i,\cdot]\in\mathcal{R}_i$.
    For each $j\in [\ell]$, a partial assignment that sets the unique $x \in X_i \cap Y_j$ to $A[i,j] = \sigma(x)$
    satisfies all clauses in $S_j$ by construction, hence $A[\cdot, j]\in\mathcal{C}_j$.

    \smallskip
    $(\Leftarrow)$
    Conversely, suppose $A\in\{0,1\}^{k \times \ell}$ satisfies $A[i, \cdot] \in\mathcal{R}_i$ for all $i \in [k]$ and $A[\cdot,j]\in\mathcal{C}_j$ for all $j \in [\ell]$.
    For each variable $x\in X_i$ choose any $j\in\iota(x)$ and set $\sigma(x):=A[i,j]$; this is well-defined because $A[i,\cdot]\in\mathcal{R}_i$ forces the $i$-th row to be constant on $\iota(x)$.
    Now consider any clause in some $S_j$ with variables $x,y,z$ drawn from blocks $i_x,i_y,i_z$. Since $A[\cdot,j]\in\mathcal{C}_j$, the clause is satisfied under the values $A[i_x,j],A[i_y,j],A[i_z,j]$, which equal $\sigma(x),\sigma(y),\sigma(z)$, respectively. Hence $\sigma$ satisfies $\Phi$.
\end{proof}

\versionparagraph{Getting a partition pair}

The remaining task is to obtain a pair of partitions with a desired property, using the following two lemmas. 
Both partitions are obtained by a coloring argument of bounded-degree graphs.
We begin with the following.
\begin{lemma}\label{lem:partition_x}
    There is a polynomial-time algorithm that partitions the set of variables into blocks $X_1,\dots, X_k$, where $k=O(\sqrt{n})$ and $|X_i| = O(\sqrt{n})$ for each $i$, such that no clause contains two different variables from the same block.
\end{lemma}
\begin{proof}
    We build a graph $H_X$ as follows. 
    The vertex set of $H_X$ is the set of variables of~$\Phi$. 
    There is an edge between two vertices if there is a clause of $\Phi$ that contains both of the corresponding variables.
    Since each variable appears in at most $D$ clauses, the maximum degree $\Delta(H_X)$ is at most $2D$.
    Properly color vertices of $H_X$ greedily with $\Delta(H_X)+1 \le (2D+1)$ colors, and let $X'_1,\dots, X'_{2D+1}$ be the color classes.
    
    We further (arbitrarily) partition each color class $X'_i$ into the minimum number of blocks so that each block contains at most $\sqrt{n}$ variables.
    Then, the total number of blocks is at most $\sqrt{n}+(2D+1) = O(\sqrt{n})$.
    Moreover, no clause contains two different variables from the same block because they should belong to different color classes.
\end{proof}
We proceed as follows.
\begin{lemma}
    Let $X_1,\dots, X_k$ be a partition satisfying the conditions in Lemma~\ref{lem:partition_x}.
    Then, there is a polynomial-time algorithm to compute a partition $S_1,\dots, S_{\ell}$ into groups, where $\ell=O(\sqrt{n})$, such that for each $(i,j)\in [k]\times [\ell]$, $|X_i\cap Y_j|\leq 1$, where $Y_j$ denotes the set of variables that appear in $S_j$. 
\end{lemma}
\begin{proof}
    We build a graph $H_S$ as follows.
    The vertex set of $H_S$ is the set of clauses of $\Phi$.
    There is an edge between two vertices $C,C'$ if there exists $i\in [k]$ such that $C$ contains a variable $x\in X_i$ and $C'$ contains a different variable $y\in X_i$.
    Since no clause contains two different variables from the same block, $H_S$ contains no self-loop.
    Since every clause contains at most three variables, each variable occurs in at most $D$ clauses, and each $X_i$ contains at most $O(\sqrt{n})$ variables, the maximum degree $\Delta(H_S)$ of $H_S$ is at most $3\cdot D\cdot \sqrt{n} = O(\sqrt{n})$.
    Properly color vertices of $H_S$ greedily with $\ell \le \Delta(H_S)+1 = O(\sqrt{n})$ colors, and let $S_1,\dots,S_\ell$ be the color classes.
    By construction, for every $i\in [k]$ and $j \in [\ell]$, at most one distinct variable from $X_i$ appears in the family of clauses $S_j$.
    %Let $Y_j$ denote the set of variables that appear in $S_j$.
    Thus, $|X_i \cap Y_j|\le 1$ for all $i \in [k]$ and $j \in [\ell]$.
\end{proof}

%In our reduction, rows index the variable blocks $X_1,\dots,X_k$ with $k=\sqrt{n}$, and columns index clause groups $S_1,\dots,S_\ell$ obtained by properly coloring a graph on the set of clauses where two clauses are adjacent when they share different variables from the same block.
%The coloring ensures that for each $i,j$, at most one variable from $X_i$ appears in $S_j$ (so $A[i,j]$ records that value if present), and the sparsification lemma guarantees $O(k)$ colors.
%With this setup, we define row lists $\mathcal{R}_i$ to enforce per‑variable consistency across columns and column lists $\mathcal{C}_j$ so that each $S_j$ is satisfied.\todo{seems this part needs more explanation}

%We will construct an instance of \textsc{LCMB} of dimension $k \times \ell$ (to get a square matrix, add dummy rows or columns where all vectors are allowed).

\versionparagraph{Analysis}
Note that the reduction runs in $2^{O(\sqrt{n})}$ time, and that the constructed instance has dimension $K:=\max(k,\ell)=\Theta(\sqrt{n})$.
Thus an $O^*(2^{o(K^2)})$-time algorithm for \textsc{LCBM} would yield a $2^{o(n)}$-time algorithm for 3-SAT, contradicting ETH.

Now we have the following.
\lcbmlb*

\subsection{Hardness of \textsc{MCIS}} \label{ssec:mcislb}

We now reduce \textsc{LCBM} to \textsc{MCIS}. The reduction proceeds by first normalizing the diagonal of the target matrix via guessing, and then encoding the remaining constraints into graph structure so that feasible matrices correspond to common induced subgraphs.

\versionparagraph{Guessing diagonal entries}
Given an instance $(k,\mathcal{R},\mathcal{C})$ of \textsc{LCBM}, we first guess the diagonal of the matrix $A$.
By appropriately flipping, for each $i\in [k]$, the $i$-th coordinate of each vector in $\mathcal{R}_i$ and $\mathcal{C}_i$ according to the guess, without loss of generality, we can assume all the diagonal entries of $A$ are $1$.
Particularly, this guess yields $2^k$ instances of \textsc{LCBM} with an additional constraint that the diagonal should be all-$1$.
For each of those instances, we construct an equivalent instance of \textsc{MCIS} in $2^{O(k)}$ time, where $|V(G_1)| + |V(G_2)| \in 2^{O(k)}$ and $\cvd(G_1) + \cvd(G_2) \in O(k)$.

    \smallskip
\begin{figure}
    \centering
    \begin{subfigure}[t]{.9\linewidth}
    \begin{center}
        \begin{tikzpicture}[
    x=2cm,y=2cm,
    main/.style={circle,draw,fill=black,inner sep=2pt},
    aux/.style={circle,draw,inner sep=1.8pt},
    every node/.style={font=\footnotesize},
    scale=1.2
  ]

  % parameters
  \def\k{4}      % number of pairs (a^i,b^i)
  \def\ay{0.8}     % y-coordinate of a-path
  \def\by{0}     % y-coordinate of b-path
  \def\rA{0.35}  % radius for 5-cycles
  \def\rB{0.35}  % radius for 7-cycles

  %--- main paths (a^1,...,a^k) and (b^1,...,b^k) ----------------------------
  \foreach \i in {1,...,\k} {
    \node[main,label=below left:{$a^{\i}$}] (a\i) at (\i,\ay) {};
    \node[main,label=above right:{$b^{\i}$}] (b\i) at (\i,\by) {};
    % \node[main] (a\i) at (\i,\ay) {};
    % \node[main] (b\i) at (\i,\by) {};
  }

  \foreach \i in {1,...,\numexpr\k-1\relax} {
    \pgfmathtruncatemacro{\j}{\i+1}
    \draw (a\i) -- (a\j);
    \draw (b\i) -- (b\j);
  }

  %--- matching edges a^i b^i -------------------------------------------------
  \foreach \i in {1,...,\k} {
    \draw (a\i) -- (b\i);
  }

  %--- 5-cycles at each a^i (regular pentagon; a^i is one vertex) -------------
  % We take a^i as the vertex at angle 270 degrees; center is above a^i.
  \foreach \i in {1,...,\k} {
    \pgfmathsetmacro{\cx}{\i}
    \pgfmathsetmacro{\cy}{\ay + \rA}

    % new vertices a^{i,1},...,a^{i,4}
    \foreach \j in {1,...,4} {
      \pgfmathsetmacro{\angle}{270 + 72*\j}
      \coordinate (tmp) at ({\cx + \rA*cos(\angle)},{\cy + \rA*sin(\angle)});
    %   \node[aux,label=above:{$a^{\i,\j}$}] (a-\i-\j) at (tmp) {};
      \node[aux] (a-\i-\j) at (tmp) {};
    }

    % edges of the 5-cycle (a^i, a^{i,1},...,a^{i,4})
    \draw (a\i)
      -- (a-\i-1) -- (a-\i-2) -- (a-\i-3) -- (a-\i-4) -- (a\i);
  }

  %--- 7-cycles at each b^i (regular heptagon; b^i is one vertex) -------------
  % We take b^i as the vertex at angle 90 degrees; center is below b^i.
  \foreach \i in {1,...,\k} {
    \pgfmathsetmacro{\cx}{\i}
    \pgfmathsetmacro{\cy}{\by - \rB}

    % new vertices b^{i,1},...,b^{i,6}
    \foreach \j in {1,...,6} {
      \pgfmathsetmacro{\angle}{90 + 360/7*\j}
      \coordinate (tmpb) at ({\cx + \rB*cos(\angle)},{\cy + \rB*sin(\angle)});
    %   \node[aux,label=below:{$b^{\i,\j}$}] (b-\i-\j) at (tmpb) {};
      \node[aux] (b-\i-\j) at (tmpb) {};
    }

    % edges of the 7-cycle (b^i, b^{i,1},...,b^{i,6})
    \draw (b\i)
      -- (b-\i-1) -- (b-\i-2) -- (b-\i-3)
      -- (b-\i-4) -- (b-\i-5) -- (b-\i-6) -- (b\i);
  }

  %--- extra vertices a^{1,0} and b^{1,0} ------------------------------------
  % Placed further along the ray from the corresponding cycle vertex.

  % a^{1,0} connected to a^{1,1}
  \pgfmathsetmacro{\cxA}{1}
  \pgfmathsetmacro{\cyA}{\ay + \rA}
%   \pgfmathsetmacro{\angA}{270 + 72*1}
  \pgfmathsetmacro{\angA}{216}
%   \node[aux,label=above:{$a^{1,0}$}] (a10)
%     at ({\cxA + 1.8*\rA*cos(\angA)},{\cyA + 1.8*\rA*sin(\angA)}) {};
  \node[aux,label=left:{$a^{1,0}$}] (a10)
    at ({\cxA + 1.8*\rA*cos(\angA)},{\cyA + 1.8*\rA*sin(\angA)}) {};
  \draw (a10) -- (a-1-4);

  % b^{1,0} connected to b^{1,1}
  \pgfmathsetmacro{\cxB}{1}
  \pgfmathsetmacro{\cyB}{\by - \rB}
  \pgfmathsetmacro{\angB}{90 + 360/7*1}
  \node[aux,label=left:{$b^{1,0}$}] (b10)
    at ({\cxB + 1.8*\rB*cos(\angB)},{\cyB + 1.8*\rB*sin(\angB)}) {};
  \draw (b10) -- (b-1-1);

\end{tikzpicture}
    \end{center} 
    \subcaption{Base gadget $G_0$ (illustrated for $k=4$): two length-$k$ paths with attached 5- and 7-cycles, matched by edges $a^i b^i$, plus the extra vertices $a^{1,0}$ and $b^{1,0}$.}
    \label{fig:lb:G0}
    \end{subfigure}

    \vspace{5ex}

    \begin{subfigure}{.9\linewidth}
    \begin{center}
    \begin{tikzpicture}[
        x=2cm,y=2cm,
        main/.style={circle,draw,fill=black,inner sep=2.4pt},
        vtx/.style={circle,draw,inner sep=2pt},
        rect/.style={draw,rounded corners,inner sep=4pt},
        every node/.style={font=\scriptsize},
        >={Latex[length=3mm,width=2mm]}
      ]
      % Rows for i = 1,2,3
      \foreach \row/\i in {0/1, 1.6/2, 3.2    /3} {
        \begin{scope}[shift={(\row,0)}]
        % a_1^i from G_0
        \node[main,label=below left:{$a_1^{\i}$}] (a-\i) at (0.6,-0.7) {};

        % P_1^i: four vertices on a horizontal line
        \node[vtx] (p-\i-1) at (0,0) {};
        \node[vtx] (p-\i-2) at (.4,0) {};
        \node[vtx] (p-\i-3) at (.8,0) {};
        \node[vtx] (p-\i-4) at (1.2,0) {};

        % Rounded rectangle around P_1^i
        % \node[ultra thick,rect,fit=(p-\i-1)(p-\i-2)(p-\i-3)(p-\i-4)] (R-\i) {};
        \node[ultra thick,rect,inner ysep=10pt,fit=(p-\i-1)(p-\i-2)(p-\i-3)(p-\i-4)] (R-\i) {};
        \node[font=\footnotesize] at (0, -.5) {$P_1^{\i}$};

        % Edges a_1^i -- p_1^{i,*}
        \foreach \j in {1,...,4} {\draw (a-\i) -- (p-\i-\j);}
        % \draw (a-\i) -- (R-\i);

        % Clique on P_1^i with curved edges
        \draw (p-\i-1) -- (p-\i-2);
        \draw (p-\i-2) -- (p-\i-3);
        \draw (p-\i-3) -- (p-\i-4);
        \draw (p-\i-1) edge[bend left=20] (p-\i-3);
        \draw (p-\i-2) edge[bend left=20] (p-\i-4);
        \draw (p-\i-1) edge[bend right=15] (p-\i-4);

        % q_1^i and b_1^i, with b_1^i -- q_1^i edge
        \node[vtx,label=above right:{$q_1^{\i}$}] (q-\i) at (0.6,0.7) {};
        \node[main,label=above right:{$b_1^{\i}$}] (b-\i) at (0.6,1.2) {};
        \draw (b-\i) -- (q-\i);

        % Connect the rounded rectangle (representing P_1^i) to q_1^i
        \foreach \j in {1,...,4} {\draw (q-\i) -- (p-\i-\j);}
        \end{scope}
      }
      \draw[dashed, thick] (q-1) -- (R-2);
      \draw[dashed, thick] (q-1) -- (R-3.155);
      \draw[dashed, thick] (q-3) -- (R-1.25);
      \draw[dashed, thick] (q-3) -- (R-2);
      \draw[dashed, thick] (q-2) -- (R-1.50);
      \draw[dashed, thick] (q-2) -- (R-3.130);
    \end{tikzpicture}
    \end{center}
    \subcaption{Construction of $G_1$: each $P_1^i$ is a clique of row vertices attached to $a_1^i$, and each $q_1^j$ attaches to $b_1^j$; dashed edges indicate optional adjacencies encoding vectors in $\mathcal{R}_i$.}
    \label{fig:lb:G1}
    \end{subfigure}
  
    \vspace{5ex}

    \begin{subfigure}{.9\linewidth}
        \begin{center}
    \begin{tikzpicture}[
        x=2cm,y=2cm,
        main/.style={circle,draw,fill=black,inner sep=2.4pt},
        vtx/.style={circle,draw,inner sep=2pt},
        rect/.style={draw,rounded corners,inner sep=4pt},
        every node/.style={font=\scriptsize},
        >={Latex[length=3mm,width=2mm]}
      ]
      % Columns for i = 1,2,3 (same horizontal layout)
      \foreach \row/\i in {0/1, 1.6/2, 3.2/3} {
        \begin{scope}[shift={(\row,0)}]
        % a_2^i from G_0 (below)
        \node[main,label=below left:{$a_2^{\i}$}] (a2-\i) at (0.6,-0.5) {};

        % p_2^i: single vertex in the middle
        \node[vtx, label=below left:{$p_2^{\i}$}] (p2-\i) at (0.6,0) {};

        % Q_2^i: four vertices on a horizontal line above
        \node[vtx] (q2-\i-1) at (0,0.7) {};
        \node[vtx] (q2-\i-2) at (.4,0.7) {};
        \node[vtx] (q2-\i-3) at (.8,0.7) {};
        \node[vtx] (q2-\i-4) at (1.2,0.7) {};

        % Rounded rectangle around Q_2^i
        \node[ultra thick,rect,inner ysep=8pt,fit=(q2-\i-1)(q2-\i-2)(q2-\i-3)(q2-\i-4)] (R2-\i) {};
        \node[font=\footnotesize] at (0, 1.15) {$Q_2^{\i}$};

        % Edge a_2^i -- p_2^i
        \draw (a2-\i) -- (p2-\i);

        % b_2^i above and its edges to all q_2^{i,*}
        \node[main,label=above right:{$b_2^{\i}$}] (b2-\i) at (0.6,1.4) {};
        \foreach \j in {1,...,4} {\draw (b2-\i) -- (q2-\i-\j);}    

        % Connect p_2^i to all q_2^{i,*}
        \foreach \j in {1,...,4} {\draw (p2-\i) -- (q2-\i-\j);}    
        \end{scope}
      }
      % Example cross-connections p_2^i -- Q_2^j according to columns
      \draw[dashed, thick] (p2-1) -- (R2-2);
      \draw[dashed, thick] (p2-1) -- (R2-3.205);
      \draw[dashed, thick] (p2-3) -- (R2-1.335);
      \draw[dashed, thick] (p2-3) -- (R2-2);
      \draw[dashed, thick] (p2-2) -- (R2-1.315);
      \draw[dashed, thick] (p2-2) -- (R2-3.225);
    \end{tikzpicture}
    \end{center}
        \subcaption{Construction of $G_2$: each $Q_2^i$ is an independent set attached to $p_2^i$ and $b_2^i$; dashed edges indicate optional adjacencies from $p_2^i$ to $Q_2^j$ encoding vectors in $\mathcal{C}_j$.}
        \label{fig:lb:G2}
    \end{subfigure}
    \caption{Constructions. Dashed lines indicate that edge may or may not exist.}
    \label{fig:mcis-lb:g0}
\end{figure}

\versionparagraph{Construction}

First, we define a base graph $G_0$.
In our construction, $G_0$ appears as an induced subgraph in both $G_1$ and $G_2$, and intuitively, it plays the role of fixing a correspondence between certain parts of $G_1$ and $G_2$; that is, in feasible solutions, each ``external'' vertex in $G_1$ necessarily corresponds to the respective vertex in $G_2$.

Now, we construct $G_0$, which is illustrated in \Cref{fig:lb:G0}.
Start with two vertex-disjoint paths $(a^1, \dots, a^k)$ and $(b^1, \dots, b^k)$, each on $k$ vertices.
For each $i \in [k]$, add four new vertices $a^{i,1},\dots,a^{i,4}$ and edges so that $(a^i, a^{i,1}, \dots, a^{i,4})$ forms a $5$-cycle.
For each $i \in [k]$, add six new vertices $b^{i,1},\dots,b^{i,6}$ and edges so that $(b^i, b^{i,1}, \dots, b^{i,6})$ forms a $7$-cycle.
For each $i \in [k]$, add the matching edge $a^i b^i$.
Finally, we add two new vertices $a^{1,0}$ and $b^{1,0}$ and add edges $a^{1,0} a^{1,1}$ and $b^{1,0} b^{1,1}$.
%Let $X = \{ a^i, b^i \mid i \in [k] \}$ and $Y = V(G_0) \setminus X$.

%See \Cref{fig:mcis-lb:g0} for an illustration of $G_0$.

Now, we construct $G_1$ (see \Cref{fig:lb:G1}).
We start with a graph isomorphic to $G_0$, renaming the vertex names by adding a subscript $\cdot_1$.
\begin{itemize}
    \item
    For each $i \in [k]$ and $v \in \mathcal{R}_i$, we add a vertex $p_{1}^{i,v}$ and add an edge between $a_1^i$ and $p_1^{i,v}$.
    \item
    For each $i \in [k]$, we add edges so that $P_1^i := \{ p_{1}^{i,v} \mid v \in \mathcal{R}_i \}$ forms a clique.
    \item
    For each $i \in [k]$, we add a vertex $q_1^i$ and add an edge between  $b_1^i$ and $q_1^i$.
    \item
    For each $i, j \in [k]$ and for each $v \in \mathcal{R}_i$, we add an edge between $p_{1}^{i,v}$ and $q_{1}^{j}$ if $v[j] = 1$.
\end{itemize}
Since we assumed that the diagonal entries are all 1,
$P_1^i \cup \{ q_1^i \}$ is a clique for all $i \in [k]$.
Let $P_1 = \bigcup_{i \in [k]} P_1^i$, and $Q_1 = \{ q_1^i \mid i \in [k] \}$.

Next, we construct $G_2$ (see \Cref{fig:lb:G2}).
We start with a graph isomorphic to $G_0$, renaming the vertex names by adding a subscript $\cdot_2$.
\begin{itemize}
    \item
    For each $i \in [k]$, we add a vertex $p_{2}^i$ and add an edge between $a_2^i$ and $p_2^i$.
    \item
  For each $i \in [k]$ and $v \in \mathcal{C}_i$, we add a vertex $q_{2}^{i,v}$ and add edges $p_2^i q_2^{i,v}$ and $b_2^i q_2^{i,v}$.
    \item
    For each $i, j \in [k]$ and for each $v \in \mathcal{C}_i$, we add an edge between $p_{2}^{i}$ and $q_{2}^{j,v}$ if $v[j] = 1$.
\end{itemize}
Since we assumed that the diagonal entries are all $1$, the vertices $p_2^i$ are adjacent to every~$q_2^{i,v}$.
Let $Q_2^i = \{ q_2^{i,v} \mid v \in \mathcal{C}_i \}$ and $P_2 = \{ p_2^i \mid i \in [k] \}$, so each $\{p_2^i\} \cup Q_2^i$ is a star.
Finally, let $Q_2 = \bigcup_{i \in [k]} Q_2^i$.

We can easily see the following.
\begin{lemma}
$\cvd(G_1),\cvd(G_2)\leq 13k$.
\end{lemma}
\begin{proof}
$(V(G_0)_1\setminus \{a^{1,0}_1,b^{1,0}_1\})\cup Q_1$ is a cluster vertex deletion set of $G_1$ with size $13k$.
%\todo{Isn't $|V(G_0)_1\cup Q_1|=13k+2$? Btw, since we don't need $a^{1,0}_{1}$ and $b^{1,0}_{1}$, $\cvd(G_1) \le 13k$ holds. (Same for $G_{2}$.)}\todo{fixed}
$(V(G_0)_2\setminus \{a^{1,0}_2,b^{1,0}_2\})\cup P_2$ is a cluster vertex deletion set of $G_2$ with size $13k$.
\end{proof}

\versionparagraph{Correctness}

Let $\mathsf{mcis}(G_1,G_2)$ denote the maximum size of a common induced subgraph of $G_1$ and $G_2$.
From now on, we concentrate on proving that the \textsc{LCBM} instance (with additional constraint of all-$1$ diagonal) is feasible if and only if $\mathsf{mcis}(G_1,G_2)\geq t:=14k+1$.
We begin with the forward direction.
\begin{lemma}
Assume the \textsc{LCBM} instance is feasible. Then, $\mathsf{mcis}(G_1,G_2)\geq t = 14k+1$.
\end{lemma}
%    $(\Rightarrow)$
%Assume the \textsc{LCBM} instance is feasible and
\begin{proof}
Let $A\in\{0,1\}^{k\times k}$ be a witnessing matrix.
For each $i\in[k]$, let $r_i\in\mathcal{R}_i$ denote the $i$-th row of $A$, and for each $j\in[k]$, let $c_j\in\mathcal{C}_j$ denote the $j$-th column of $A$, so $r_i[j]=A[i,j]=c_j[i]$ for all $i,j$.

Define the vertex sets
\begin{align*}
U_1 \;&:=\; V(G_0)_1 \;\cup\; \{\,p_1^{i,r_i} \mid i\in[k]\,\} \;\cup\; \{\,q_1^j \mid j\in[k]\,\}
\text{ and } \\
U_2 \;&:=\; V(G_0)_2 \;\cup\; \{\,p_2^{i} \mid i\in[k]\,\} \;\cup\; \{\,q_2^{j,c_j} \mid j\in[k]\,\}.
\end{align*}
Let $H_1:=G_1[U_1]$ and $H_2:=G_2[U_2]$.
We claim that $H_1$ and $H_2$ are isomorphic.
Since $|U_1|=|U_2|=|V(G_0)|+2k=14k+1$ (by the definition of $G_0$ we have $|V(G_0)|=12k+1$), this yields a common induced subgraph of size $t$.

Construct a bijection $\psi \colon U_1\to U_2$ as follows.
\begin{itemize}
  \item Map the copy of $G_0$ in $G_1$ isomorphically onto the copy of $G_0$ in $G_2$.
  \item For each $i\in[k]$, set $\psi(p_1^{i,r_i}):=p_2^{i}$.
  \item For each $j\in[k]$, set $\psi(q_1^{j}):=q_2^{j,c_j}$.
\end{itemize}

We verify that $\psi$ is an isomorphism between the induced subgraphs as follows.
\begin{itemize}
  \item Adjacencies inside $G_0$ are preserved by the first item.
  \item For each $i$, $p_1^{i,r_i}$ is adjacent to $a_1^i$ (and to no other vertex of $G_0$), and $p_2^i$ is adjacent to $a_2^i$ (and to no other vertex of $G_0$); hence $\psi$ preserves all incidences between the $p$-vertices and $G_0$.
  \item For each $j$, $q_1^j$ is adjacent to $b_1^j$ (and to no other vertex of $G_0$), and $q_2^{j,c_j}$ is adjacent to $b_2^j$ (and to no other vertex of $G_0$); hence $\psi$ preserves all incidences between the $q$-vertices and $G_0$.
  \item There are no edges among distinct $p$-vertices in either $H_1$ or $H_2$, and no edges among distinct $q$-vertices; this follows from the construction (in $G_1$ we take exactly one vertex from each clique $P_1^i$, and there are no edges between different $P_1^i$'s; in $G_2$ the vertices $p_2^i$ and the $q_2^{j,\cdot}$ only connect as specified below).
  Hence $\psi$ preserves these non-adjacencies.
  \item Finally, for every $i,j$,
  \[
  \{p_1^{i,r_i},q_1^{j}\}\in E(G_1)
  \iff r_i[j]=1
  \iff c_j[i]=1
  \iff \{p_2^{i},q_2^{j,c_j}\}\in E(G_2).
  \]
  Thus $\psi$ preserves exactly the $p$-$q$ adjacencies.
\end{itemize}
Since all edges among the vertices of $U_1$ (and $U_2$) are covered by the items above, $\psi$ is an isomorphism.
Therefore, the constructed \textsc{MCIS} instance is a \yes-instance whenever the \textsc{LCBM} instance is feasible.
\end{proof}

Now, we prove the backward direction.
The following lemma is easy but useful.
\begin{lemma} \label{lem:mcis-bound}
For any graphs $G_1,G_2$ and any $S\subseteq V(G_1)$,
\[
    \mathsf{mcis}(G_1,G_2)\ \le\ \mathsf{mcis}(G_1-S,G_2)\ +\ |S|.
\]
If equality holds, then $S$ is fully used by every maximum common induced subgraph, i.e., $S\subseteq \phi_1(V(H))$ for any maximum $H$ with induced embedding $\phi_1$.
\end{lemma}
\begin{proof}
Let $H$ be a maximum common induced subgraph of $G_1$ and $G_2$ with induced embeddings $\phi_i\colon V(H)\to V(G_i)$.
Set $Z := \{x\in V(H) : \phi_1(x)\notin S\}$.
Then $H[Z]$ is a common induced subgraph of $G_1-S$ and $G_2$ (witnessed by the restrictions $\phi_i|_{Z}$), hence $|Z|\le \mathsf{mcis}(G_1-S,G_2)$.
So, we have
\[
    |V(H)| \;=\; |Z| \;+\; |\phi_1(V(H))\cap S|
    \;\le\; \mathsf{mcis}(G_1-S,G_2) \;+\; |S|.
\]
If equality $\mathsf{mcis}(G_1,G_2) = \mathsf{mcis}(G_1-S,G_2) + |S|$ holds, then $|\phi_1(V(H))\cap S|=|S|$, which forces $S\subseteq \phi_1(V(H))$; in particular, $S$ is fully used.
\end{proof}

Now, we prove the following.
\begin{lemma}
Assume $\mathsf{mcis}(G_1,G_2)\geq t = 14k+1$.
Then, the \textsc{LCBM} instance is feasible.
\end{lemma}
\begin{proof}

    Let $H$ be a common induced subgraph of $G_1$ and $G_2$ of size $t=14k+1$ with induced embeddings $\phi_i \colon V(H)\to V(G_i)$ ($i=1,2$).
    We show that the original \textsc{LCBM} instance is a yes-instance by constructing a feasible matrix.

    First, we show that the whole copy of $G_0$ is used on both sides.
    Let $S_1:=V(G_0)_1$ and $S_2:=V(G_0)_2$.
    We observe that $G_2$ is triangle-free as follows:
    \begin{itemize}
        \item $G_0$ is triangle-free by construction.
        \item $G_2-S_2$ is bipartite with bipartition $(P_2,Q_2)$, hence triangle-free.
        \item Two vertices in $S_2$ do not share a neighbor in $P_2 \cup Q_2$.
        \item Two vertices in $P_2 \cup Q_2$ share a neighbor in $S_2$ if and only if they are both part of $Q_2^i$ (which is an independent set) for some $i \in [k]$. 
    \end{itemize}
    Consequently, $G_2$ is triangle-free, and any common induced subgraph of $G_1 - S_1$ and $G_2$ must be triangle-free.
    %\todo{couldn't understand. Which of $\mathsf{mcis}(G_1-S_1,G_2)$ and $\mathsf{mcis}(G_1,G_2-S_2)$ are you talking about?}
    %\todo{Fixed?}
    Hence, for a clique $C$ in $G_1$, at most two vertices from $C$ can be part of the common subgraph.
    Observe that the vertex set of $G_1-S_1$ can be partitioned into $k$ cliques $P_1^i \cup \{ q_1^i \}$ for $i \in [k]$.
    Therefore, we have $\mathsf{mcis}(G_1-S_1,G_2) \le 2k$.
    By Lemma~\ref{lem:mcis-bound} and $|S_1|=|V(G_0)|=12k+1$, we get
    \[
      \mathsf{mcis}(G_1,G_2)\ \le\ 2k + (12k+1)\ =\ 14k+1\ =\ |V(H)|.
    \]
    Hence equality holds throughout, and Lemma~\ref{lem:mcis-bound} implies $S_1$ is fully used, i.e., $S_1 \subseteq \phi_1(V(H))$.

    In particular, $H$ contains a copy of $G_0$ as an induced subgraph.
    By construction, $G_0$ contains $2k$ odd vertex-disjoint odd cycles $(a^i, a^{i,1}, \dots, a^{i,4})$ and $(b^i, b^{i,1}, \dots, b^{i,6})$ for $i \in [k]$.
    Observe that $G_2-\{ a^i_2, b^i_2 \mid i \in [k] \}$ is bipartite.
    So, for each of the $2k$ $5$-/$7$-cycles $C$ in $H$, $\phi_2(V(C))$ contains exactly one vertex from $\{ a^i_2, b^i_2 \mid i \in [k] \}$.
    By construction of $G_2$, for each $i \in [k]$, there is only one odd cycle $C_2$ with $V(C_2) \cap \{ a^i_2, b^i_2 \mid i \in [k] \} = \{ a_2^i \}$, namely, $(a^i_2, a^{i,1}_2, \dots, a^{i,4}_2)$, and only one odd cycle $C_2$ with $V(C_2) \cap \{ a^i_2, b^i_2 \mid i \in [k] \} = \{ b_2^i \}$, namely, $(b^i_2, b^{i,1}_2, \dots, b^{i,6}_2)$.
    Thus $S_2$ is fully used as well.
    The ladder formed by the two paths and the matching edges is bipartite, so the only odd cycles in $G_0$ are the attached $5$- and $7$-cycles. Hence any induced isomorphism between the copies of $G_0$ maps $a$-vertices (on $5$-cycles) to $a$-vertices and $b$-vertices (on $7$-cycles) to $b$-vertices. Moreover, the unique leaves $a^{1,0}$ and $b^{1,0}$ mark the cycles attached to $a^1$ and $b^1$, fixing these endpoints. The path structure then forces $a_1^i$ to correspond to $a_2^i$ and $b_1^j$ to $b_2^j$ for all $i,j\in[k]$.

     Next, we argue that exactly one $p$-vertex per $i$ and exactly one $q$-vertex per $j$ are used.
    Consider the residual pair $G_1' := G_1 - S_1$ and $G_2' := G_2 - S_2$.
    %Since $P_2$ is a vertex cover of $G_2'$ (all edges of $G_2'$ go between $P_2$ and $Q_2$), 
    Lemma~\ref{lem:mcis-bound} gives
    \[
      \mathsf{mcis}(G_1',G_2') \;\le\; \mathsf{mcis}(G_1',\,G_2'-P_2) + |P_2|.
    \]
    But since $V(G_1')$ can be partitioned into $k$ cliques and $V(G_2') \setminus P_2 = Q_2$ is an independent set,
    $\mathsf{mcis}(G_1',\, G_2' - P_2) \le k$, and therefore
    \[
      \mathsf{mcis}(G_1',G_2') \;\le\; k + |P_2| \;=\; 2k.
    \]
    Since $|V(H)| = |S_1| + 2k$, we have $\mathsf{mcis}(G_1',G_2') = 2k$ and equality holds in the lemma. Thus $P_2$ is fully used by $H$.
    Let $a^i \in V(H)$ be the vertex with $\phi_1(a^i)=a_1^i$ (equivalently, $\phi_2(a^i)=a_2^i$), and let $p^i$ be the vertex with $\phi_2(p^i) = p_2^i$.
    Since $a_2^i$ and $p_2^i$ are adjacent in $G_2$, $a^i$ and $p^i$ are adjacent in $H$.
    Thus, $\phi_1(a^i)=a_1^i$ is adjacent to $\phi_1(p^i)$, that is, $\phi_1(p^i) = p_1^{i,r_i}$ for some $r_i \in \mathcal{R}_i$.

    Similarly, fix $j\in[k]$. 
    In $G_2$ the neighbors of $b_2^j$ outside $S_2$ are precisely the vertices $\{q_2^{j,v}: v\in\mathcal{C}_j\}$.
    From the equality case above we already know that $|V(H)\setminus \phi_2^{-1}(S_2)|=2k$ and that $P_2$ is fully used; hence exactly $k$ vertices of $Q_2$ are used by $H$.
    If for some $j$ the subgraph $H$ used two distinct vertices $q_2^{j,v}\neq q_2^{j,v'}$, then (since both are adjacent to $b_2^j$ in $G_2$) the vertex of $H$ mapped to $b_2^j$ would have two neighbors outside $\phi_2^{-1}(S_2)$.
    % Because $\phi_1$ and $\phi_2$ are induced embeddings of the same graph $H$, this would force the vertex of $H$ mapped to $b_1^j$ to have two neighbors outside $\phi_1^{-1}(S_1)$, 
    This is impossible because in $G_1$ the only such neighbor of $b_1^j$ is $q_1^j$.
    Therefore, for each $j$ at most one vertex from $\{q_2^{j,v}: v\in\mathcal{C}_j\}$ is used; combined with the fact that exactly $k$ vertices of $Q_2$ are used in total, it follows that for each $j$ there is exactly one $c_j\in\mathcal{C}_j$ with $\phi_2^{-1}(q_2^{j,c_j})\in V(H)$, and similarly $\phi_1^{-1}(q_1^j)\in V(H)$ is the unique outside neighbor of $\phi_1^{-1}(b_1^j)=b^j$.
    In particular, $ \phi_1^{-1}(q_1^j) = \phi_2^{-1}(q_2^{j,c_j}) $ for every $j\in[k]$.

    Finally, we are ready to extract a feasible matrix for \textsc{LCBM}.
    Define a $k\times k$ Boolean matrix $A$ by taking, for each $i\in[k]$, its $i$-th row to be $r_i\in\mathcal{R}_i$, and, for each $j\in[k]$, its $j$-th column to be $c_j\in\mathcal{C}_j$.
    It remains to verify consistency, i.e., that for all $i,j$,
    \[
      r_i[j]\ =\ A[i,j]\ =\ c_j[i].
    \]
    By construction, we have
    \[
      \{p_1^{i,r_i}, q_1^j\}\in E(G_1) \iff r_i[j]=1
      \quad\text{and}\quad
      \{p_2^{i}, q_2^{j,c_j}\}\in E(G_2) \iff c_j[i]=1
    \]
    Since $\phi_1^{-1}(p_1^{i,r_i}) = \phi_2^{-1}(p_2^i)$ and $\phi_1^{-1}(q_1^j) = \phi_2^{-1}(q_2^{j,c_j})$,
    we obtain
    \begin{align*}
      \{p_1^{i,r_i}, q_1^j\}\in E(G_1)
      &\iff
      \{\phi_1^{-1}(p_1^{i,r_i}), \phi_1^{-1}(q_1^j)\}\in E(H) \\
      &\iff 
      \{\phi_2^{-1}(p_2^{i}),\phi_2^{-1}(q_2^{j,c_j}) \}\in E(H)
      \iff
      \{p_2^{i}, q_2^{j,c_j}\}\in E(G_2).
    \end{align*}
    Therefore $r_i[j]=c_j[i]$ for all $i,j$, so $A[i,\cdot]=r_i\in\mathcal{R}_i$ and $A[\cdot,j]=c_j\in\mathcal{C}_j$ simultaneously, i.e., the \textsc{LCBM} instance is feasible.
    %This completes the proof of the backward direction.
\end{proof}

\versionparagraph{Analysis}

Now we are ready prove our main theorem of this section.

\mcislb*
\begin{proof}
Given an instance of \textsc{LCBM}, the construction in this section yields $2^k$ instances of \textsc{MCIS} with $|V(G_1)|+|V(G_2)|\in 2^{O(k)}$ and $\cvd(G_1)+\cvd(G_2)\leq O(k)$ such that the \textsc{LCBM} instance is an yes-instance if and only if at least one of \textsc{MCIS} instances is an yes-instance.
Thus, if \textsc{MCIS} can be solved in $O^*(2^{o(k^2)})=2^{o(k^2)}\cdot 2^{O(k)}$ time for $k=\cvd(G_1)+\cvd(G_2)$, \text{LCBM} can also be solved in $2^{o(k^2)}\cdot 2^{O(k)}\cdot 2^{O(k)}=2^{o(k^2)}$ time, contradicting to ETH.
\end{proof}

%\begin{proposition}
%    Given an instance of the matrix problem, one can construct in time $2^{O(k)}$ an equivalent instance of \textsc{MCIS} where $|V(G_1)| + |V(G_2)| \in 2^{O(k)}$ and $\cvd(G_1) + \cvd(G_2) \in O(k)$.
%\end{proposition}

\section{On the complexity of 3-MCIS}

In this section, we study \prob{$3$-MCIS}, the problem of finding a maximum common induced subgraph of three input graphs.
In \Cref{ssec:3mcis-cvd} we give a simple linear-time algorithm for the case when all three graphs are cluster graphs and show NP-hardness even when each input graph $G_i$ satisfies $\cvd(G_i)=2$.
The case in which the input graphs have cluster deletion number~1 remains open.
In \Cref{ssec:3mcis-vc} we give an FPT algorithm parameterized by vertex cover number.

\subsection{Parameterizing by Cluster Vertex Deletion Number} \label{ssec:3mcis-cvd}

We first observe that \textsc{3-MCIS} can be solved in polynomial time on cluster graphs.

% --- 3MCIS on cluster graphs: tight characterization ---
\begin{lemma}\label{lem:3mcis-cluster}
Let $G_1,G_2,G_3$ be cluster graphs. 
For each $i\in\{1,2,3\}$, let $c_i$ be the number of connected components of $G_i$, and let the component sizes, sorted nonincreasingly, be $n_i^1\ge n_i^2\ge \dots \ge n_i^{c_i}$. 
Adopting the convention $n_i^j=0$ for all $j>c_i$, let $n^j = \min \{n_i^j, n_2^j, n_3^j \}$.
The instance of \textnormal{\textsc{3-MCIS}} with target size $t$ is a \yes-instance if and only if
\(
  t \le \sum_{j\ge 1} n_j.
\)
\end{lemma}

\begin{proof}
($\Rightarrow$) Let $t:=\sum_{j\ge 1} n_j$ and build a cluster graph $H$ of $t$ vertices that has, for each $j$ with $n_j>0$, one clique of size $n_j$. 
We obtain an induced embedding by mapping the $j$-th clique of $H$ into the $j$-th largest component of each $G_i$. 

($\Leftarrow$) Conversely, let $H$ be any common induced subgraph of $G_1,G_2,G_3$. 
Let its component sizes, sorted nonincreasingly, be $h^1\ge h^2\ge\dots\ge h^r$. 
Fix $i\in\{1,2,3\}$ and any $j\in\{1,\dots,r\}$. 
The $j$ largest components of $H$ must be mapped to $j$ distinct components of $G_i$. 
Among those $j$ target components, the smallest size is at most $n_i^j$ (since the $j$-th largest component of $G_i$ upper-bounds the minimum over any $j$ chosen components).
Because each mapped component of $H$ must fit entirely inside its target clique, we have $h^j\le n_i^j$.
As this holds for every $i$, we obtain $h^j\le \min\{n_1^j,n_2^j,n_3^j\}=n_j$ for all $j$. 
\end{proof}

The following is immediate from Lemma~\ref{lem:3mcis-cluster}.

\begin{theorem}
    \textsc{3-MCIS} is polynomial-time solvable when $G_i$ is a cluster graph  for all $i \in \{ 1, 2, 3 \}$.
\end{theorem}

From now on, we prove the following:
\tmcislb*

\begin{proof}

We prove NP-hardness by a polynomial-time reduction from \textsc{3-SAT}. 
Given a 3-CNF formula $\Phi$, we construct three graphs $G_1,G_2,G_3$ with $\cvd(G_i)=2$ such that
$\Phi$ is satisfiable if and only if there exists a common induced subgraph of size
\[
t \;=\; \sum_{j \ge 1} n^j \;+\; 2,
\]
where $n^j:=\min\{n_1^j,n_2^j,n_3^j\}$ as defined below.

Let $X_i=\{v_i^1,v_i^2\}$ be a cluster vertex deletion set of size $2$ for $G_i$.
As the notation suggests, the intended embeddings map the vertices $v_i^1$ to one another and $v_i^2$ to one another across the three graphs.
By definition, $G_i-X_i$ is a cluster graph, i.e., every connected component is a clique.
As in \Cref{lem:3mcis-cluster}, label the components $C_i^1,C_i^2,\dots$ in nonincreasing order of size and set $n_i^j:=|C_i^j|$ and $n_j:=\min\{n_1^j,n_2^j,n_3^j\}$.
By \Cref{lem:3mcis-cluster}, any common induced subgraph contained entirely in the cluster parts has size at most $\sum_{j\ge 1} n_j$; hence the choice of $t$ forces two additional vertices into $H$.
% \todo{[Check] It may need additional care; it may be possible that, only $G_1$ uses the modulator}

To control adjacencies to $v_i^1$ and $v_i^2$, we encode each clique $C\subseteq G_i-X_i$ by its signature
\[
\sigma_i(C)=(a,b,c,d),
\]
where
\[
\begin{aligned}
a&=|C\cap (N_{G_i}(v_i^1)\cap N_{G_i}(v_i^2))|, &
b&=|C\cap (N_{G_i}(v_i^1)\setminus N_{G_i}(v_i^2))|,\\
c&=|C\cap (N_{G_i}(v_i^2)\setminus N_{G_i}(v_i^1))|, &
d&=|C\setminus (N_{G_i}(v_i^1)\cup N_{G_i}(v_i^2))|.
\end{aligned}
\]
The clique size is $\|\sigma_i(C)\|_1=a+b+c+d$. Let $\mathcal{V}_i$ be the multiset of signatures of all components of $G_i-X_i$.
Specifying $\mathcal{V}_i$ together with the edges from $\{v_i^1,v_i^2\}$ determined by each signature fixes $G_i$ up to isomorphism for our purposes.
In the reduction, we prescribe $\mathcal{V}_1,\mathcal{V}_2,\mathcal{V}_3$ so that a common induced subgraph of size $t$ exists if and only if $\Phi$ is satisfiable.

We reduce from the bounded-occurrence variant \textsc{(3,B2)-SAT}: every clause has exactly 3 literals and each variable appears exactly twice positively and twice negatively.
This variant is NP-complete \cite{DBLP:journals/eccc/ECCC-TR03-049}.
Let $\Phi$ be such an instance with variables $x_1,\dots,x_n$ and clauses $C_1,\dots,C_m$.
Counting literal occurrences yields $4n=3m$, hence $n$ is a multiple of $3$ and
\(
m = \tfrac{4n}{3}.
\)

\begin{figure}
    \centering
    \begin{tikzpicture}[scale=0.13, every node/.style={font=\small}, >={Latex[length=4mm,width=3mm]}]
    \draw[thick,->] (-1,0) -- (37,0) node[below right] {$a$};
    \draw[thick,->] (0,-1) -- (0,37) node[above left] {$b$};
    \draw[step=10,help lines] (-0.5,-0.5) grid (33,33);

    \tikzset{
        V1/.style={draw=black,fill=black,inner sep=2pt, minimum size=5pt},
        V2/.style={draw=red!75,fill=red!50,rectangle,inner sep=2pt, minimum size=5pt},
        V3/.style={draw=blue!75,fill=blue!50,rectangle,inner sep=2pt, minimum size=5pt}
    }

    \node[V2,label={[red]above right:$b_{i,0}$}] at ($(0,0)+(0.2,0.2)$) {};
    \node[V2,label={[red]above right:$b_{i,1}$}] at (10,30) {};
    \node[V2,label={[red]above left:$b_{i,2}$}] at (30,10) {};

    \node[V3,label={[blue]right:$c_{i,0}$}] at ($(0,30)+(0.2,0.2)$) {};
    \node[V3,label={[blue]right:$c_{i,1}$}] at (20,20) {};
    \node[V3,label={[blue]above left:$c_{i,2}$}] at ($(30,0)+(0.2,0.2)$) {};

    \node[V1,label={below left:$a_{i,0}$}] at ($(0,0)-(0.2,0.2)$) {};
    \node[V1,label={above left:$a_{i,1}$}] at ($(0,30)-(0.2,0.2)$) {};
    \node[V1,label={above left:$a_{i,2}$}] at (10,20) {};
    \node[V1,label={above left:$a_{i,3}$}] at (20,10) {};
    \node[V1,label={below right:$a_{i,4}$}] at ($(30,0)-(0.2,0.2)$) {};

    \end{tikzpicture}
    \caption{Variable gadget positions in the $(x,y)$-plane. Only the first two coordinates are shown; the last two coordinates are $N^2$ (or $(p_{i,t}N,(n-p_{i,t})N)$), and the shift $v_i$ applies in all four dimensions.}
    \label{fig:variable-gadget}
\end{figure}

\versionparagraph{Construction}
Fix a scaling parameter $N := 1000n$. We first create a unique \emph{anchor} clique that pins which of the two modulator vertices is $v_i^1$ in all three graphs:
In each $G_i$, add a clique on $N^4$ vertices, where $N^4-10$ vertices $u_i^1, \dots, u_i^{N^4-10}$ are adjacent to $v_i^1$ and nonadjacent to $v_i^2$, and the remaining ten vertices $w_i^1, \dots, w_i^{10}$ are nonadjacent to $v_i^1$ and adjacent to $v_i^2$.
In signature notation this corresponds to the single vector $(0,\,N^4-10,\,10,\,0)\in\mathcal{V}_i$.

\versionparagraph{Variable gadgets}
For each variable $x_i$ ($i\in[n]$), suppose that $x_i$ appears positively in clauses $C_{p_{i,1}}$ and $C_{p_{i,2}}$ and negatively in clauses $C_{q_{i,1}}$ and $C_{q_{i,2}}$.
Define the \emph{offset} 
\[
v_i \ :=\ (iN^2,\ (n-i)N^2,\ 0,\ 0).
\]
This separates the $i$-th variable block from all others in the first two coordinates.

Add to $\mathcal{V}_1$ the five vectors
\[
\begin{aligned}
a_{i,0}&=(0,\ 0,\ N^2,\ N^2)+v_i,\\
a_{i,1}&=(0,\ 3,\ p_{i,1}N,\ (m-p_{i,1})N)+v_i,\\
a_{i,2}&=(1,\ 2,\ q_{i,1}N,\ (m-q_{i,1})N)+v_i,\\
a_{i,3}&=(2,\ 1,\ p_{i,2}N,\ (m-p_{i,2})N)+v_i,\\
a_{i,4}&=(3,\ 0,\ q_{i,2}N,\ (m-q_{i,2})N)+v_i.
\end{aligned}
\]
Add to $\mathcal{V}_2$ the three vectors
\[
\begin{aligned}
b_{i,0}&=(0,\ 0,\ N^2,\ N^2)+v_i,\\
b_{i,1}&=(1,\ 3,\ mN,\ mN)+v_i,\\
b_{i,2}&=(3,\ 1,\ mN,\ mN)+v_i.
\end{aligned}
\]
Add to $\mathcal{V}_3$ the three vectors
\[
\begin{aligned}
c_{i,0}&=(0,\ 3,\ N^2,\ N^2)+v_i,\\
c_{i,1}&=(2,\ 2,\ mN,\ mN)+v_i,\\
c_{i,2}&=(3,\ 0,\ N^2,\ N^2)+v_i.
\end{aligned}
\]
See \Cref{fig:variable-gadget} for the $(a,b)$-plane placement (only the first two coordinates are shown after adding $v_i$).

\versionparagraph{Clause gadgets}
For each $j\in[m]$, add to $\mathcal{V}_2$ the vector
\[
b_j := (0,\ 0,\ jN,\ (m-j)N),
\]
and to $\mathcal{V}_3$ the vector
\[
c_j := (0,\ 0,\ jN,\ (m-j)N).
\]
This concludes the construction.

For each $\mathcal{V}_i$, there is a vector with the $\ell_1$-norm of $N^4$.
In $\mathcal{V}_1$, there are $n$ vectors whose $\ell_{1}$-norm is $(n+2)N^2$, and $4n$ vectors have $\ell_1$-norm $nN^2 + mN + 3$.
In $\mathcal{V}_2$, $\ell_1$-norms are:
$(n + 2)N^2$ for $n$ vectors, 
$nN^2 + 2mN + 4$ for $2n$ vectors,
 and $mN$ for $m$ vectors.
In $\mathcal{V}_3$, $\ell_1$-norms are:
$(n+2)N^2 + 3$ for $2n$ vectors, $nN^2 + 2mN + 4$ for $n$ vectors and $mN$ for $m$ vectors.
So, 
\[
  n^j = \begin{cases}
    N^4 & \text{ if $j = 0$} \\
    (n + 2)N^2 &  \text{ if $1 \le j \le n$} \\
    nN^2 + mN + 3 & \text{ if $n + 1 \le j \le 3n$} \\
    mN & \text{ if $3n + 1 \le 3n + m$}
  \end{cases} \quad\text{ and }
\]
\[
t \;=\; N^4 \;+\; n\cdot\bigl((n+2)N^2\bigr) \;+\; 2n\cdot\bigl(nN^2 + mN + 3\bigr) \;+\; m\cdot(mN) \;+\; 2.
\]

\versionparagraph{Correctness}
We now prove correctness.
We first establish the correspondence among the $v_i^1$'s and among the $v_i^2$'s for the converse direction, which allows us to continue the proof with signature vectors.

\subparagraph*{Correspondence.}
% Suppose that $H$ is a common induced subgraph on $t$ vertices with induced embeddings into $G_i$ for all $i\in[3]$.
% Since $\cvd$ is induced-subgraph-monotone and $\cvd(G_i)=2$, we have $\cvd(H)\le 2$; let $X$ be a minimum cluster vertex deletion set of $H$.
% By \Cref{lem:3mcis-cluster} and the choice of $t$, any common induced subgraph contained entirely in the cluster parts has size at most $t-2$, hence $|X|=2$.
% Applying \Cref{lem:mcis-bound} to the pairs $(H,G_i)$ and deleting $X$ and $X_i$, we get $\mathsf{mcis}(H-X,G_i-X_i)\ge t-4$, so $H-X$ contains a clique component $C^\star$ of size at least $N^4-4$.
% This forces $C^\star$ to embed into the unique anchor clique of each $G_i$, since all other cliques are much smaller.
% If one of the two vertices of $X$ were mapped to $v_i^1$ (resp.\ $v_i^2$) while the other were not mapped to $v_i^2$ (resp.\ $v_i^1$), then any common induced subgraph containing that vertex would have to exclude all $w_i^{j'}$ (resp.\ all $u_i^{j}$) to avoid the induced $P_3$'s $(v_i^1,u_i^j,w_i^{j'})$ or $(v_i^2,w_i^{j'},u_i^{j})$, losing at least $10$ vertices from the anchor clique.
% This contradicts the bound $|C^\star|\ge N^4-4$.
% Therefore $X$ is fully used and its two vertices must map to $v_i^1$ and $v_i^2$ in each $G_i$, yielding the desired correspondence among all $v_i^1$'s and among all $v_i^2$'s.

% \begin{comment}
Suppose that $H$ is a common induced subgraph on $t$ vertices with induced embeddings into $G_i$ for all $i\in[3]$.
Since $\cvd$ is induced-subgraph-monotone and $\cvd(G_i)=2$, we have $\cvd(H)\le 2$; let $X$ be a minimum cluster vertex deletion set of $H$.
Applying \Cref{lem:mcis-bound} to the pairs $(H,G_i)$ and deleting $X$ and $X_i$, we obtain
\[
  \mathsf{mcis}(H-X,G_i-X_i)\ge \mathsf{mcis}(H,G_i) - |X| - |X_i| = t-4.
\]
In particular, $H-X$ contains a clique component $C$ of size at least $N^4-4$.
Let $C^{\star}$ be a maximal clique of $H$ containing $C$.
Any vertex of $C^{\star}\setminus C$ is adjacent to all vertices of $C$, so under any embedding into $G_i$ it must map to a vertex adjacent to all of the image of $C$; thus $C^{\star}$ embeds into the unique anchor clique of each $G_i$, since no vertex outside the anchor clique has that many neighbors inside it.

Let $N_H(C^{\star})$ denote the set of vertices of $H$ adjacent to some vertex of $C^{\star}$.
In $G_i$, the only vertices adjacent to the anchor clique are $v_i^1$ and $v_i^2$, so every vertex of $N_H(C^{\star})$ maps into $X_i$.
Conversely, because $|C^{\star}|\ge N^4-4$, its image contains at least $6$ of the $10$ vertices $w_i^1,\dots,w_i^{10}$, hence both $v_i^1$ and $v_i^2$ are adjacent to $C^{\star}$; thus every vertex mapped to $X_i$ lies in $N_H(C^{\star})$.
%Therefore the two vertices of $X$ are exactly the preimages of $v_i^1$ and $v_i^2$.
Therefore, the vertices of $N_H(C^{\star})$ are exactly the preimages of $v_i^1$ and $v_i^2$.

%Moreover, one of these vertices is adjacent to at least $N^4-14$ vertices of $C^{\star}$ (corresponding to $v_i^1$), while the other is adjacent to at most $10$ vertices (corresponding to $v_i^2$).
Moreover, each vertex in $N_H(C^{\star})$ is either adjacent to at least $N^4-14$ vertices of $C^{\star}$ (corresponding to $v_i^1$), or adjacent to at most $10$ vertices (corresponding to $v_i^2$).
This distinguishes the two and forces a consistent correspondence among all $v_i^1$'s and among all $v_i^2$'s.
In particular, removing $v_i^1$ and $v_i^2$ from each $G_i$ decreases the $\mathsf{mcis}$ value by at most two; this ensures the image of $N_H(C^{\star})$ is exactly $\{v_i^1,v_i^2\}$.

% \end{comment}

\subparagraph*{Correctness via signature vectors.}
With the correspondence established, we now turn to the signature-vector analysis.
When a clique of $H$ is mapped to $C_i$ in $G_i$ for all $i\in\{1,2,3\}$, we say that $(C_1,C_2,C_3)$ is a \emph{matched triple}. If $\sigma_i(C_i)=(a_i,b_i,c_i,d_i)$, its \emph{contribution} is
\[
\mathrm{contr}(C_1,C_2,C_3)
:=\min\{a_1,a_2,a_3\}+\min\{b_1,b_2,b_3\}+\min\{c_1,c_2,c_3\}+\min\{d_1,d_2,d_3\}.
\]
We call a matched triple \emph{full} if $\mathrm{contr}(C_1,C_2,C_3)=\min\{|C_1|,|C_2|,|C_3|\}$.

With this notation, the constructed instance is a yes-instance if and only if there exists a family $\mathcal{M}$ of mutually disjoint matched triples such that: (i) exactly $n$ triples in $\mathcal{M}$ have full contribution $(n+2)N^2$ (one per variable, the $a_{i,0}$-type); (ii) exactly $2n$ triples in $\mathcal{M}$ have full contribution $nN^2+mN+3$ (the two occurrence-types per variable, from $a_{i,1}, a_{i,2}, a_{i,3}, a_{i,4}$); and (iii) for each $j\in[m]$ there is exactly one triple in $\mathcal{M}$ that contains $b_j$ and $c_j$, whose (full) contribution is $mN$. 

\medskip
($\Rightarrow$) Fix a satisfying assignment of $\Phi$.
For each variable $x_i$ include the triple $(a_{i,0},\,b_{i,0},\,c_{i,0})$ if $x_i$ is true and 
$(a_{i,0},\,b_{i,0},\,c_{i,2})$ if $x_i$ is false.
In either case, it is full with contribution $(n+2)N^2$.

If $x_i$ is true, also match $(a_{i,2},\,b_{i,1},\,c_{i,1})$ and $(a_{i,4},\,b_{i,2},\,c_{i,2})$, leaving $a_{i,1}$ and $a_{i,3}$ unused.
If $x_i$ is false, also match $(a_{i,1},\,b_{i,1},\,c_{i,0})$ and $(a_{i,3},\,b_{i,2},\,c_{i,1})$, leaving $a_{i,2}$ and $a_{i,4}$ unused.
In all four cases the contribution equals $nN^2+mN+3$: the first two coordinates contribute $nN^2+3$ (the minima of the offset-shifted constants), and the last two contribute $mN$.
Thus we obtain $3n$ full triples.

For each clause $C_j$ choose a literal that is true under the assignment. If the chosen literal is the $r$-th positive occurrence of $x_i$ (so $p_{i,r}=j$), match $(a_{i,2r-1},\,b_j,\,c_j)$; if it is the $s$-th negative occurrence (so $q_{i,s}=j$), match $(a_{i,2s},\,b_j,\,c_j)$. By the variable step the $a_{i,\cdot}$ used here was left unused, so no conflicts arise. Moreover,
the triplet is full with contribution $mN$,
since $b_j$ and $c_j$ have $(a,b)$-coordinates $(0,0)$ and $(c,d)$-coordinates $(jN,(m-j)N)$, giving coordinate-wise minima $jN$ and $(m-j)N$ with the corresponding $a_{i,\cdot}$. 

\medskip
($\Leftarrow$)
% \todo[inline]{[PROOF] Justify rigorously that $x_1^i$ must be mapped to $x_2^i$ in any maximum matching of triples, removing the current heuristic assumption.}
Assume there is a family $\mathcal{M}$ of pairwise-disjoint matched triples with:
(i) exactly $n$ triples of contribution $(n+2)N^2$,
(ii) exactly $2n$ triples of contribution $nN^2+mN+3$, and
(iii) for each $j\in[m]$, exactly one triple containing $b_j$ and $c_j$ of contribution $mN$.
We build a satisfying assignment for $\Phi$.

Fix $i\in[n]$. A triple that contains $a_{i,0}$ contributes $(n+2)N^2$ if and only if it is
\[
(a_{i,0},\,b_{i,0},\,c_{i,0}) \quad\text{or}\quad (a_{i,0},\,b_{i,0},\,c_{i,2}).
\]
Indeed, with $b_{i,0}$ and  $c\in\{c_{i,0},c_{i,2}\}$, the first two coordinate-wise minima are $(iN^2,(n-i)N^2)$, summing to $nN^2$, and with the last two minima are $(N^2,N^2)$, summing to $2N^2$, for a total of $(n+2)N^2$. Using $b_{i,1}$, $b_{i,2}$, or $c_{i,1}$ instead would force the last two minima to be $(mN,mN)$, yielding a total contribution smaller than $(n+2)N^2$.

Define the assignment by
\[
x_i=\text{true}\ \text{if}\ (a_{i,0},b_{i,0},c_{i,0})\in\mathcal{M},
\qquad
x_i=\text{false}\ \text{if}\ (a_{i,0},b_{i,0},c_{i,2})\in\mathcal{M}.
\]
We show that this assignment satisfies each clause.

Now consider triples in $\mathcal{M}$ that use $a\in\{a_{i,1},a_{i,2},a_{i,3},a_{i,4}\}$ and have contribution $nN^2+mN+3$ (hence are full) inside the $i$-block.
A direct check of the $(a,b)$-constants shows the only possibilities are
\[
(a_{i,1},b_{i,1},c_{i,0}),\quad
(a_{i,2},b_{i,1},c_{i,1}),\quad
(a_{i,3},b_{i,2},c_{i,1}),\quad
(a_{i,4},b_{i,2},c_{i,2}),
\]
since in each of these the minima in the first two coordinates sum to $nN^2+3$ while those in the last two sum to $mN$, and no other combination attains this total.
Because triples are disjoint, at most one triple using $b_{i,1}$, at most one using $b_{i,2}$, and at most one using $c_{i,1}$ can appear.
As $\mathcal{M}$ contains exactly two full triples of contribution $nN^2 + mN + 3$ in the $i$-block, it must include one of the following pairs:
\begin{align*}
  \{ (a_{i,1},b_{i,1},c_{i,0}),\ (a_{i,3},b_{i,2},c_{i,1}) \} \subseteq \mathcal M\ \text{or}\  
  \{ (a_{i,2},b_{i,1},c_{i,1}),\ (a_{i,4},b_{i,2},c_{i,2}) \} \subseteq \mathcal M.
\end{align*}
% Consequently, the used $a$-indices are either $\{1,3\}$ or $\{2,4\}$.

This choice is determined by which $c_{i,\cdot}$ was used together with $a_{i,0}$ earlier: if the $a_{i,0}$-triple uses $c_{i,0}$ (so $x_i$ is true), then 
$\{ (a_{i,2},b_{i,1},c_{i,1}),\ (a_{i,4},b_{i,2},c_{i,2}) \} \subseteq \mathcal M$, leaving $a_{i,1},a_{i,3}$ unused;
if it uses $c_{i,2}$ (so $x_i$ is false), then $\{ (a_{i,1},b_{i,1},c_{i,0}),\ (a_{i,3},b_{i,2},c_{i,1}) \} \subseteq \mathcal M$ is unavailable, leaving $a_{i,2},a_{i,4}$ unused.

For each $j\in[m]$, $\mathcal{M}$ contains a unique triple that uses $b_j$ and $c_j$.
Since $b_j$ and $c_j$ have $(a,b)$-coordinates $(0,0)$, the first two coordinates contribute $0$, and the last two must sum to $mN$;
hence the chosen $a$ must have $(c,d)$-coordinates $(jN,(m-j)N)$, i.e., it corresponds to a literal that actually appears in $C_j$.
Therefore, $C_j$ is satisfied, and $\Phi$ is satisfiable.
\end{proof}

\subsection{Parameterizing by Vertex Cover Number} \label{ssec:3mcis-vc}

We show that \prob{$\ell$-MCIS} is FPT when parameterized by vertex cover number.
This stands in contrast to the para-NP-hardness by cluster deletion number (\Cref{thm:tmcislb}).

\tmcisvc*
\begin{proof}
We show there are at most $2^{O(k^2)}$ possible common subgraphs (up to isomorphism).
For each candidate $H$, check whether it is an induced subgraph of every $G_i$; this is an ISI instance and can be decided in $O^*(k^{O(k)})$ time for $\vc(G_i)\le k$ using the vertex-cover parameterization algorithm from Abu-Khzam \cite{DBLP:journals/ipl/Abu-Khzam14}.

Without loss of generality, we assume $G_1$ has the smallest number of vertices among input graphs and let $n_1:=|V(G_1)|$.
Then, each $G_i$ has an independent set of size $n_1-k$; the optimal solution is obtained by removing at most $k$ vertices from $G_1$.
Since $\vc(G_1)\leq k$, vertices of $G_1$ is classified into $2^{O(k)}$ classes with the same neighborhood, and thus, $G_1$ has at most $(2^{O(k)})^{k}=2^{O(k^2)}$ different induced subgraphs of size at least $n_1-k$ up to isomorphism.
\end{proof}

}

\section{Conclusions}

We study \textsc{ISI} and \textsc{MCIS} parameterized by $k=\cvd(G_1)+\cvd(G_2)$.
Our main algorithmic result gives a randomized $O^*(k^{O(k)})$-time algorithm for \textsc{ISI} (\Cref{thm:isi-fpt}), and this matches the known ETH lower bound.
For \textsc{MCIS} we give a randomized $O^*(2^{O(k^2)})$-time algorithm, resolving the open question of Hanaka et al.~\cite{HanakaOOV25}, together with a matching ETH lower bound via \textsc{LCBM} (\Cref{thm:mcis-fpt,thm:lcbmlb,thm:mcislb}).
These results slightly separate \textsc{ISI} and \textsc{MCIS} under the cluster vertex deletion number parameterization; by contrast, under vertex cover number both admit $O^*(k^{O(k)})$-time algorithms.

In contrast to \textsc{ISI}, the \textsc{MCIS} algorithm must guess the interaction between the two modulators (equivalently, the structure of $H[J_1\cup J_2]$) and uses a weighted multicolored matching reduction; the tight ETH lower bound suggests that these additional guesses and the $k^2$ exponent are unavoidable.
A natural question is whether \textsc{ISI} and \textsc{MCIS} admit deterministic FPT algorithms. One route is to derandomize the matching subroutine (which however seems difficult): by \Cref{prop:wemmtoem} this would also yield deterministic FPT for \textsc{Exact Matching}.
Another open direction is to resolve the complexity of \textsc{3-MCIS} when all input graphs have cluster deletion number~1.

\bibliography{references}

\end{document}